\documentclass[final,runningheads,envcountsect,envcountsame]{llncs}
\usepackage{ifthen}
\newboolean{arxivversion}

\setboolean{arxivversion}{true}

\usepackage{textcomp}
\usepackage{pgfplots}
\pgfplotsset{width=10cm,compat=1.9}
\usepackage{algorithm}
\usepackage{ifdraft}
\usepackage[noend]{algpseudocode}
\usepackage{mathtools}
\usepackage{enumerate}
\usepackage{cite}
\usepackage{nicefrac}
\usepackage{caption}
\usepackage{subcaption}
\usepackage{microtype}
\usepackage{xcolor,colortbl}
\usepackage{multirow}
\usepackage{ amssymb }
\usepackage{forest}
\usepackage{wrapfig}
\usepackage{ stmaryrd }
\usepackage{paralist}
\forestset{qtree/.style={for tree={parent anchor=south, 
           child anchor=north,align=center,inner sep=0pt}}}
\usepackage{flexisym}
\usepackage{framed}
\usepackage{enumitem}
\usepackage{float}
\usepackage[most]{tcolorbox}

\makeatletter
\RequirePackage[bookmarks,unicode,colorlinks=true]{hyperref}%
\def\@citecolor{blue}%
\def\@urlcolor{blue}%
\def\@linkcolor{blue}%
\def\orcidID#1{\smash{\protect\raisebox{-1.25pt}{\protect\href{http://orcid.org/#1}{\includegraphics{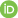}}}}}
\makeatother

\usepackage{cleveref}

\usepackage{booktabs}
\usepackage{ wasysym }
\usepackage{adjustbox}
\usepackage{pgf}

\makeatletter
\newcommand{\otherlabel}[2]{\protected@edef\@currentlabel{#2}\label{#1}}
\makeatother

\usepackage{etoolbox}
\usepackage[colorinlistoftodos,prependcaption,textsize=tiny]{todonotes}
\usepackage{xargs}                      % Use more than one optional parameter in a new commands
\newcommandx{\unsure}[2][1=]{\todo[linecolor=red,backgroundcolor=red!25,bordercolor=red,#1]{#2}}
\newcommandx{\change}[2][1=]{\todo[linecolor=blue,backgroundcolor=blue!25,bordercolor=blue,#1]{#2}}
\newcommandx{\toadd}[2][1=]{\todo[linecolor=pink,backgroundcolor=pink!25,bordercolor=pink,#1]{#2}}

\newcommandx{\sj}[2][1=]{\todo[linecolor=orange,backgroundcolor=orange!25,bordercolor=orange,#1]{SJ:#2}}
\newcommandx{\lk}[2][1=]{\todo[linecolor=green,backgroundcolor=green!25,bordercolor=green,#1]{LK:#2}}

\newcommandx{\jr}[2][1=]{\todo[linecolor=blue,backgroundcolor=blue!25,bordercolor=blue,#1]{JR:#2}}

\usepackage{tikz}

\newcommand{\bestvalue}[1]{%
  \begin{tikzpicture}[baseline=(t.base)]
    \node[outer sep=0pt,inner sep=0pt] (t) {#1};
    \node[overlay
    ,outer sep=0pt,inner sep=1pt,
    rounded corners=1pt,
    fill=yellow!20!white,
    ] (t) {#1};
  \end{tikzpicture}
}

\usetikzlibrary{arrows.meta,
                chains,
                positioning,
                quotes,
                shapes.geometric}

\tikzset{
    recgray/.style={draw,minimum width=2cm,minimum height=2cm,align=center,text width=2cm,fill=green!15!white,font=\sffamily},
    recgray2/.style={draw,minimum width=2cm,minimum height=1.3cm,align=center,text width=2cm,fill=green!15!white,font=\sffamily},
    teacher/.style={draw,minimum width=2cm,minimum height=3.2cm,align=center,text width=10cm,fill=yellow!10!white,font=\sffamily},
    sul/.style={draw,minimum width=2cm,minimum height=4.5cm,align=center,text width=2cm,fill=teal!30!white,font=\sffamily},
    learner/.style={draw,minimum width=2cm,minimum height=4.5cm,align=center,text width=2cm,fill=teal!30!white,font=\sffamily},
}

\usetikzlibrary{automata, positioning, arrows.meta, fit}
\tikzset{
>=stealth, % makes the arrow heads bold
node distance=3cm, % specifies the minimum distance between two nodes. Change if necessary.
every state/.style={thick, fill=gray!10}, % sets the properties for each ’state’ node
initial text=$ $, % sets the text that appears on the start arrow
}
\usetikzlibrary{patterns, intersections, positioning, shapes.geometric, calc,
	automata, arrows, decorations.pathreplacing,decorations.pathmorphing}
\usetikzlibrary{cd}
\tikzset{
	treenode/.style = {align=center, inner sep=0pt, text centered},
  basis/.style = {
    pattern=north east lines,
    pattern color=lightgray,
  },
  newbasis/.style={
    pattern=north west lines,
    pattern color=teal!80!black!60!white,
  },
  other/.style={
    pattern=north east lines,
    pattern color=red,
  },
  frontier/.style={
    pattern=north east lines,
    pattern color=yellow!80!black,
  },
  pw/.style={
    pattern=north east lines,
    pattern color=green!80!black,
  },
  q0class/.style={
    pattern=north west lines,
    pattern color=red!40!white,
  },
  q1class/.style={
    pattern=north west lines,
    pattern color=blue!60!white,
  },
  q2class/.style={
    pattern=crosshatch,
    pattern color=green!50!black!60!white,
  },
  basic/.style = {
    fill=white,
  }
}
\usepackage[outline]{contour}
\contourlength{1.2pt}
\newcommand{\treeNodeLabel}[1]{\contour{white}{#1}}

\usepackage{svg}

\pgfplotsset{compat=1.8}

\pgfplotsset{vasymptote/.style={
    before end axis/.append code={
        \draw[densely dashed] ({rel axis cs:0,0} -| {axis cs:#1,0})
        -- ({rel axis cs:0,1} -| {axis cs:#1,0});
    }
}}

\definecolor{darkred}{rgb}{.75,0,0}

\definecolor{darkgreen}{rgb}{0,.75,0}

\newcommand{\moe}[1]{$\mathsf{MoE}$(#1)}

\newcommand{\partialto}{\ensuremath{\rightharpoonup}}

\newcommand{\takeout}[1]{\relax}

\newcommand{\weights}[0]{\mathsf{weights}}
\newcommand{\probs}[0]{\mathsf{probs}}

\renewcommand{\H}{\mathcal{H}}
\renewcommand{\S}{\mathcal{S}}
\newcommand{\M}{\mathcal{M}}

\newcommand{\bigO}{\mathcal{O}}

\newcommand{\converges}{\ensuremath{\mathord{\downarrow}}}
\newcommand{\diverges}{\ensuremath{\mathord{\uparrow}}}

\newcommand{\ETS}{\mathsf{ETS}}

\usepackage{xspace}
\newcommand{\defineApiFunction}[1]{%
\expandafter\newcommand\csname #1\endcsname{\text{\upshape\textsc{#1}}\xspace}%
}

\newcommand{\activeMM}{\mathsf{active}}

\newcommand{\myparagraph}[1]{\smallskip\noindent \emph{#1}}
\newcommand{\mysubsubsection}[1]{\medskip\noindent \textbf{#1}}

\defineApiFunction{LSharp}
\defineApiFunction{GetAccessSeq}
\defineApiFunction{GetExpert}
\defineApiFunction{ComputeSubalphabet}
\defineApiFunction{GetInfixSeq}
\defineApiFunction{GetSeparatingSeq}

\defineApiFunction{CheckConsistency}
\defineApiFunction{BuildHypothesis}
\defineApiFunction{EquivalenceQuery}
\defineApiFunction{ProcessCounterexample}
\defineApiFunction{OutputQuery}
\defineApiFunction{Ads}
\defineApiFunction{TestHypothesisMAB}
\defineApiFunction{GenerateTestCaseWithExpert}
\defineApiFunction{UpdateProbabilities}
\defineApiFunction{UpdateWeights}
\defineApiFunction{Learn}
\defineApiFunction{MAT}
\defineApiFunction{MABEQ}

\algdef{SE}[MYWHILE]{MyWhile}{MyWhile}{\algorithmicwhile\#1}%

\usepackage[marginclue,footnote]{fixme}

\definecolor{ballblue}{rgb}{0.13, 0.67, 0.8}

\algdef{SE}[SUBALG]{Indent}{EndIndent}{}{\algorithmicend\ }%
\algtext*{Indent}
\algtext*{EndIndent}

\tikzset{
  do guard/.style={
    inner xsep=0pt,
  },
  guard line offset/.style={
    xshift=2mm,
  },
  bigtalloblong/.style={
    draw=black,
    minimum width=3pt,
    minimum height=1em,
    inner xsep=0pt,
  },
}
\newcommand{\connectDoGuards}[2]{%
  \draw[overlay,draw=black!40!white] ([guard line offset]#1) -- ([guard line offset]#2);
}

\algtext*{EndWhile}%

\algblockdefx{DoIf}{EndDoIf}
[1]{\begin{tikzpicture}[remember picture,baseline=(do.base)]
    \node[do guard] (do) {\textbf{do}};
    \coordinate (last do) at (do.south west);
  \end{tikzpicture}
     #1 $\rightarrow$}
   {\begin{tikzpicture}[remember picture,baseline=(od.base)]
       \node[do guard] (od) {X};
       \connectDoGuards{opening do}{od.base west}
    \end{tikzpicture}\textbf{od}}
\algcblockdefx[DoIf]{DoIf}{ElsDoIf}{EndDoIf}
[1]{\begin{tikzpicture}[remember picture,baseline=(new do.base)]
    \node[do guard] (new do) {\phantom{\bfseries d}};
    \draw[draw=none] ([yshift=3pt]new do.north);
    \connectDoGuards{last do}{new do.north west}
    \coordinate (last do) at (new do.south west);
    \end{tikzpicture}
  #1 $\rightarrow$}
{\begin{tikzpicture}[remember picture,baseline=(od.base)]
       \node[do guard] (od) {\textbf{end do}};
       \connectDoGuards{last do}{od.north west}
    \end{tikzpicture}%
  }

\algloopdefx{IfLoop}[1]{\textbf{If} #1 \textbf{then}}

\begin{document}

\title{Small Test Suites for Active Automata Learning\thanks{%
This research is partially supported by the NWO grant No.~VI.Vidi.223.096.}}

\titlerunning{Small Test Suites for Active Automata Learning}
% If the paper title is too long for the running head, you can set
% an abbreviated paper title here
%
% In the final version, the orcids will be replaced with the orcid logo
% \author{Loes Kruger, seb, jur}
\author{Loes Kruger \orcidID{0009-0003-3275-6806} 
\and  Sebastian Junges \orcidID{0000-0003-0978-8466} 
\and Jurriaan Rot \orcidID{0000-0002-1404-6232}
  }
\authorrunning{L. Kruger et al.}
% First names are abbreviated in the running head.
% If there are more than two authors, 'et al.' is used.
%
\institute{
 Institute for Computing and Information Sciences, \\
 Radboud University, Nijmegen, the Netherlands\\
 \email{\{loes.kruger,sebastian.junges,jurriaan.rot\}@ru.nl}
}
\maketitle              % typeset the header of the contribution

\begin{abstract}
A bottleneck in modern active automata learning is to test whether a hypothesized Mealy machine correctly describes the system under learning. 
The search space for possible counterexamples is given by so-called test suites, consisting of input sequences that have to be checked to decide whether a counterexample exists. 
This paper shows that significantly smaller test suites suffice under reasonable assumptions on the structure of the black box. 
These smaller test suites help to refute false hypotheses during active automata learning, even when the assumptions do not hold.
We combine multiple test suites using a multi-armed bandit setup that adaptively selects a test suite.
An extensive empirical evaluation shows the efficacy of our approach. For small to medium-sized models, the performance gain is limited. However, the approach allows learning models from large, industrial case studies that were beyond the reach of known methods. 

\end{abstract}

%\end{document}

\section{Introduction} \label{sec:introduction}

%% AAL IS RELEVANT
System identification algorithms aim to capture the behavior of a black-box system, often called the \emph{system under learning} (SUL), in a formal model.
Among the system identification approaches, \emph{active automata learning} (AAL)~\cite{AngluinRegularSets1987,IsbernerThesis2015,HowarS2018} is a popular methodology to extract finite automata from a black-box. AAL has been successfully applied to learn security-critical protocol implementations~\cite{FiterauTCP2016, FiterauSSH2017, RuiterTLS2015}, 
legacy code~\cite{AslamCSB20,SHV16}, smart cards~\cite{ChaluparPPR14}, interfaces of data structures~\cite{HowarISBJ12}, embedded control software~\cite{SmeenkPrinter2015}, and (explainable) neural network policies~\cite{DBLP:conf/icml/WeissGY18}.

%% FROM AAL TO TEST SUITES.
%% AAL REQUIRES EQUIVALENCE QUERIES
Modern AAL methods~\cite{IsbernerTTT2014,VaandragerLsharp2022} are available via mature tool sets~\cite{IsbernerHS15,BolligKKLNP10,MuskardinAPPT21} that implement these methods. They are primarily built around Angluin's Minimal Adequate Teacher (MAT) framework~\cite{AngluinRegularSets1987}.
In essence, the theoretically elegant MAT framework requires access to two types of queries. 
First, an \emph{output query} (OQ) allows to execute a sequence of inputs on the black-box and observe its outputs. Second, an \emph{equivalence query} (EQ) asks whether a hypothesized Mealy machine is indeed equivalent to the SUL. Implementing the equivalence query provides practitioners with an impossibility~\cite{MooreGedanken1956}: \emph{How do we decide whether a learned model is equivalent to the behavior of a black-box?}  
%% TEST SUITES ARE A STANDARD APPROACH
To overcome this impossibility, practitioners take a more modest stance and only \emph{approximate} equivalence queries.
%We differentiate between two commonly used approaches.
One approach is to randomly sample from all possible input sequences, which leads to a statistical guarantee\footnote{Typically, a probably-approximately correct (PAC) guarantee.}, as pioneered in the context of learning by Angluin~\cite{AngluinConceptLearning1987}. 
Alternatively, based on ideas pioneered in~\cite{ChowTesting1978, VasilevskiiFailure1973}, the structure of the hypothesis is used to select a finite set of input sequences to be checked. These sets are \emph{test suites} and the approach is called conformance testing~\cite{2004test}. This paper considers EQs via test suites; for an overview, see Sec.~\ref{sec:overview}.

%% TEST SUITES ARE TOO LARGE.
\myparagraph{Challenge: Finding small test suites.}
Finite test suites can be obtained using the notion of \emph{$k$-completeness}. In short, $k$-completeness guarantees equivalence under the assumption that the number of states in the SUL is at most $k$ states larger than the hypothesis. Popular $k$-complete test are the classical W-method~\cite{ChowTesting1978, VasilevskiiFailure1973} and variations thereon, such as Wp~\cite{FujiwaraTesting1991}, HSI~\cite{LuoTesting1994, PetrenkoTesting1993} and Hybrid-ADS~\cite{MoermanThesis2019}; see empirical evaluations in~\cite{AslamThesis2021}. We call these methods \emph{Access-Step-Identify} (ASI). These are standard in tools like LearnLib~\cite{IsbernerHS15} and AALpy~\cite{MuskardinAPPT21}. However, $k$-complete test suites such as the $W$-method grow with $|I|^k$, where $|I|$ is the number of input symbols. Consequently, even for small $k$, these test suites are prohibitively large. 

\myparagraph{Our approach for smaller test suites.}
Towards smaller test suites, we adapt ASI methods and make assumptions on the \emph{shape of the SUL} in relation to
the shape of the hypothesis.
%More precisely, we make assumptions
%; more precisely, on the relation between the shape of the SUL in relation to the shape of the hypothesis. 
In particular, we consider several natural assumptions that may occur in real-world systems. For instance, one of these assumptions is that in most states, most inputs either lead to an error-state or are simply discarded. Other assumptions are that certain inputs are used only in the beginning (e.g.\ in the authentication phase of a protocol), or that the SUL has a component structure where inputs are primarily used together within components.
%We may assume that these transitions are correct, i.e., that the hypothesis and the SUL agree on these transitions. 
We formalize these assumptions, demonstrate the applicability on industrial benchmarks, and develop a notion of completeness under these assumptions. The resulting test suites are much smaller, as the factor $|I|^k$ is restricted to $|I'|^k$, with $|I'| < |I|$. 

%% FAST COUNTEREXAMPLES
\myparagraph{Challenge: finding counterexamples as soon as possible.}
The time to find a counterexample during EQs is the bottleneck in AAL applications~\cite{YangASML2019,VaandragerLearning2017}. To accelerate this process, it is helpful to constrain the search space of possible counterexamples, allowing for a targeted search. Here complete test suites are again helpful, even if they can not be fully executed and can only be approximated through sampling, as implemented for instance in the randomised W-method of LearnLib. Complete test suites then provide a constrained search space that still contains all actual counterexamples. Another relevant aspect for finding counterexamples fast is the order that tests are chosen: an adequate ordering in which counterexamples (empirically) occur early in the test suites is preferred. 

\myparagraph{Our approach for finding counterexamples faster.}
In the context of randomised W-methods, pruning input sequences that are not counterexamples yields a larger probability of sampling a counterexample and thus speeds up the procedure. 
However, for the smaller test suites described earlier, without domain-specific knowledge, we can not be certain that they contain (a larger fraction of) counterexamples, as we do not know whether the underlying assumptions are met. Instead, our idea is to \emph{combine} multiple test suites. We prefer tests from test suites that led to counterexamples in previous invocations of the EQ during the learning process. We operationalize this idea using multi-armed bandits.

\myparagraph{Contributions.}
In summary, this paper introduces three new test suites that are complete under additional assumptions on the SUL (Sec.~\ref{sec:subalphabet}). We combine these test suites via a multi-armed bandit framework to accelerate finding counterexamples in EQs (Sec.~\ref{sec:subalphabetselection}). The paper demonstrates performance on scalable self-generated benchmarks, standard benchmarks and industry benchmarks (Sec.~\ref{sec:experiments}).
\ifthenelse{\boolean{arxivversion}}{}{The proofs of all theorems, the complete benchmarks results and additional figures can be found in the appendix of the extended version of this paper~\cite{Kruger2024SmallTestSuiteArxiv}.}

\section{Overview} \label{sec:overview}

\ifoptionfinal{
\begin{figure}[t]
\begin{adjustbox}{width=\textwidth}
    \centering

    \begin{tikzpicture}[x=3.5cm,y=2cm]
    \node[teacher] (teacher) at (4.5,1.7) {};
    \node[learner] (learner) at (2,2) {\textbf{1. Learner}};
    \node[sul] (sul) at (7,2) {\textbf{2. SUL}};
    \node[recgray2] (1) at (3.5,1.35) {3b. Generate test case};
    \node[recgray] (2) at (5.5,1.85) {3a. Choose test suite \& Execute test case};

    \node[] at (3.35,2.4) {\textbf{3. Teacher}};
    
     \draw
       
        (5.82,2) edge[->,very thick] node[above,midway,inner sep=2pt] {OQ $\sigma$} (6.68,2)
        (6.68,1.6) edge[->,very thick] node[above,midway,inner sep=2pt,xshift=0.2cm] {response $v \in O^*$} (5.82,1.6)
        (6.68,2.7) edge[->,very thick] node[above,midway,inner sep=2pt] {response $v \in O^*$} (2.32,2.7)
        (2.32,3.0) edge[->,very thick] node[above,midway,inner sep=2pt] {OQ $w \in I^*$} (6.68,3.0)
        (2.32,2.1) edge[->,very thick] node[above,pos=0.12,inner sep=2pt] {EQ $\mathcal{H}$} (5.18,2.1)
        (5.18,1.8) edge[->,very thick] node[above,midway,inner sep=2pt,xshift=1cm] {\textbf{if outputs different}, $\sigma$} (2.32,1.8) 
        (5.18,1.5) edge[->,very thick] node[above,pos=0.7,inner sep=2pt] {\textbf{else}} (3.82,1.5)
        ;
    \draw[->,very thick] (3.82,1.1) -| node[above,pos=0.3,inner sep=2pt] {test case $\sigma$} (5.5,1.35);
    \end{tikzpicture}
    \end{adjustbox}
    \caption{Interaction between the learner, teacher and SUL in the MAT framework.}
    \label{fig:overview}
\end{figure}
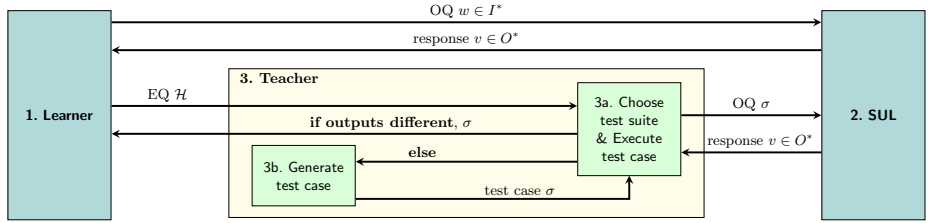}

We briefly illustrate the interactions in the MAT framework, the W-method, and our approach for generating smaller test suites, using a toy example. Recall that in the MAT framework the learner can pose output queries (OQ) and equivalence queries (EQ). This is depicted in Fig.~\ref{fig:overview}, where EQs are implemented by the teacher.
The Mealy machine in Fig.~\ref{fig:overviewtarget} depicts the SUL for a coffee machine with input alphabet $I = \{\textit{coffee}, \textit{espresso}, \textit{tea}, 1 \}$.
Coffee costs $1$~euro, espresso costs $2$~euros, and tea never gets dispensed.
Via a series of queries, we may obtain the hypothesis in Fig.~\ref{fig:overviewH0}.
The hypothesis is easy to refute with an EQ, e.g., via the counterexample $1 \cdot \textit{coffee}$.
After various OQs, we learn the hypothesis in Fig.~\ref{fig:overviewH1}. A short counterexample that distinguishes the hypothesis $\H_1$ from the SUL $\S$ is
\[ \underbrace{\vphantom{\textit{coffee}}1}_{\text{access}} \hspace{0.2cm} \cdot \underbrace{\vphantom{\textit{coffee}} 1 \cdot \textit{coffee}}_{\text{infix}} \cdot \underbrace{\vphantom{\textit{coffee}} \textit{coffee}}_{\text{distinguish}}. \]
The counterexample consists of three parts. We first \emph{access} $q_1$ and $t_1$, from which we run an infix that leads to either $q_1$ or $t_0$, and then we distinguish both states with \textit{coffee}. Executing input \textit{coffee} from $q_1$ returns output \textit{coffee} while executing input \textit{coffee} from $t_0$ returns output $-$.
The W-method generates test suites that consist of input words of a similar shape. Concretely, test suites are constructed as $P\cdot I^{\leq k+1} \cdot W$, where $P$ ensures access to the states in the hypothesis, $I^{\leq k+1}$ is the set of sequences of at most $k+1$ arbitrary input symbols, used to step to states in the (larger) SUL, and $W$ contains sequences that help to distinguish states. Test suites constructed in this way tend to contain many input sequences which do not help to refute the hypothesis. In our example, the W-method test suite with $k=2$ for $\mathcal{H}_1$ also contains uninformative sequences such as
\[ 
\underbrace{\epsilon }_{\text{access}} \hspace{0.2cm} \cdot \underbrace{\vphantom{\textit{coffee}}1 \cdot \textit{espresso} \cdot 1}_{\text{infix}} \cdot \underbrace{ \textit{coffee}}_{\text{distinguish}} \qquad \text{and} \qquad
\underbrace{\vphantom{\textit{coffee}} 1}_{\text{access}} \hspace{0.2cm} \cdot \underbrace{\vphantom{\textit{coffee}} \textit{tea} \cdot \textit{espresso}}_{\text{infix}} \cdot \underbrace{\vphantom{\textit{coffee}} \textit{espresso}}_{\text{distinguish}}.\]

\ifoptionfinal{
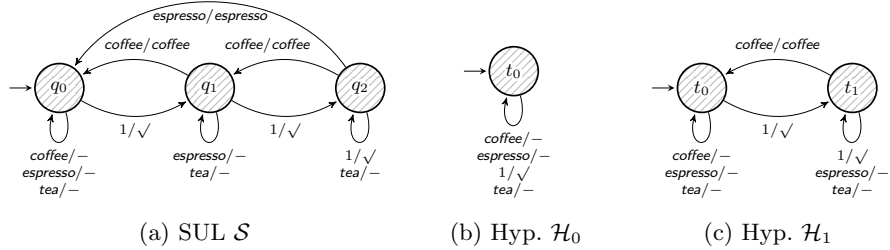
\begin{figure}[t]
    \centering
    \begin{subfigure}[b]{0.45\textwidth}
        \scalebox{0.8}{
        \begin{tikzpicture}[shorten >=1pt,auto,node distance=2.5cm,main node/.style={circle,draw,font=\sffamily\large\bfseries}]
            ]
            \node[initial ,state,basis] (q0) {\treeNodeLabel{$q_0$}};
            \node[state,basis] (q1) [right of=q0] {\treeNodeLabel{$q_1$}};
            \node[state,basis] (q2) [right of=q1] {\treeNodeLabel{$q_2$}};

            \path[->,>=stealth',every node/.style={font=\sffamily\scriptsize}]
            (q0) edge[bend right] node[below] {$1/\surd$} (q1)
                 edge[loop below] node[align=center] {$\textit{coffee}/-$\\$\textit{espresso}/-$\\$\textit{tea}/-$} (q0)
            (q1) edge[bend right] node[below] {$1/\surd$} (q2)
                 edge[bend right] node[above,xshift=0.2cm] {$\textit{coffee}/\textit{coffee}$} (q0)
                 edge[loop below] node[align=center] {$\textit{espresso}/-$\\$\textit{tea}/-$} (q0)
            (q2) edge[loop below] node[align=center] {$1/\surd$\\$\textit{tea}/-$} (q2)
                 edge[bend right] node[above,xshift=-0.3cm] {$\textit{coffee}/\textit{coffee}$} (q1)
                 edge[bend right=55] node[below] {$\textit{espresso}/\textit{espresso}$} (q0)
            ;
    
        \end{tikzpicture}}
        \caption{SUL $\mathcal{S}$}
        \label{fig:overviewtarget}
    \end{subfigure}
    \hfill
    \begin{subfigure}[b]{0.16\textwidth}
        \scalebox{0.8}{
        \begin{tikzpicture}[shorten >=1pt,auto,node distance=2cm,main node/.style={circle,draw,font=\sffamily\large\bfseries},
        ]
        \node[initial,state,basis] (q0) {\treeNodeLabel{$t_0$}};

        \path[->,>=stealth',every node/.style={font=\sffamily\scriptsize}]
        (q0) edge[loop below] node[align=center] {$\textit{coffee}/-$\\
                                            $\textit{espresso}/-$\\
                                            $1/\surd$\\
                                            $\textit{tea}/-$} (q0)
        ;

        \end{tikzpicture}}
        %\vspace{1.2cm}
        \caption{Hyp.\ $\mathcal{H}_0$}
        \label{fig:overviewH0}
        \end{subfigure}
        \hfill
        \begin{subfigure}[b]{0.3\textwidth}
        \scalebox{0.8}{
        \begin{tikzpicture}[shorten >=1pt,auto,node distance=2.5cm,main node/.style={circle,draw,font=\sffamily\large\bfseries},
        ]
            \node[initial,state,basis] (q0) {\treeNodeLabel{$t_0$}};
            \node[state,basis] (q1) [right of=q0] {\treeNodeLabel{$t_1$}};

            \path[->,>=stealth',every node/.style={font=\sffamily\scriptsize}]
            (q0) edge[bend right] node[below] {$1/\surd$} (q1)
                 edge[loop below] node[align=center] {$\textit{coffee}/-$\\$\textit{espresso}/-$\\$\textit{tea}/-$} (q0)
            (q1) edge[bend right] node[above] {$\textit{coffee}/\textit{coffee}$} (q0)
                 edge[loop below] node[align=center] {$1/\surd$\\$\textit{espresso}/-$\\$\textit{tea}/-$} (q0)
            ;

        \end{tikzpicture}}
        \caption{Hyp.\ $\mathcal{H}_1$}
        \label{fig:overviewH1}
    \end{subfigure}
    \caption{A coffee machine and two hypotheses which can be generated using AAL.}
    \label{fig:coffeemachine}
\end{figure}}

\myparagraph{A smaller test suite.}
In hypothesis $\mathcal{H}_1$, \textit{espresso} and \textit{tea} self-loop in all states. The counterexample to refute this hypothesis only requires the inputs \textit{coffee} and 1. It is natural that input \textit{tea} is not necessary to reach new states, as this option is obsolete. 
This leads us to a test suite for $\mathcal{H}_1$ that excludes the inputs \textit{tea} and \textit{espresso} in the infix. If we generate infixes of length at most 3 ($k=2$) with the full alphabet, the test suite contains 112 test cases. If we exclude two inputs, only 12 test cases remain.

\myparagraph{A set of smaller test suites.}
The restricted test suites that aim to exclude obsolete inputs can be refined. These restrictions can be adapted for other typical scenarios. Consider, e.g., network protocols that only perform a three-way handshake in the initial phase. In states where the communication protocol is initialized, these inputs are no longer relevant. Likewise, there are often clusters where the same input symbols are relevant. For instance, if a 10 cent coin is a relevant input in some state of a vending machine, then a 50 cent coin is likely also relevant. 

\myparagraph{Mixing test suites.}
Restricting the test suites yields the risk of missing counterexamples. 
While the test suite may be complete under (natural) additional assumptions, in a black-box setting we have no way to check whether these assumptions hold. We therefore present a methodology where various restricted test suites are combined, using multi-armed bandits to select test suites. During learning, the EQs then increasingly use test suites for which the assumptions hold, without the need for advanced knowledge of the SUL.  

\section{Complete Test Suites} \label{sec:completeness}

We recall complete test suites and start with preliminaries on Mealy machines.
% Mealy machines are similar to deterministic finite automata (DFA) but they have no accepting states and have transitions with inputs and outputs.

\begin{definition}
A \emph{Mealy machine} is a tuple $\M = (Q, I, O, q_0, \delta, \lambda)$ with finite sets $Q$, $I$ and $O$ of \emph{states}, \emph{inputs} and \emph{outputs} respectively; the \emph{initial state} $q_0 \in Q$, the \emph{transition function} $\delta\colon Q \times I \to Q$ and the \emph{output function} $\lambda\colon Q \times I \to O$.

 \end{definition}
Below, we also use \emph{partial} Mealy machines; these are defined as above but with $\delta\colon Q \times I \partialto Q$
and $\lambda\colon Q \times I \partialto O$ partial functions with the same domain. 
For a partial function $f$ we write $f(x)\converges$ if $f(x)$ is defined and $f(x)\diverges$ otherwise. 
The transition and output functions are extended to input words of length $n \in \mathbb{N}$ in the standard way, as functions $\delta\colon Q \times I^n \partialto Q$ and $\lambda\colon Q \times I^n \partialto O^n$. 
We abbreviate $\delta(q_0, w)$ by $\delta(w)$. 
Given $Q' \subseteq Q$ and $L \subseteq I^*$, we write $\Delta^{\M}(Q',L) = \{\delta(q,w) \mid q \in Q', w \in L\}$ for the set of states reached from $Q'$ via words in $L$, and we let $\Delta^{\M}(L) = \Delta^{\M}(\{q_0\}, L)$. In particular $\Delta^{\M}(I^*)$ is the set of reachable states of $\M$. We use the superscript $\M$ to indicate to which Mealy machine we refer, e.g. $Q^{\M}$ and $\delta^{\M}$. We write $|\M|$ for the number of states in $\M$.
A state $q \in Q^{\M}$ is a \emph{sink} if for all $i \in I$, $\delta(q,i)=q$. We denote the set of sinks by $Q_{\text{sink}}$.

\begin{definition} \label{def:eq}
    Given a language $L \subseteq I^*$ and Mealy machines $\H$ and $\S$, states $p \in Q^\H$ and $q \in Q^\S$ are $L$\emph{-equivalent}, written as $p \sim_L q$, if $\lambda^\H(p,w)=\lambda^\S(q,w)$ for all $w \in L$. 
    States $p, q$ are \emph{equivalent}, written $p \sim q$, if they are $I^*$-equivalent. The Mealy machines $\H$ and $\S$ are equivalent, written  $\H \sim \S$, if $q_0^{\H} \sim q_0^{\S}$.
\end{definition} 

Conformance testing techniques construct from a current hypothesis $\H$ a suitable test suite $T \subseteq I^*$, to be executed on the (black-box) SUL $\S$. If a test case fails, we know the machines are inequivalent. Ideally, we want a test suite that contains a failing test case for every possible inequivalent Mealy machine. This is called a \emph{complete} test suite. We define completeness in a more generic way than usual to make it easier to add conditions to the set of Mealy machines for which the test suite is complete in subsequent sections.

\begin{definition}
    Given a Mealy machine $\H$ and set of Mealy machines $\mathcal{C}$, a test suite $T \subseteq I^*$ is \emph{complete} for $\H$ w.r.t. $\mathcal{C}$ if for all $\S \in \mathcal{C}$, $\H \sim_{T} \S$ implies $\H \sim \S$.
\end{definition}
In general, there are no test suites that are complete w.r.t. the (infinite) set $\mathcal{C}$ containing all (inequivalent) Mealy machines~\cite{MooreGedanken1956}. In practice, we often use $k$-completeness, where we assume that $\mathcal{C}$ only contains Mealy machines which have at most $k$ states more than the hypothesis.
\begin{definition}
    Let $\S$ and $\H$ be Mealy machines. A test suite $T \subseteq I^*$ is $k$\emph{-complete} for $\H$ if it is complete w.r.t. $\mathcal{C}^k_{\H} = \{ \S \mid |\S| - |\H|\,\leq k \}$.
\end{definition}
Conformance testing techniques often build $k$-complete test suites in a structured manner using a \emph{state cover} and a \emph{characterization set}.  We give a formal description of a classical $k$-completeness technique: the W-method~\cite{ChowTesting1978, VasilevskiiFailure1973}.

\begin{definition}
An \emph{access sequence} for $q \in Q^{\H}$ is a word $w \in I^*$ such that $\delta^{\H}(w)=q$. A language $P \subseteq I^*$ is a \emph{state cover} if $\varepsilon \in P$ and $P$ contains an access sequence for every reachable state, i.e., 
$\Delta^{\H}(P) = \Delta^{\H}(I^*)$. 
\end{definition}

\begin{definition}
	A \emph{characterization set} for a Mealy machine $\H$ is a non-empty language $W \subseteq I^*$ such that $p \sim_W q$ implies $p \sim q$
	for all $p,q \in Q^\H$.
\end{definition}

% The W-method builds a test suite based on a hypothesis $\H$ and an integer $k$. 
Let $P$ be a minimal state cover and $W$ a characterization set for $\H$. Then the $W$-method, given $k \in \mathbb{N}$, is given by the test suite $T = P \cdot I^{\leq k + 1} \cdot W$. The state cover $P$ makes sure all states in $\H$ are reached. The role of the set of infixes in $I^{\leq k + 1}$ is to reach states in $\S$. The characterization set $W$ checks if the states reached in $\H$ and $\S$ after reading a word from $P \cdot I^{\leq k + 1}$ match. Other ASI methods differ in the computation of the characterization set and the structure of the test suite but are constructed from the same sets $P$, $I$, and $W$.

In the remainder of this section, we prove that the W-method is $k$-complete~\cite{ChowTesting1978, VasilevskiiFailure1973}.
We recall the proof strategy from~\cite{MoermanThesis2019}, based on reachability and bisimulations up-to $\sim_{L}$, in \ifthenelse{\boolean{arxivversion}}{Appendix~\ref{sec:appA}.}{Appendix~A~of~\cite{Kruger2024SmallTestSuiteArxiv}.} With minimal changes, this proof also works for other ASI methods. Here, we summarize the approach in two main steps, which we reuse in Sec.~\ref{sec:subalphabet} to prove completeness for different test suites under additional conditions. The first step concerns reachability in $\S$. We assume that $\H$ is minimal w.r.t.\ number of states, which is an invariant of active learning algorithms, our intended application. This assumption is only used for $k$ to be correct; alternatively, one can bound the number of states of $\S$ to the sum of $k$ with the number of \emph{inequivalent} states in $\H$.
\begin{lemma} \label{lem:reach-new}
    Let $\H$ and $\S$ be Mealy machines with $|\S| - |\H|\,\leq k$ for some integer $k$, and assume $\H$ is minimal. 
    Moreover, let $P$ be a state cover for $\H$ and $W$ a characterization set for $\H$. Finally, let $L = P \cdot I^{\leq k}$ and
    $T = P \cdot W$ and suppose that $\H \sim_T \S$. Then $L$ is a state cover for $\S$.
\end{lemma}
It is in the above lemma that the assumption $\varepsilon \in P$ is used, to ensure that all states in $\S$ are reached from a state in $\Delta^\S(P)$. 
The second step extends this to actual equivalence of the two Mealy machines.
\begin{lemma} \label{lem:bisim-new}
    Suppose $L \subseteq I^*$ is a state cover for both $\H$ and $\S$. Let $W$ be a characterization set for $\H$, and $T = L \cdot I^{\leq 1} \cdot W$. If $\H \sim_T \S$, then $\H \sim \S$.
\end{lemma}
Combining the above two lemmas, we recover $k$-completeness of the $W$-method.
\begin{corollary} \label{cor:Wk}
    The $W$-method is $k$-complete.
\end{corollary}

\section{Complete Test Suites with Subalphabets} \label{sec:subalphabet}

We introduce test suites that are similar to the W-method but have fewer infixes. These test suites are roughly of the form $T = P \cdot I_{sub}^{\leq k + 1} \cdot W$, with different choices for $I_{sub} \subseteq I$. If the subalphabet gets smaller, the test suite size always decreases. If we choose $I$ for $I_{sub}$, we recover the original W-method test suite.

In the following subsections, we provide three new functions, called \emph{experts}, for generating subalphabets. 
These experts are tailored to perform well for certain Mealy machine structures. For each expert, we provide a parameterized family of Mealy machines for which the expert should work well, and we show they are complete under specific assumptions that strengthen those of $k$-completeness. 

The experts can be embedded in any ASI method. For conciseness, we focus on the W-method.
In the definition of expert, the output is a set of subalphabets rather than a single one $I_{sub}$ as described above; this is used in one of the experts.

% Throughout this section we fix a hypothesis $\H$ together with a state cover $P$ and a characterization set $W$. 
\begin{definition}
    An \emph{expert} is a function $E$ which takes as arguments a Mealy machine $\H$ and a word $v \in I^*$, and returns a set of subalphabets $I_1,\ldots,I_n$.
\end{definition}
The embedding in the W-method is as follows.
\begin{definition}
\label{def:ets}
    The expert test suite $\ETS$ for $\H$, expert $E$ and $k \in \mathbb{N}$ is:
    \[ \ETS_{E,k}(\H) = \bigcup_{v \in P} ( v \cdot (\bigcup_{I_{sub} \in E(\H,v)} I_{sub}^{\leq k-1}) \cdot I^{\leq 2} \cdot W) \]
    where $P$ is a minimal state cover for $\H$ and $W$ a characterization set.
\end{definition}
Before introducing the new experts we define the trivial expert.
\begin{definition}
    The \emph{trivial expert} $E_{\mathsf{T}}$ is given by $E_{\mathsf{T}}(\H,q) = \{ I^{\H} \}$.
\end{definition}
If $P$ is a minimal state cover and $W$ a characterization set, then $\ETS_{E_{\mathsf{T}},k}(\H)$ is given by $P \cdot I^{\leq k - 1} \cdot I^{\leq 2} \cdot W$, which is precisely the W-method test suite.

\subsection{Active Inputs Expert}

\ifoptionfinal{
\begin{figure}[t]
    \centering
    \begin{adjustbox}{width=\textwidth}
    \begin{tikzpicture}[shorten >=1pt,auto,node distance=2.2cm,main node/.style={circle,draw,font=\sffamily\large\bfseries},
        ]
        \node[initial,state,basis] (a0) {\treeNodeLabel{$x_1$}};
        \node[] (a) [right of=a0] {\ldots};
        \node[state,basis] (aa) [right of=a] {\treeNodeLabel{$x_a$}};
        \node[state,basis] [right of=aa,yshift=0.8cm] (d0c0) {\treeNodeLabel{$y_{0,0}$}};	
        \node[] [right of=d0c0] (d0d) {\ldots};	
        \node[state,basis] [right of=d0d] (d0cc) {\treeNodeLabel{$y_{0,a}$}};	
        \node[state,basis] [right of=aa,yshift=-0.8cm] (ddc0) {\treeNodeLabel{$y_{a,0}$}};	
        \node[] [right of=ddc0] (ddd) {\ldots};	
        \node[state,basis] [right of=ddd] (ddcc) {\treeNodeLabel{$y_{a,a}$}};	
        \node[] [right of=aa,xshift=2.2cm] (d) {$\vdots$};	
        \node[state, basis] [right of=aa, xshift=6.6cm] (q) {\treeNodeLabel{$z$}};
        % \node[] [above of=q, yshift=1.2cm] (arr1) {};
        % \node[] [above of=a0, yshift=1.2cm] (arr2) {};
        
        \path[->,>=stealth',every node/.style={font=\sffamily\scriptsize}]
        (a0) edge[above] node {$x_1/x_1$} (a)
        (a) edge[above] node {$x_a/x_a$} (aa)
        (aa) edge[above] node {$x_1/x_1$} (d0c0)
        (aa) edge[below ==] node {$x_a/x_a$} (ddc0)

        (d0c0) edge[above] node {$x_2/x_2$} (d0d)
        (d0d) edge[above] node {$x_{a-1}/x_{a-1}$} (d0cc)
        (ddc0) edge[above] node {$x_1/x_1$} (ddd)
        (ddd) edge[above] node {$x_{a-2}/x_{a-2}$} (ddcc)
        (d0cc) edge[below] node {$x_a/x_a$} (q)
        (ddcc) edge[below] node[xshift=0.3cm] {$x_{a-1}/x_{a-1}$} (q)
        (q) edge[looseness=0.5,bend right=55, below] node {$x_1/x_1$} (a0)
        ;

        % \draw (q) edge[-] node{} (11,2.5);
        % \draw (11,2.5) edge[-,above] node{$x_1/x_1$} (.5,2.5);
        % \draw (.5,2.5) edge[->] node{} (a0);
        
        \draw [decorate,decoration={brace,amplitude=10pt,mirror,raise=4pt},yshift=0pt]
        (0,-0.4) -- (4.5,-0.4) node [black,midway,yshift=-0.9cm] {$a$};

        \draw [decorate,decoration={brace,amplitude=10pt,mirror,raise=4pt},yshift=0pt]
        (4.6,-1.5) -- (13.2,-1.5) node [black,midway,yshift=-0.9cm] {$a$};

        \draw [decorate,decoration={brace,amplitude=10pt,raise=4pt},yshift=0pt]
        (13.5,1.5) -- (13.5,-1.5) node [black,midway,xshift=0.5cm] {$a$};

        \end{tikzpicture}
        \end{adjustbox}
    \caption{$\mathsf{ASML}_{a,b}$ models over inputs and outputs $\{ x_i \mid 1 \leq i \leq a \} \cup \{ y_i \mid 1 \leq i \leq b \}$. Transitions not drawn, including all transitions $y_i$ with $1 \leq i \leq b$, lead to a sink with a unique output.}
    \label{fig:asml}
\end{figure}
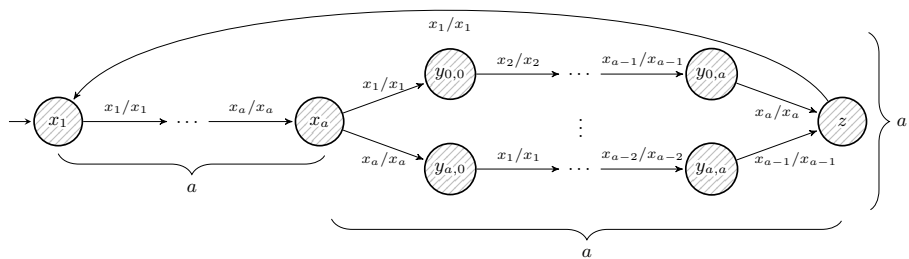
% \sj{Schematic for Mealy machines $\mathsf{S135}_{a,b}$ over input alphabet $\{ x_i \mid 1 \leq i \leq a \} \cup \{ y_i \mid 1 \leq i \leq b \}$ and the same output alphabet. Transitions that are not depicted, including all $y_i$, lead to a sink state with a unique output. !I WOULD PUT THE REST OF THE TEXT IN THE MAIN TEXT! }

% Abstract depiction of the family of Mealy machines $\mathcal{S}_{a,b}$ inspired by the ASML models. This model consists of a spine of length $a$ where each state $x_i$ with $i \in \{1..a\}$ has a transition $x_i/x_i$ to state $x_{i+1}$. Next, we have $a$ branches of length $a$ with a rotation of the order of the spine. From the state $z$, it is possible to reset the system with $x_1/x_1$. Additionally, we have $b$ distinct $y$ inputs that always transition to the sink with output $-$. All non-specified transitions lead to a (not depicted) sink state.
}

\paragraph{Motivation.} Mealy machines with many inputs are challenging, even when most inputs induce no interesting behavior, i.e., when most inputs transition to sinks. This challenge is exemplified by the \emph{ASML models} which were first described in~\cite{YangASML2019} and partially made available for the 2019 RERS challenge~\cite{JasperRERS2019}. The ASML models represent components of lithography systems used at ASML. These models feature many inputs that often lead to a sink state. Model \emph{m135} in particular has approximately 100 inputs that always transition to the sink state with the same output. The Mealy machines $\mathsf{ASML}_{a,b}$ where $a, b \in \mathbb{N}$, displayed in Fig.~\ref{fig:asml}, closely resemble \emph{m135}. The model starts with a spine, then there is a choice between $a$ branches, and the spine inputs are reused in a different order after the choice. There are $b$ distinct inputs that always lead to a sink.

\myparagraph{The expert.}
The \emph{active inputs expert} addresses Mealy machines where there is a significant set of inputs that always lead to the sink state or self-loop. We define the active version of a Mealy machine and then the active inputs expert.

\begin{definition}
 An input $i\in I$ is active in $q \in Q$, if $\delta(q,i) \notin Q_{sink} \text{ and } \delta(q,i) \neq q$.
    The \emph{active Mealy machine} of $\H = (Q,I,O,q_0,\delta,\lambda)$ is the partial Mealy machine $\activeMM(\H) = (Q \setminus Q_{sink},I',O,q_0,\delta',\lambda')$ such that
    \begin{align*}
       I' &= \Biggl\{ i \in I \mid \exists q \in Q.~ i \text{ active in } q\Biggr\},\\  
        \delta'(q,i) &= 
         \begin{cases}
             \delta(q,i) &  \text{if } i \text{ active in }q,\\
             \diverges & \text{otherwise},\\
         \end{cases}\quad\text{ and }\quad
        \lambda'(q,i) = 
        \begin{cases}
            \lambda(q,i) & \text{if } \delta'(q,i)\downarrow,\\
            \diverges &  \text{otherwise}.\\
        \end{cases}
    \end{align*}
\end{definition}

\begin{definition}
The active inputs expert $E_{\mathsf{AI}}$ is given by $E_{\mathsf{AI}}(\H,p) = \{I^{\activeMM(\H)} \}$.
\end{definition}

\myparagraph{Complexity.} The time complexity of $E_{\mathsf{AI}}$ is $\bigO(nk)$, where $n$ is the number of states and $k$ the number of inputs. This is achieved by first determining the set $Q_{sink}$ in $\bigO(nk)$, and then computing $\delta'$ and $I'$ simultaneously in $\bigO(nk)$.

\myparagraph{Completeness.} Test suite $\ETS_{E_{\mathsf{AI}},k}$ is complete for the set of Mealy machines which 1) have at most $k$ additional states and 2) where all non-sink states can be reached by a word in the state cover followed by at most $k$ active inputs.

\begin{theorem} \label{thm:activeinputs}
    Suppose $\ETS_{E_{\mathsf{AI}},k}(\H)$ uses state cover $P$.
    Let $\mathcal{C} = \{\S \in \mathcal{C}^k_{\H} \mid
    Q^{\S} \setminus Q_{sink}^{\S} \subseteq \Delta^\S(P \cdot  (I^{\activeMM(\H)})^{\leq k})\}
    $.
    Then $\ETS_{E_{\mathsf{AI}},k}(\H)$ is complete for $\mathcal{C}$.
\end{theorem}
The proof follows from Lemma~\ref{lem:bisim-new}; the hypotheses make sure that a variant of Lemma~\ref{lem:reach-new} holds.
The above theorem applies in particular, if $\H$ is minimal, for the restriction of $\mathcal{C}^k_{\H}$ to those Mealy machines $\S$ where all non-sink states are reachable by the sub-alphabet generated by $E_{\mathsf{AI}}$.

The active inputs expert performs well on $\mathsf{ASML}_{a,b}$ once the spine is learned because it will not generate infixes with inputs that always lead to the sink state. In the empirical evaluation performed in Sec.~\ref{sec:experiments}, it can be observed that $E_{\mathsf{AI}}$ requires significantly fewer symbols to learn $\mathsf{ASML}_{a,b}$ compared to $E_{\mathsf{T}}$.

\subsection{Future Expert}

\myparagraph{Motivation.} Real-world systems often contain an `initialization phase' where inputs like $start$ or $login$ are used that are not used later in the system.  
Fig.~\ref{fig:tcp} shows the family of Mealy machines $\mathsf{TCP}_{a,b}$ inspired by the TCP models~\cite{FiterauTCP2016}.
The models of TCP clients contain two distinct phases: the three-way handshake and the connected part. After the three-way handshake, some inputs are never active again. $\mathsf{TCP}_{a,b}$ has the same two phases. For the last few hypotheses that arise during learning, all inputs will be active. Therefore, $E_{\mathsf{AI}}$ will generate the same $\ETS$ as $E_{\mathsf{T}}$. $E_{\mathsf{AI}}$ is too coarse here because, at different parts of the system, different sets of inputs are active.

% To avoid generating infix sequences with irrelevant inputs, we only consider inputs that transition to reachable states in the active version of the hypothesis.

\ifoptionfinal{
\begin{figure}[t]
    \centering
    \begin{adjustbox}{width=\textwidth}
    \begin{tikzpicture}[shorten >=1pt,auto,node distance=2.2cm,main node/.style={circle,draw,font=\sffamily\large\bfseries},
        ]
        \node[initial,state,basis] (a0) {\treeNodeLabel{$y_1$}};
        \node[] (a) [right of=a0] {\ldots};
        \node[state,basis] (aa) [right of=a] {\treeNodeLabel{$y_a$}};
        \node[state,basis] [right of=aa,yshift=0.8cm] (d0c0) {\treeNodeLabel{$x_{0,0}$}};	
        \node[] [right of=d0c0] (d0d) {\ldots};	
        \node[state,basis] [right of=d0d] (d0cc) {\treeNodeLabel{$x_{0,a}$}};	
        \node[state,basis] [right of=aa,yshift=-0.8cm] (ddc0) {\treeNodeLabel{$x_{a,0}$}};	
        \node[] [right of=ddc0] (ddd) {\ldots};	
        \node[state,basis] [right of=ddd] (ddcc) {\treeNodeLabel{$x_{a,a}$}};	
        \node[] [right of=aa,xshift=2cm] (d) {$\vdots$};	
        \node[state, basis] [right of=aa, xshift=6.6cm] (q) {\treeNodeLabel{$z$}};
        % \node[] [above of=q, yshift=1.2cm] (arr1) {};
        % \node[] [above of=a0, yshift=1.2cm] (arr2) {};
        
        \path[->,>=stealth',every node/.style={font=\sffamily\scriptsize}]
        (a0) edge[above] node {$y_1/y_1$} (a)
        (a) edge[above] node {$y_b/y_b$} (aa)
        (aa) edge[below] node {$x_1/x_1$} (d0c0)
        (aa) edge[below] node {$x_a/x_a$} (ddc0)

        (d0c0) edge[above] node {$x_2/x_2$} (d0d)
        (d0d) edge[above] node {$x_{a-1}/x_{a-1}$} (d0cc)
        (ddc0) edge[above] node {$x_1/x_1$} (ddd)
        (ddd) edge[above] node {$x_{a-2}/x_{a-2}$} (ddcc)
        (d0cc) edge[below] node {$x_a/x_a$} (q)
        (ddcc) edge[below] node[xshift=0.3cm] {$x_{a-1}/x_{a-1}$} (q)
        (q) edge[looseness=0.7,bend right=50, below] node {$x_0/x_0$} (aa)
        ;

        % \draw (q) edge[-] node{} (11,2.5);
        % \draw (11,2.5) edge[-,above] node{$x_1/x_1$} (4.5,2.5);
        % \draw (4.5,2.5) edge[->] node{} (aa);
        
        \draw [decorate,decoration={brace,amplitude=10pt,mirror,raise=4pt},yshift=0pt]
        (0,-0.5) -- (4,-0.5) node [black,midway,yshift=-0.9cm] {$b$};

        \draw [decorate,decoration={brace,amplitude=10pt,mirror,raise=4pt},yshift=0pt]
        (4.6,-1.5) -- (13.2,-1.5) node [black,midway,yshift=-0.9cm] {$a$};

        % \draw [decorate,decoration={brace,amplitude=10pt,mirror,raise=4pt},yshift=0pt]
        % (6.2,-1.2) -- (11.5,-1.2) node [black,midway,yshift=-0.9cm] {$a$};

        \draw [decorate,decoration={brace,amplitude=10pt,raise=4pt},yshift=0pt]
        (13.6,1.5) -- (13.6,-1.5) node [black,midway,xshift=0.5cm] {$a$};

        \end{tikzpicture}
        \end{adjustbox}
    \caption{$\mathsf{TCP}_{a,b}$ models over inputs and outputs $\{ x_i \mid 1 \leq i \leq a \} \cup \{ y_i \mid 1 \leq i \leq b \}$. Transitions not shown lead to a sink  with a unique output.}
    \label{fig:tcp}
\end{figure}
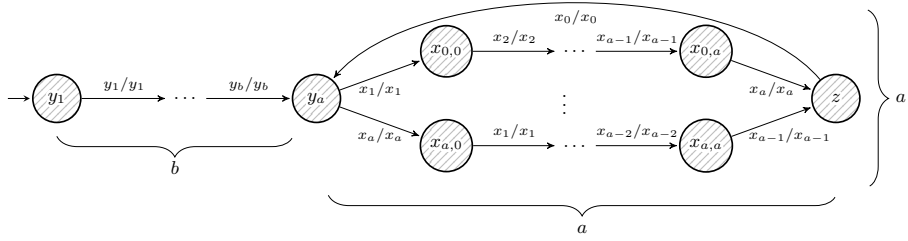

% Abstract depiction of the family of Mealy machines $\mathcal{S}'_{a,b}$ inspired by the TCP models. This model consists of a spine of length $b$ where each state $y_i$ with $i \in \{0..b-1\}$ has a transition $y_i/y_i$ to state $y_{i+1}$. Next, we have $a$ branches of length $a$ with a rotation of the order of the first branch. From the state $z$, it is possible to reset the system with $x_1/x_1$ to state $y_a$. All non-specified transitions lead to the (not depicted) sink state with output $-$. 
}

\myparagraph{The expert.}
The \emph{future expert} generates a subalphabet for each state in the hypothesis. This subalphabet contains all inputs that are active from that state onwards, within a given number of steps.
Bounding the number can be useful in large models, and avoids that we end up with the complete alphabet if the Mealy machine is strongly connected. 

\begin{definition}
The future expert $E^l_{\mathsf{F}}$, is given for $l \in \mathbb{N}$ by $E^l_{\mathsf{F}}(\H,v) = \{ I_{v,l} \}$ where $I_{v,l} = \{ i \mid \exists q \in \Delta^{\activeMM(\H)}(v \cdot I^{\leq {l-1}}) \land \delta^{\activeMM(\H)}(q,i)\downarrow \}$.
\end{definition}

\myparagraph{Complexity.} The time complexity $\bigO(n(n + n|I|))$ can be achieved for $E_{\mathsf{F}}$ with a bounded BFS for each state.

\myparagraph{Completeness.} For $E^l_{\mathsf{F}}$, we have the following completeness result.

\begin{theorem} \label{thm:future}
    Suppose $\ETS_{E^l_{\mathsf{F}},k}(\H)$ uses state cover $P$.
    Let $\mathcal{C} = \{\S \in \mathcal{C}^k_{\H} \mid Q^\S \setminus Q^\S_{sink} \subseteq \bigcup_{v\in P} \Delta^\S(v \cdot I_{v, l}^{\leq k})\}$.
    Then $\ETS_{E^l_{\mathsf{F}},k}(\H)$ is complete for $\mathcal{C}$.  
\end{theorem}
$E_{\mathsf{F}}$ performs well on $\mathsf{TCP}_{a,b}$ once the spine is learned because the subalphabet for states after $y_a$ does not contain $y$-symbols, contrary to subalphabet from $E_{\mathsf{T}}$. Sec.~\ref{sec:experiments} shows that $E^l_{\mathsf{F}}$ often outperforms the trivial expert $E_{\mathsf{T}}$.

\subsection{Components Expert}

\myparagraph{Motivation.}
In some systems, sets of inputs are often used together. For example, after entering a username you often enter a password as well. It is possible that the set of inputs that are used together occur at multiple places in the system. 
Fig.~\ref{fig:ssh} shows Mealy machines $\mathsf{SSH}_{a,b}$, loosely inspired by OpenSSH~\cite{FiterauSSH2017}. The OpenSSH model contains three phases: the key exchange, the authentication, and then the connection phase where re-keying is possible. For the family of Mealy machines $\mathsf{SSH}_{a,b}$, we assume there is a fixed set of possible keys and the key exchange and re-keying uses the same key-specific inputs, i.e., the inputs  for the key exchange of key $k$ are the relevant inputs  for re-keying with key $k$.

\myparagraph{The expert.}
The component expert generates subalphabets based on sets of states and is defined as follows.\footnote{$E^l_{\mathsf{F}}$ can be seen as a refined form of the $E^g_{\mathsf{C}}$ which returns the subalphabets $I_X$ with $X$ consisting of the states reachable in at most $l$ steps from state $p$.}

\ifoptionfinal{
\begin{figure}[t]
    \centering
    \begin{adjustbox}{width=\textwidth}
    \begin{tikzpicture}[shorten >=1pt,auto,node distance=1.6cm,main node/.style={circle,draw,font=\sffamily\large\bfseries},
        ]
        \node[initial,state,basis] (a0) {\treeNodeLabel{$x_0$}};
        
        \node[state,basis] [right of=a0,yshift=1cm] (a11) {\treeNodeLabel{$x_{1,1}$}};
        \node[state,basis] [right of=a11] (a12) {\treeNodeLabel{$x_{1,2}$}};
        \node[state,basis] [right of=a12] (a13) {\treeNodeLabel{$x_{1,3}$}};
        \node[state,basis] [right of=a13,yshift=-1cm] (a4) {\treeNodeLabel{$x_4$}};

        \node[state,basis] [right of=a0,yshift=-1cm] (aa1) {\treeNodeLabel{$x_{a,1}$}};
        \node[state,basis] [right of=aa1] (aa2) {\treeNodeLabel{$x_{1,2}$}};
        \node[state,basis] [right of=aa2] (aa3) {\treeNodeLabel{$x_{1,3}$}};

        \node[state,basis] [right of=a4] (bs) {\treeNodeLabel{$y$}};	
        \node[state,basis] [below left of=bs,yshift=-0.5cm] (b0) {\treeNodeLabel{$y_1$}};	
        \node[] [below of=bs,yshift=-0.5cm] (b) {\ldots};	
        \node[state,basis] [below right of=bs,yshift=-0.5cm] (bb) {\treeNodeLabel{$y_b$}};

        \node[state,basis] [right of=bs] (c0) {\treeNodeLabel{$z_0$}};
        
        \node[state,basis] [right of=c0,yshift=1cm] (c11) {\treeNodeLabel{$z_{1,1}$}};
        \node[state,basis] [right of=c11] (c12) {\treeNodeLabel{$z_{1,2}$}};
        \node[state,basis] [right of=c12] (c13) {\treeNodeLabel{$z_{1,3}$}};
        \node[state,basis] [right of=c13,yshift=-1cm] (c4) {\treeNodeLabel{$z_4$}};

        \node[state,basis] [right of=c0,yshift=-1cm] (ca1) {\treeNodeLabel{$z_{a,1}$}};
        \node[state,basis] [right of=ca1] (ca2) {\treeNodeLabel{$z_{1,2}$}};
        \node[state,basis] [right of=ca2] (ca3) {\treeNodeLabel{$z_{1,3}$}};

        \node[] [right of=a0,xshift=1.6cm] (d) {$\vdots$};	
        \node[] [right of=c0,xshift=1.6cm] (d) {$\vdots$};	
        
        \path[->,>=stealth',every node/.style={font=\sffamily\scriptsize}]

        (a0) edge[above] node[xshift=-0.2cm,yshift=0.2cm] {$x_{1,1}/x_1$} (a11)
        (a0) edge[above] node[xshift=-0.2cm,yshift=-0.6cm] {$x_{a,1}/x_{a}$} (aa1)
        (a11) edge[bend left] node {$x_{1,2}/x_1$} (a12)
        (aa1) edge[bend left] node {$x_{a,2}/x_{a}$} (aa2)
        (a12) edge[bend left] node {$x_{1,1}/x_1$} (a13)
        (aa2) edge[bend left] node {$x_{a,1}/x_{a}$} (aa3)
        (a13) edge[above] node[xshift=0.4cm] {$x_{1,2}/x_1$} (a4)
        (aa3) edge[above] node[xshift=-0.4cm] {$x_{a,2}/x_{a}$} (a4)
        (a12) edge[bend left] node {$x_{1,2}/x_1$} (a11)
        (aa2) edge[bend left] node {$x_{a,2}/x_{a}$} (aa1)
        (a13) edge[bend left] node {$x_{1,1}/x_1$} (a12)
        (aa3) edge[bend left] node {$x_{a,1}/x_{a}$} (aa2)

        (c0) edge[above] node[xshift=-0.2cm,yshift=0.2cm] {$x_{1,1}/x_1$} (c11)
        (c0) edge[above] node[xshift=0.2cm,yshift=0.2cm] {$x_{a,1}/x_{a}$} (ca1)
        (c11) edge[bend left] node {$x_{1,2}/x_1$} (c12)
        (ca1) edge[bend left] node {$x_{a,2}/x_{a}$} (ca2)
        (c12) edge[bend left] node {$x_{1,1}/x_1$} (c13)
        (ca2) edge[bend left] node {$x_{a,1}/x_{a}$} (ca3)
        (c13) edge[above] node[xshift=0.4cm] {$x_{1,2}/x_1$} (c4)
        (ca3) edge[above] node[xshift=-0.4cm] {$x_{a,2}/x_{a}$} (c4)
        (c12) edge[bend left] node {$x_{1,2}/x_1$} (c11)
        (ca2) edge[bend left] node {$x_{a,2}/x_{a}$} (ca1)
        (c13) edge[bend left] node {$x_{1,1}/x_1$} (c12)
        (ca3) edge[bend left] node {$x_{a,1}/x_{a}$} (ca2)

        (bs) edge[] node[left] {$y_1/y_{fail}$} (b0)
        (b0) edge[loop below] node[] {$y_1/y_{fail}$} (b0)
        (bs) edge[] node[right] {$y_b/y_{fail}$} (bb)
        (bb) edge[loop below] node[] {$y_b/y_{fail}$} (bb)
        
        (a4) edge[above] node {$y/y$} (bs)
        (bs) edge[above] node {$y/y$} (c0)
        ;

        % \draw (q) edge[-] node{} (11,2.5);
        % \draw (11,2.5) edge[-,above] node{$x_1/x_1$} (.5,2.5);
        % \draw (.5,2.5) edge[->] node{} (a0);
        
        % \draw [decorate,decoration={brace,amplitude=10pt,mirror,raise=4pt},yshift=0pt]
        % (-0.8,2) -- (-0.8,-2) node [black,midway,xshift=-1cm] {$a$};

        \draw [decorate,decoration={brace,amplitude=10pt,raise=4pt},yshift=0pt]
        (16.5,2) -- (16.5,-2) node [black,midway,xshift=0.5cm] {$a$};

        \draw [decorate,decoration={brace,amplitude=10pt,mirror,raise=4pt},yshift=0pt]
        (6,-2.8) -- (10,-2.8) node [black,midway,yshift=-1cm] {$b$};

        % \draw [decorate,decoration={brace,amplitude=10pt,raise=4pt},yshift=0pt]
        % (11.8,1.7) -- (11.8,-1.7) node [black,midway,xshift=0.5cm] {$a$};

        \end{tikzpicture}
        \end{adjustbox}
    \caption{$\mathsf{SSH}_{a,b}$ models over inputs $\{ x_{i,j} \mid 1 \leq i \leq a, j = 1, 2 \} \cup \{ y_i \mid 1 \leq i \leq b \} \cup \{y\}$ and the outputs
    $\{ x_i \mid 1 \leq i \leq a \} \cup \{y, y_{fail}\}$. Transitions not shown lead to a sink with a unique output.}
    \label{fig:ssh}
\end{figure}
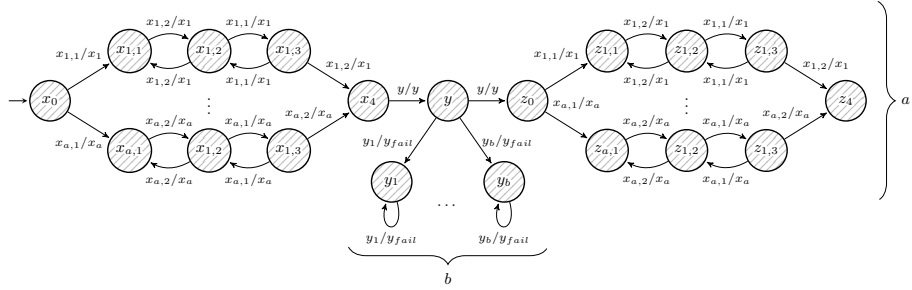

% Abstract depiction of the family of Mealy machines $\mathcal{S''}_{a,b}$ inspired by the SSH models. This model consists of $a$ components which use inputs $x_{i,1}$ and $x_{i,1}$ where $i \in \{0..a\}$. Next, there is one state $y$ which has transitions with $y_i$ with $i \in \{0..b\}$ to states $y_i$. After state $y$, there are again $a$ components with the same behaviour. All non-specified transitions lead to a (not depicted) sink state. 
}

\begin{definition}
    Let $g$ be a function that takes a Mealy machine $\H$ and returns a set of subsets of $Q$, referred to below as \emph{components}. The \emph{component expert} $E^g_{\mathsf{C}}$ with parameter $g$ is defined s.t. $E_C^g(\H, p) = \{ I_X \mid X \in g(\H) \}$ where $I_X = \{ i \mid \exists q,q' \in X. \delta(q,i)=q' \}$.
\end{definition}

\ifoptionfinal{
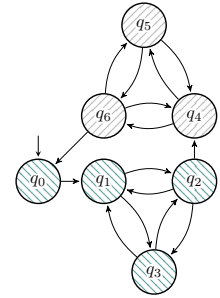
\begin{wrapfigure}{r}{0.22\textwidth}
    % \centering
    \vspace{-0.1cm}
    \begin{center}
    \begin{adjustbox}{width=0.22\textwidth}
    \begin{tikzpicture}[shorten >=1pt,auto,node distance=1.8cm,main node/.style={circle,draw,font=\sffamily\large\bfseries},
        ]
        \node[initial above,state,newbasis] (q0) {\treeNodeLabel{$q_0$}};
        \node[state,newbasis] (q1) [right of=q0,xshift=-0.5cm] {\treeNodeLabel{$q_1$}};
        \node[state,newbasis] (q2) [right of=q1] {\treeNodeLabel{$q_2$}};
        \node[state,newbasis] (q3) [below of=q1, xshift=1cm] {\treeNodeLabel{$q_3$}};
        \node[state,basis] (q4) [above of=q2,yshift=-0.5cm] {\treeNodeLabel{$q_4$}};
        \node[state,basis] (q5) [above of=q4, xshift=-1cm] {\treeNodeLabel{$q_5$}};
        \node[state,basis] (q6) [left of=q4] {\treeNodeLabel{$q_6$}};
        
        \path[->,>=stealth',every node/.style={font=\sffamily\scriptsize}]
        (q0) edge node[above]  {} (q1)
        (q1) edge[bend left=20] node[above] {} (q2)
        (q1) edge[bend left=20] node[above] {} (q3)
        (q2) edge[bend left=20] node[above] {} (q1)
        (q2) edge[bend left=20] node[above] {} (q3)
        (q2) edge[] node[above] {} (q4)
        (q3) edge[bend left=20] node[above] {} (q1)
        (q3) edge[bend left=20] node[above] {} (q2)
        
        (q4) edge[bend left=15] node[above] {} (q5)
        (q4) edge[bend left=20] node[above] {} (q6)
        (q5) edge[bend left=20] node[above] {} (q4)
        (q5) edge[bend left=20] node[above] {} (q6)
        (q6) edge[bend left=20] node[above] {} (q4)
        (q6) edge[bend left=20] node[above] {} (q5)
        (q6) edge[] node[above] {} (q0)
        ;

        \end{tikzpicture}
        \end{adjustbox}
    \caption{Example with colored communities.}
    \label{fig:community}
    \end{center}
    \vspace{-0.8cm}
\end{wrapfigure}} 

\paragraph{Finding components.}
Finding a suitable subroutine $g$ to determine components from a hypothesis is a non-trivial task. One relatively easy method for finding components is to compute the strongly connected components (SCCs). However, if the system can be reset at any state, then the complete model is an SCC and the components expert reduces to the trivial expert. Therefore, SCCs are often too strict. Another possibility is to utilize algorithms used in graph theory to decompose graphs into subgraphs.
We propose to use Newman's algorithm for detecting community structure~\cite{NewmanCommunity2004} to identify components. The algorithm outputs sets of states with high transition density between states within the group. It starts with singleton communities and then greedily joins communities based on the maximal change in modularity, as long as it is positive. The modularity value $\text{mod}(c)$ for component $c$ is:
\[ \text{mod}(c) = \frac{\#\text{edges staying in c}}{\#\text{edges}} - \frac{\#\text{outgoing edges of c} \cdot \#\text{incoming edges of c}}{\#\text{edges}^2} \]

\begin{example} We illustrate Newman's algorithm on  Fig.~\ref{fig:community}.
Initially, $\text{mod}(\{q_1\}) = 0 - \frac{2\cdot 3}{15^2} \approx -0.027, \text{mod}(\{q_3\}) = 0 - \frac{2\cdot 2}{15^2} \approx -0.018$. The difference between the initial modularity and the modularity of $\{q_1,q_3\}$ ($2 - \frac{4\cdot 5}{15^2} = 0.0444$) is % $0.0888$, this is 
the highest possible change in modularity. We thus merge communities $\{ q_1 \}$ and $\{ q_3 \}$. Likewise, we then merge $\{q_1, q_3\}$ and $\{q_2\}$. After several steps we get to the final communities $\{q_0, q_1, q_2, q_3\}$ and $\{q_4,q_5,q_6\}$. 
\end{example}

To apply Newman's algorithm, the subroutine $g$ transforms $\activeMM(\H)$ to a directed graph $G = (Q,E)$ where $E = \{ (q,q') \mid q,q' \in Q \land \exists i \in I. \delta(q,i)=q' \}$ and then applies Newman's algorithm on $G$. 

\myparagraph{Complexity.} The time complexity of $E_C^g$ is $\bigO(g + nk)$ where $g$ is the complexity of the subroutine. The subterm $\bigO(nk)$ originates from the active transformation. With Newman's algorithm, the total complexity is in  $\bigO(n(n + n|I|))$ \cite{NewmanCommunity2004}.

\paragraph{Completeness.}
$\ETS_{E^{g}_{\mathsf{C}},k}$ is $k$-complete if all non-sink states in the SUL can be reached from a state $p$ in the hypothesis with at most $k$ inputs from some $I_X$. 

\begin{theorem} \label{thm:components}
    Suppose $\ETS_{E_C^g,k}(\H)$ uses state cover $P$. 
    Let $\mathcal{C} = \{\S \in \mathcal{C}^k_{\H} \mid
    Q^{\S} \setminus Q_{sink}^{\S} \subseteq \bigcup_{X \in g(\H)} \Delta^\S(P \cdot I_X^{\leq k})\}$.
    Then
    $\ETS_{E_C^g,k}(\H)$ is complete for $\mathcal{C}$.
\end{theorem}

$E^{\text{Newman}}_{\mathsf{C}}$ performs well on $\mathsf{SSH}_{a,b}$ once the key exchange and authentication phase have been learned because the subalphabet mostly contains symbols that belong together and allows discovery of a whole new key exchange component. Ideally, $\{x_{i,1}, x_{i,2}, x_{i,3}\}$, $\{z_{i_1}, z_{i_2}, z_{i,3}\}$ for $1 \leq i \leq a$ and $\{y, y_0, ..., y_b\}$ form components for $\mathsf{SSH}_{a,b}$. In our experiments, Newman's algorithm sometimes finds slightly bigger components. 
\section{Test Case Prioritization} \label{sec:subalphabetselection}
%Explain that from a set we have to make a list of test cases
To establish equivalence, \emph{all} tests in a complete test suite need to be executed and their order is then irrelevant. However, to find a counterexample, we only need to execute tests until we hit that counterexample. This means that different orderings lead to significant performance changes~\cite{AslamThesis2021}. In this section, we first describe the state-of-the-art in (ordered) test suites. We then create new, ordered test suites, that combine the $\ETS$'s from Sec.~\ref{sec:subalphabet} adaptively.

\subsection{Randomised Test Suites} \label{sec:5.1}
Test suites are often stored in a tree-like data structure. The straightforward ordering iterates over this tree to process the test cases deterministically.
However, a variety of deterministic orderings for $P$, $I$, $W$ are all (on average) outperformed by randomised methods that do a better job in diversification~\cite[Ch.\ 4]{AslamThesis2021}.
 State-of-the-art randomised test suite generation methods are described in~\cite{SmeenkPrinter2015,MoermanThesis2019} and make use of a geometric distribution to determine the length of the infix. We present a simpler\footnote{The randomised Hybrid-ADS method in~\cite{SmeenkPrinter2015} first exhausts $P \cdot I^2 \cdot W$ and then generates test cases from $P \cdot I^2 \cdot I^* \cdot W$ where the length of the infix is described by a geometric distribution.} and more generic variation: 
 Given an expert $e$ and a distribution $\mu$ over natural numbers, the \emph{randomised $\ETS$} $S_{e,\mu}$ is a distribution over words $v \cdot i\cdot w \in P \cdot I^* \cdot W$  such that:
\begin{equation}\label{eq:ets_distribution}
 S_{e,\mu}(v \cdot i \cdot w) = \frac{\mu(l)}{|\ETS_{e,l}|} \qquad \text{for } |i|=l
\end{equation}
Informally, \eqref{eq:ets_distribution} indicates that the probability of sampling a test case with infix length $l$ is the probability of sampling infix length $l$ from distribution $\mu$ and then uniformly sampling a test case from $\ETS_{e,l}$. 

For any $\mu$ with infinite support, the generated test suite is infinite. Thus, randomised ASI methods are test case prioritizations over infinite test suites $P \cdot I^* \cdot W$. Still, randomised ASI methods often find counterexamples faster than $k$-complete ASI methods \cite{AichernigBenchmarking2020, GarhewalExperimental2023}. 
To ensure $k$-completeness in randomised ASI methods, we need extra bookkeeping to determine whether the right tests have been executed and we can only guarantee that we execute these tests in the limit.

\subsection{Multi-Armed Bandits}
% new experts are subsets of test suite
We want to use all experts from Sec.\ref{sec:subalphabet} to generate test cases. A naive solution is to determine a static distribution that describes how often an expert should be selected for generating a test case. However, it is unclear how such a distribution should be determined. Instead, we use so-called multi-armed bandits to dynamically update the distribution over available experts using information from previous testing rounds. We refer to this algorithm as the Mixture of Experts. The multi-armed bandits problem was first described by Robbins~\cite{RobbinsMAB1952} and is a classic reinforcement learning problem. We instantiate the $\mathsf{EXP3}$ algorithm for adversarial multi-armed bandits~\cite{AuerMAB2002}. 
Intuitively, our instantiation prioritizes test cases by better performing experts. We embed the Mixture of Expert algorithm in the MAT framework from Fig.~\ref{fig:overview} and list the pseudocode in Algorithm~\ref{alg:mab}.

\begin{algorithm}[t]
\caption{Instantiated $\mathsf{EXP3}$ Algorithm for Test Case Generation } \label{alg:mab}
\begin{algorithmic}[1]
\Procedure{MAB\_EQ}{$\mathcal{H}$, $\weights$}
\While{true}
\State $\probs \xleftarrow[]{}$ \textsc{UpdateProbs}($\weights$)  \Comment{Eq. \ref{eq:updateprobs}}
\State $e \xleftarrow{} $ sample expert proportional to $\probs$
\State $\sigma \xleftarrow{}$ sample next test case from $S_{{e,\mu}}$ \Comment{Eq. \ref{eq:ets_distribution}}
\State $v \xleftarrow[]{} \lambda^{\mathcal{H}}(q_0^{\mathcal{H}},\sigma)$
\State $v' \xleftarrow[]{}$ \textsc{OutputQuery}($\sigma$)
\State $\weights \xleftarrow[]{}$ \textsc{UpdateWeights}($\probs, \weights, e, v \neq v'$) \Comment{Eq. \ref{eq:updateweights}}
\State \textbf{if} {$v \neq v'$} \textbf{return} Some($\sigma$), $\weights$
\EndWhile
\EndProcedure
\end{algorithmic}
\end{algorithm}

Algorithm~\ref{alg:mab} is used with randomised ASI-methods. The algorithm is called with a hypothesis $\H$ and $\weights$. The parameter $\weights$ indicates how good an expert is and is initialized to 1 for each enabled expert. The algorithm uses the set of enabled experts $E$, constant $k$, distribution $\mu$, and exploration parameter $\gamma$ as global parameters. The exploration parameter determines how often we choose an expert at random. 
In Algorithm~\ref{alg:mab}, each iteration of the loop represents the generation of one test case. In each iteration, we first update the distribution $\probs$ for each expert $i \in E$ using Eq.~\ref{eq:updateprobs}.
\begin{equation}\label{eq:updateprobs}
    \probs(i) \gets (1-\gamma) \cdot \frac{\weights(i)}{\Sigma_{j \in E} \weights(j)} + \frac{\gamma}{|E|} 
\end{equation} 
Next, we sample an expert from $\probs$ and sample a test case from $S_{e,\mu}$. We determine the output of the test case on $\H$ and $\S$ and update the $\weights$ for chosen expert $e$ using Eq.~\ref{eq:updateweights} if $v\neq v'$, otherwise the $\weights$ remain the same.
\begin{equation}\label{eq:updateweights}
\weights(e) \gets 
    \weights(e) \cdot \mathit{exp}\left( {\frac{\gamma}{\probs(i) \cdot |E|}}\right)
\end{equation}
If $v \neq v'$, then we have found a counterexample and the $\weights$ value for the chosen expert significantly increases. Consequently, the expert is more likely to be chosen to generate test cases in the next rounds. Finally, if $v \neq v'$, we return the counterexample. Otherwise, generate a new test case.

\section{Experimental Evaluation} \label{sec:experiments}
In this section, we empirically investigate the performance of our implementation of~\Cref{alg:mab} in comparison with a state-of-the-art baseline. The source code and all benchmarks are available online\footnote{\url{https://gitlab.science.ru.nl/sws/lsharp/-/tree/testingstrategy}}~\cite{KrugerZenodo}. We investigate the performance on four benchmark sets with varying complexities in the first three experiments:

\begin{description}[beginpenalty=99,topsep=1pt]
    \item[RQ1:] How does \Cref{alg:mab} scale on the models from Figs.~\ref{fig:asml},~\ref{fig:tcp}, and~\ref{fig:ssh}?
    \item[RQ2:] How does \Cref{alg:mab} compare to the state-of-the-art on industrial benchmarks from the RERS challenge \cite{JasperRERS2019}?
    \item[RQ3:] How does \Cref{alg:mab} perform on the standard automata wiki~\cite{NeiderBenchmarks2018} benchmark suite and randomly generated Mealy machines? 
\end{description}
In Experiment~3, we additionally consider an alternative non-randomised version of the presented algorithm which is not feasible to apply to the benchmarks of Experiment~2 given the worse performance of non-randomised test suites.
Experiment~4 provides an in-depth analysis of runs on two benchmarks from the RERS challenge. Detailed benchmark results can be found in \ifthenelse{\boolean{arxivversion}}{Appendix~\ref{sec:appC}.}{Appendix~C~of~\cite{Kruger2024SmallTestSuiteArxiv}.}

\mysubsubsection{Experimental Setup}
We have extended the $L^{\#}$ learning library~\cite{VaandragerLsharp2022} 
with the multi-armed bandits approach described in Sec.~\ref{sec:subalphabetselection}. We compare our implementation instantiated with different experts. 
We write \moe{$*$} to refer to our key contribution, using the Mixture of all Experts, i.e., \moe{$E_{\mathsf{T}}, E_{\mathsf{AI}}, E^k_{\mathsf{F}}, E^{\text{Newman}}_{\mathsf{C}}$}. The exploration parameter $\gamma$ used in Algorithm~\ref{alg:mab} is set to 0.2 (determined by grid search) and the number of hypothesis states before we start sampling experts to 5.
We evaluate within a MAT framework as in~\Cref{fig:overview}. Our contributions can be paired with any learning algorithm in the MAT framework. We use $L^{\#}$~\cite{VaandragerLsharp2022}, as this is a recent learning algorithm.
We sample test cases from $S_{E_{\mathsf{T}},\mu}$ as our \emph{baseline}. More precisely, we use randomised Hybrid-ADS, as formulated in~\cite[Ch.\ 1]{MoermanThesis2019}, as conformance testing technique. 
For both the baseline and \cref{alg:mab}, the $\mu$ in Eq.~\eqref{eq:ets_distribution} is instantiated as follows:
Let $\textsf{geom}$ be the geometric distribution with mean $2$, then randomised Hybrid-ADS generates $S_{e,\mu}$ as in Eq.~\eqref{eq:ets_distribution}, where $\mu(x) =  \textsf{geom}(x)$ if $x > 3$,
$\mu(3) = \nicefrac{7}{8}$,
and $\mu(x) = 0$ otherwise. 
These hyperparameters are chosen to match~\cite{GarhewalExperimental2023}. 
We run Experiments~1 and~3 with 30 seeds, and Experiments~2 and~4 with 50~seeds.
In Experiments~1,~3 and~4 we evaluate the performance based on the total number of symbols and resets which is the sum of the length of all test cases plus the number of test cases sent to the SUT. Additional plots based on only the symbols or only the resets can be found in \ifthenelse{\boolean{arxivversion}}{Appendix~\ref{sec:appD}.}{Appendix~D~of~\cite{Kruger2024SmallTestSuiteArxiv}.}

\mysubsubsection{Experiment 1}
We evaluate the performance on the benchmark families $\mathsf{ASML}_{a,b}$, $\mathsf{TCP}_{a,b}$, and $\mathsf{SSH}_{a,b}$, for several choices of $a$ and $b$.
In all models, increasing $a$ leads to a general increase in difficulty, while $b$ adds the number of `irrelevant' inputs. 
Beyond the baseline and \moe{$*$}, we include for each family the associated experts discussed in Sec.~\ref{sec:subalphabet}, to validate that they indeed perform well on these families. Thus, for $\mathsf{ASML}_{a,b}$ we run \moe{$E_{\mathsf{T}}, E_{\mathsf{AI}}$}, for $\mathsf{TCP}_{a,b}$ we run \moe{$E_{\mathsf{T}}, E^k_{F}$}, and for $\mathsf{SSH}_{a,b}$ we run \moe{$E_{\mathsf{T}}, E^{\text{Newman}}_{\mathsf{C}}$}.

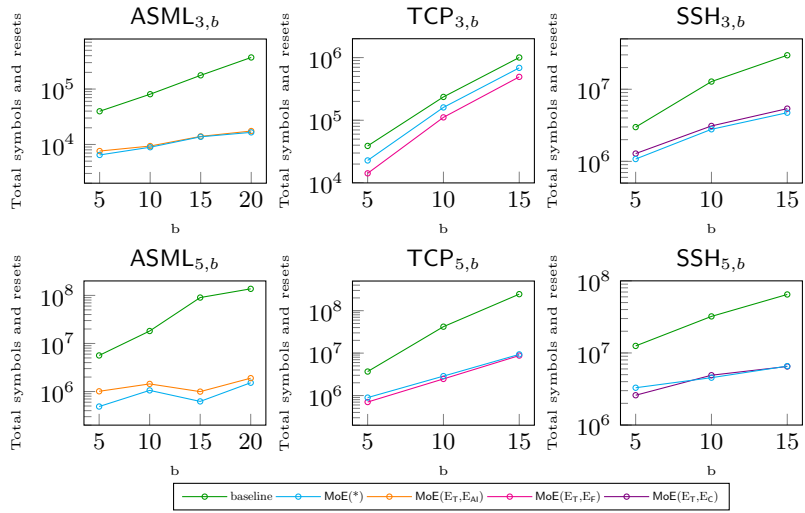
\begin{figure}[t] %finished results, 30 seeds
    \centering
    \begin{tikzpicture}

\begin{axis}[
    legend style={nodes={scale=0.5, transform shape}},
    title style={yshift=-1.5ex,},
    title={$\mathsf{ASML}_{3,b}$},
    width=4cm,height=3.5cm,
    ymode=log,
    xlabel=\tiny{b},
    ylabel=\tiny{Total symbols and resets},
    xtick={5, 10, 15, 20},
    ytick={1000, 10000, 100000, 1000000},
    xtick pos=bottom,
    ytick pos=left,
    ymin=2000, ymax=800000,
    name=plot1,
]

\addplot[
    color=green!60!black,
    mark=o,
    mark size=1pt
    ]
    coordinates {
        (5,39728.63333333333)(10,80729.0)(15,176924.33333333334)(20,371881.73333333334)
    };

\addplot[
    color=orange,
    mark=o,
    mark size=1pt
    ]
    coordinates {
        (5,7607.633333333333)(10,9365.066666666668)(15,13991.2)(20,17474.533333333333)
    };

\addplot[
    color=cyan,
    mark=o,
    mark size=1pt
    ]
    coordinates {
        (5,6466.8)(10,8917.033333333333)(15,13755.333333333334)(20,16536.4)
    };

\end{axis}

\begin{axis}[
    at=(plot1.right of south east), anchor=left of south west,
    legend style={nodes={scale=0.5, transform shape}},
    title style={yshift=-1.5ex,},
    title={$\mathsf{TCP}_{3,b}$},
    width=4cm,height=3.5cm,
    ymode=log,
    xlabel=\tiny{b},
    ylabel=\tiny{Total symbols and resets},
    xtick={5, 10, 15},
    ytick={10000, 100000, 1000000, 10000000},
    xtick pos=bottom,
    ytick pos=left,
    ymin=10000, ymax=2000000,
    legend pos=north west,
    name=plot2,
    ]

    \addplot[
    color=green!60!black,
    mark=o,
    mark size=1pt
    ]
    coordinates {
        (5,38949.066666666666)(10,235650.26666666666)(15,1003336.6666666666)
    };

    \addplot[
        color=magenta,
        mark=o,
        mark size=1pt
        ]
        coordinates {
            (5,14218.266666666666)(10,111168.36666666667)(15,492372.6)
        };

    \addplot[
        color=cyan,
        mark=o,
        mark size=1pt
        ]
        coordinates {
            (5,22835.0)(10,160213.16666666666)(15,685838.0333333333)
        };

\end{axis}

\begin{axis}[
    at=(plot2.right of south east), anchor=left of south west,
    legend style={nodes={scale=0.5, transform shape}},
    title style={yshift=-1.5ex,},
    title={$\mathsf{SSH}_{3,b}$},
    width=4cm,height=3.5cm,
    ymode=log,
    xlabel=\tiny{b},
    ylabel=\tiny{Total symbols and resets},
    xtick={5, 10, 15},
    ytick={100000, 1000000, 10000000, 100000000},
    xtick pos=bottom,
    ytick pos=left,
    ymin=500000, ymax=50000000,
    legend pos=north west,
    name=plot3,
    ]

    \addplot[
    color=green!60!black,
    mark=o,
    mark size=1pt
    ]
    coordinates {
        (5,2985299.9)(10,12790314.566666666)(15,29631939.9)
    };

    \addplot[
        color=violet,
        mark=o,
        mark size=1pt
        ]
        coordinates {
            (5,1283096.9333333333)(10,3100288.3666666667)(15,5387210.366666666)
        };

    \addplot[
        color=cyan,
        mark=o,
        mark size=1pt
        ]
        coordinates {
            (5,1077119.4)(10,2784301.5)(15,4730388.3)
        };

\end{axis}

\begin{axis}[
    at=(plot3.below south west), anchor=above north west,
    legend style={nodes={scale=0.5, transform shape}},
    title style={yshift=-1.5ex,},
    title={$\mathsf{SSH}_{5,b}$},
    width=4cm,height=3.5cm,
    ymode=log,
    xlabel=\tiny{b},
    ylabel=\tiny{Total symbols and resets},
    xtick={5, 10, 15},
    ytick={100000, 1000000, 10000000, 100000000},
    xtick pos=bottom,
    ytick pos=left,
    ymin=1000000, ymax=100000000,
    legend pos=north west,
    name=plot4,
    ]

    \addplot[
    color=green!60!black,
    mark=o,
    mark size=1pt
    ]
    coordinates {
        (5,12517277.833333334)(10,32132599.833333332)(15,64748140.4)
    };

    \addplot[
        color=violet,
        mark=o,
        mark size=1pt
        ]
        coordinates {
            (5,2590051.8)(10,4910185.2)(15,6490321.233333333)
        };

    \addplot[
        color=cyan,
        mark=o,
        mark size=1pt
        ]
        coordinates {
            (5,3290519.2)(10,4539635.933333334)(15,6562893.066666666)
        };

\end{axis}

\begin{axis}[
    at=(plot4.left of south west), anchor=right of south east,
    legend style={nodes={scale=0.5, transform shape}},
    title style={yshift=-1.5ex,},
    title={$\mathsf{TCP}_{5,b}$},
    width=4cm,height=3.5cm,
    ymode=log,
    xlabel=\tiny{b},
    ylabel=\tiny{Total symbols and resets},
    xtick={5, 10, 15},
    ytick={100000, 1000000, 10000000, 100000000, 1000000000},
    xtick pos=bottom,
    ytick pos=left,
    ymin=200000, ymax=500000000,
    legend pos=north west,
    name=plot5,
    ]

    \addplot[
    color=green!60!black,
    mark=o,
    mark size=1pt
    ]
    coordinates {
        (5,3678440.5)(10,41988918.93333333)(15,245136860.16666666)
    };

    \addplot[
        color=magenta,
        mark=o,
        mark size=1pt
        ]
        coordinates {
            (5,703832.5)(10,2465377.6666666665)(15,8755705.566666666)
        };

    \addplot[
        color=cyan,
        mark=o,
        mark size=1pt
        ]
        coordinates {
            (5,895228.4)(10,2871385.1666666665)(15,9333639.9)
        };

\end{axis}

\begin{axis}[
    at=(plot5.left of south west),
    shift={(-71pt,0)},
    legend style={legend columns=-1, nodes={scale=0.5, transform shape}, at={(3.55, -0.4)}},
    title style={yshift=-1.5ex,},
    title={$\mathsf{ASML}_{5,b}$},
    width=4cm,height=3.5cm,
    ymode=log,
    xlabel=\tiny{b},
    ylabel=\tiny{Total symbols and resets},
    xtick={5, 10, 15, 20},
    ytick={100000, 1000000, 10000000, 100000000},
    xtick pos=bottom,
    ytick pos=left,
    ymin=200000, ymax=200000000,
    name=plot6,
    ]

    \addplot[
    color=green!60!black,
    mark=o,
    mark size=1pt
    ]
    coordinates {
        (5,5630923.6)(10,18160950.933333334)(15,90199497.23333333)(20,136953725.5)
    };
    \addlegendentry{baseline}

    \addplot[
    color=cyan,
    mark=o,
    mark size=1pt
    ]
    coordinates {
        (5,487304.13333333336)(10,1056658.0666666667)(15,622854.8666666667)(20,1523979.5666666667)
    };
    \addlegendentry{\moe{*}}

    \addplot[
        color=orange,
        mark=o,
        mark size=1pt
        ]
        coordinates {
            (5,1011102.6666666666)(10,1440401.0666666667)(15,995535.7)(20,1897652.5666666667)
        };
    \addlegendentry{\moe{E$_{\mathsf{T}}$,E$_{\mathsf{AI}}$}}

    \addplot[
        color=magenta,
        mark=o,
        mark size=1pt
        ]
        coordinates {
            (2,2)(3,3)
        };
    \addlegendentry{\moe{E$_{\mathsf{T}}$,E$_{\mathsf{F}}$}}

    \addplot[
        color=violet,
        mark=o,
        mark size=1pt
        ]
        coordinates {
            (2,2)(3,3)
        };
    \addlegendentry{\moe{E$_{\mathsf{T}}$,E$_{\mathsf{C}}$}}

\end{axis}

\end{tikzpicture}
    \caption{Results Experiment 1.}
    \label{fig:res1}
\end{figure}

\myparagraph{Results.}
Fig.~\ref{fig:res1} plots the results, distinguishing six cases. Each column reflects another benchmark family. The top row shows the values for the parameterized models with $a=3$, while the bottom row shows the values for the parameterized models with $a=5$.
In each figure, the x-axis reflects the value of $b$. The y-axis (log scale) shows the total number of symbols and resets to learn and test a model. The y-axis is different for all subplots.  

\myparagraph{Discussion.}
From the plot, we observe that the baseline is outperformed by the other algorithms. Interestingly, the performance of \moe{$*$} and the algorithm belonging to the parameterized model is often comparable. Increasing $a$ leads to an increase in the total number of symbols and resets, which illustrates the scalability of the parameterized models. Increasing the value $b$ has more influence on the baseline than the other algorithms, as expected.

\mysubsubsection{Experiment 2} We compare \moe{$*$} to the baseline on the ASML benchmarks introduced in the RERS challenge~\cite{JasperRERS2019}. We consider 23 models with 25-289 states and 10-177 inputs. We skip models with less than 15 states because $\mathsf{MoE}$ needs time to learn which expert works best. The ASML models frequently do not terminate within a timeout of an hour~\cite{YangASML2019}. Therefore, we set a maximal symbol budget. The SUL rejects new OQs once the budget is depleted.

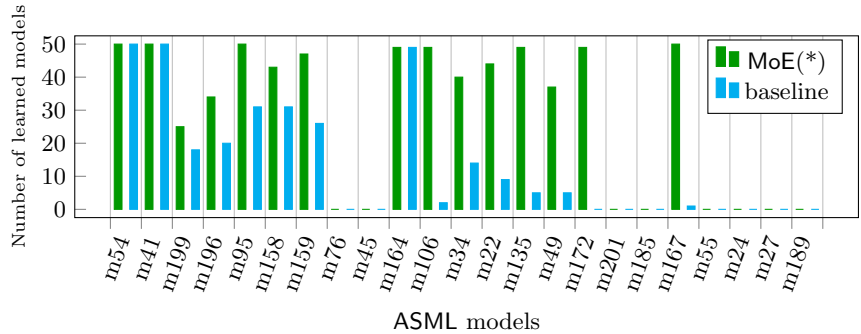
\begin{figure}[t] %finished results, 30 seeds
    \centering
    \begin{tikzpicture}
\begin{axis}[
    width=12.0cm,height=4.0cm,
    xticklabels={m54, m41, m199, m196, m95, m158, m159, m76, m45, m164, m106, m34, m22, m135, m49, m172, m201, m185, m167, m55, m24, m27, m189},
    ytick={0,10,20,30,40,50},
    xtick pos=bottom,
    ytick pos=left,
    enlargelimits=0.05,
	xlabel=$\mathsf{ASML}$ models,
    ylabel=\scriptsize{Number of learned models},
    x tick label style={rotate=70,anchor=east},
	ybar interval=0.5,
]
\addplot[color=green!60!black,fill=green!60!black] 
	coordinates {(0,50)(1,50)(2,25)(3,34)(4,50)(5,43)(6,47)(7,0)(8,0)(9,49)(10,49)(11,40)(12,44)(13,49)(14,37)(15,49)(16,0)(17,0)(18,50)(19,0)(20,0)(21,0)(22,0)(23,0)};
\addplot[color=cyan,fill=cyan]
	coordinates {(0,50)(1,50)(2,18)(3,20)(4,31)(5,31)(6,26)(7,0)(8,0)(9,49)(10,2)(11,14)(12,9)(13,5)(14,5)(15,0)(16,0)(17,0)(18,1)(19,0)(20,0)(21,0)(22,0)(23,0)};
\legend{\moe{*}, baseline}
\end{axis}
\end{tikzpicture}
    \caption{Results Experiment 2. }
    \label{fig:res2}
\end{figure}

\myparagraph{Results.}
Fig.~\ref{fig:res2} lists different models sorted by the number of transitions. For each model, we show how often out of 50 seeds an algorithm learns the model within a symbol budget of $10^8$. We provide a similiar figure with half the budget \ifthenelse{\boolean{arxivversion}}{in Appendix~\ref{sec:appD}}{in Appendix~D~of~\cite{Kruger2024SmallTestSuiteArxiv}}.

\myparagraph{Discussion.}
From the plot, we observe that \moe{$*$} learns the model more often than the baseline. The \moe{$*$} algorithm can learn 12 models with at least $80\%$ of the seeds while the baseline only learns 3 models with at least $80\%$ of the seeds. The same pattern can be observed for half the budget.

\mysubsubsection{Experiment 3} 
We consider the protocol implementations used in~\cite{VaandragerLsharp2022,FiterauDTLS2020} (38 models, 15-133 states, 7-22 inputs) and randomly generated models (27 models, 20-60 states, 11-31 inputs).
 For the standard benchmarks, we perform the experiment with the randomised $\ETS$, as used in the other experiments, and the deterministically ordered $\ETS$ from Sec.~\ref{sec:5.1} with $k=2$.

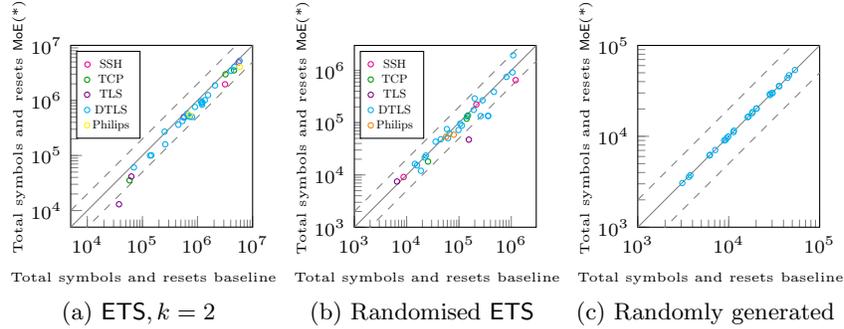
\begin{figure}[t]
    \centering
    \begin{subfigure}[b]{0.3\textwidth}
        \begin{tikzpicture}
    \begin{axis}[
        legend style={nodes={scale=0.5, transform shape}},
        width=4cm,height=4cm,
        xmode=log,
        ymode=log,
        x label style={at={(axis description cs:0.4,-0.2)},anchor=north},
        y label style={at={(axis description cs:-0.2,.55)},anchor=south},
        xlabel=\tiny{Total symbols and resets baseline},
        ylabel=\tiny{Total symbols and resets \moe{*}},
        xtick={10000, 100000, 1000000, 10000000},
        ytick={10000, 100000, 1000000, 10000000},
        every x tick label/.append style={font=\scriptsize},
        every y tick label/.append style={font=\scriptsize},
        xtick pos=bottom,
        ytick pos=left,
        xmin=5000, xmax=10000000,
        ymin=5000, ymax=10000000,
        legend pos=north west,
    ]

    \addplot[
        only marks,
        color=magenta,
        mark=o,
        mark size=1pt
        ]
        coordinates {
            (3159542.566666667,1964511.0)(563191.5333333333,498650.56666666665)(5573300.533333333,4932625.2)
        };
        \addlegendentry{SSH}

    \addplot[
        only marks,
        color=green!60!black,
        mark=o,
        mark size=1pt
        ]
        coordinates {
            (4597336.066666666,3516907.7333333334)(58438.166666666664,35363.86666666667)(3224711.6333333333,2982922.533333333)
        };
        \addlegendentry{TCP}

    \addplot[
        only marks,
        color=violet,
        mark=o,
        mark size=1pt
        ]
        coordinates {
            (62998.36666666667,41525.46666666667)(37648.53333333333,12998.233333333334)
        };
    \addlegendentry{TLS}

    \addplot[
        only marks,
        color=cyan,
        mark=o,
        mark size=1pt
        ]
        coordinates {
            (534266.9666666667,422220.06666666665)(1154151.4333333333,966842.4333333333)(145653.16666666666,100836.8)(252877.26666666666,270951.0)(572966.9,486506.13333333336)(70162.4,60712.0)(1187491.2666666666,891233.9333333333)(1554596.9666666666,1248682.3)(795638.7666666667,501185.4)(4060110.933333333,3444971.9)(4694579.233333333,4085035.566666667)(5809570.5,5254003.233333333)(659778.3333333334,526530.8666666667)(449322.13333333336,361758.76666666666)(2080336.4666666666,1871946.8333333333)(901325.9333333333,760143.9666666667)(1217805.1,926991.3666666667)(1217805.1,926991.3666666667)(1358487.1,1008384.7)(725724.3666666667,503509.9666666667)(1214254.0333333334,836982.0)(137408.9,99754.2)(260236.63333333333,159422.0)(646872.9333333333,571200.9666666667)
        };
    \addlegendentry{DTLS}

    \addplot[
        only marks,
        color=yellow,
        mark=o,
        mark size=1pt
        ]
        coordinates {
            (721852.3333333334,548674.5)(5898493.466666667,4068602.7333333334)
        };
    \addlegendentry{Philips}

    \addplot[color=gray]
        coordinates {(5000,5000)(10000000,10000000) };

    \addplot[color=gray,style=dashed]
        coordinates {(5000,5000/2)(10000000,10000000/2) };

    \addplot[color=gray,style=dashed]
        coordinates {(5000,5000*2)(10000000,10000000*2) };

    \end{axis}
    \end{tikzpicture}
        \vspace{-0.5cm}
        \caption{$\mathsf{ETS}, k=2$}
        \label{fig:res3a}
    \end{subfigure}
    \begin{subfigure}[b]{0.3\textwidth}
        \begin{tikzpicture}
    \begin{axis}[
        legend style={nodes={scale=0.5, transform shape}},
        width=4cm,height=4cm,
        xmode=log,
        ymode=log,
        x label style={at={(axis description cs:0.4,-0.2)},anchor=north},
        y label style={at={(axis description cs:-0.2,.55)},anchor=south},
        xlabel=\tiny{Total symbols and resets baseline},
        ylabel=\tiny{Total symbols and resets \moe{*}},
        every x tick label/.append style={font=\scriptsize},
        every y tick label/.append style={font=\scriptsize},
        xtick pos=bottom,
        ytick pos=left,
        xmin=1000, xmax=3000000,
        ymin=1000, ymax=3000000,
        legend pos=north west,
    ]

    \addplot[
        only marks,
        color=magenta,
        mark=o,
        mark size=1pt
        ]
        coordinates {
            (214764.8,222264.3)(8738.8,9141.866666666667)(1211962.5333333334,651492.4333333333)
        };
        \addlegendentry{SSH}

    \addplot[
        only marks,
        color=green!60!black,
        mark=o,
        mark size=1pt
        ]
        coordinates {
            (148088.3,136556.66666666666)(25644.3,18017.366666666665)(140042.03333333333,118705.66666666667)
        };
        \addlegendentry{TCP}

    \addplot[
        only marks,
        color=violet,
        mark=o,
        mark size=1pt
        ]
        coordinates {
            (154357.36666666667,47337.933333333334)(6550.666666666667,7488.033333333334)
        };
    \addlegendentry{TLS}

    \addplot[
        only marks,
        color=cyan,
        mark=o,
        mark size=1pt
        ]
        coordinates {
            (23303.766666666666,23435.533333333333)(44832.066666666666,48068.666666666664)(15783.766666666666,15232.5)(65901.03333333334,62596.26666666667)(58714.86666666667,75417.6)(22129.133333333335,21530.166666666668)(61618.066666666666,50570.13333333333)(199834.36666666667,290645.5333333333)(282695.5,267385.1666666667)(465316.0333333333,387655.63333333336)(808463.2,749662.0333333333)(1054771.3333333333,915221.4666666667)(35670.9,43249.8)(97907.8,72170.06666666667)(141945.8,131194.46666666667)(191224.36666666667,175270.63333333333)(360074.76666666666,134775.83333333334)(360073.4,134774.8)(259871.46666666667,133688.8)(105421.2,95534.5)(113951.96666666666,88040.4)(18753.933333333334,12025.133333333333)(14360.766666666666,16313.6)(1087443.3666666667,1940210.5)
        };
    \addlegendentry{DTLS}

    \addplot[
        only marks,
        color=orange,
        mark=o,
        mark size=1pt
        ]
        coordinates {
            (56365.73333333333,53473.26666666667)(79367.03333333334,58859.333333333336)
        };
    \addlegendentry{Philips}

    \addplot[color=gray]
        coordinates {(1000,1000)(3000000,3000000) };

    \addplot[color=gray,style=dashed]
        coordinates {(1000,1000/2)(3000000,3000000/2) };

    \addplot[color=gray,style=dashed]
        coordinates {(1000,1000*2)(3000000,3000000*2) };

    \end{axis}
    \end{tikzpicture}
        \vspace{-0.5cm}
        \caption{Randomised $\mathsf{ETS}$}
        \label{fig:res3b}
    \end{subfigure}
    \begin{subfigure}[b]{0.3\textwidth}
        \begin{tikzpicture}
    \begin{axis}[
        legend style={nodes={scale=0.5, transform shape}},
        width=4cm,height=4cm,
        xmode=log,
        ymode=log,
        x label style={at={(axis description cs:0.4,-0.2)},anchor=north},
        y label style={at={(axis description cs:-0.2,.55)},anchor=south},
        xlabel=\tiny{Total symbols and resets baseline},
        ylabel=\tiny{Total symbols and resets \moe{*}},
        every x tick label/.append style={font=\scriptsize},
        every y tick label/.append style={font=\scriptsize},
        xtick pos=bottom,
        ytick pos=left,
        xmin=1000, xmax=100000,
        ymin=1000, ymax=100000,
        legend pos=north west,
    ]

    \addplot[
        only marks,
        color=cyan,
        mark=o,
        mark size=1pt
        ]
        coordinates {
            (3787.6666666666665,3732.9666666666667)(3646.733333333333,3575.133333333333)(3087.5333333333333,3059.2)(7105.133333333333,7064.566666666667)(6141.566666666667,6154.2)(6220.3,6265.466666666666)(11406.366666666667,11529.133333333333)(9880.3,9940.733333333334)(9233.066666666668,9105.133333333333)(11290.066666666668,11132.366666666667)(9538.8,9574.1)(8915.2,8993.566666666668)(20196.933333333334,20222.6)(18353.233333333334,18537.266666666666)(18103.2,17986.6)(30400.7,30008.733333333334)(28372.966666666667,29161.3)(28415.433333333334,28446.9)(20244.8,20229.433333333334)(16148.1,16243.766666666666)(16546.166666666668,16490.9)(35962.26666666667,35555.03333333333)(35690.26666666667,35694.26666666667)(29778.6,29812.9)(53469.333333333336,53683.36666666667)(46156.566666666666,47006.1)(44327.166666666664,43940.066666666666)
        };

    \addplot[color=gray]
        coordinates {(1000,1000)(100000,100000) };

    \addplot[color=gray,style=dashed]
        coordinates {(1000,1000/2)(100000,100000/2) };

    \addplot[color=gray,style=dashed]
        coordinates {(1000,1000*2)(100000,100000*2) };

    \end{axis}
    \end{tikzpicture}
        \vspace{-0.5cm}
        \caption{Randomly generated}
        \label{fig:res3c}
    \end{subfigure}
    \caption{Results Experiment 3.}
    \label{fig:res3}
\end{figure}

\myparagraph{Results.}
Fig.~\ref{fig:res3} shows the number of symbols and resets needed to learn and test a model (log-scaled). The y-axis shows \moe{$*$} and the x-axis shows the baseline. The diagonal solid lines correspond to using the same number of symbols and resets, the dotted lines indicate a factor two difference. Points in the right triangle indicate that \moe{$*$} used fewer symbols and resets than the baseline.

\myparagraph{Discussion}
From Fig.~\ref{fig:res3a}, we observe that \moe{$*$} slightly outperforms the baseline in the $k$-complete test suite setting. From Fig.~\ref{fig:res3b}, we observe that the performance of \moe{$*$} leads to slightly better results than the baseline. The performance is comparable for the randomly generated models (Fig.~\ref{fig:res3c}).

\mysubsubsection{Experiment 4}
We analyze runs of \moe{$*$} and the baseline for models \emph{m159} and \emph{m189} to provide insights on the behavior of the algorithms.

\myparagraph{Results.} Fig.~\ref{fig:res4} shows the runs of the first 3 seeds for \emph{m159} and \emph{m189}. Each data point at $(x, y)$ in the subplots represents one hypothesis, with $x$ states, that was learned using a total of $y$ symbols (notice the log scale).
%The x-axis displays the number of states and the y-axis shows the number of symbols used to learn the hypothesis. 
The green (or blue) lines correspond to runs with the baseline (or \moe{$*$}). The different markers for \moe{$*$} indicate which expert was used to generate the counterexample. %We start sampling experts once the hypothesis has more than 5 states (highlighted with a dashed line).
The vertical lines extending to $10^8$ indicate that the algorithm ran out of budget before learning the correct model.

\myparagraph{Discussion.}
In line with Experiment 2, we see that more runs lead to learning the full model using \moe{$*$}. The plots use the number of states as a rough progress measure. Based on this progress measure, we see that the difference is negligible for small hypothesis sizes, but for larger hypotheses, the difference is substantial. For \emph{m159}, we observe that the baseline runs out of budget before all states have been found, whereas the \moe{$*$} is able to learn the correct model within the budget (using the smaller test suites). In \emph{m189}, we observe a significant divergence in progress. On average, the future expert is most used to find counterexamples.

\begin{figure}[t]
    \centering
    \begin{subfigure}[b]{0.48\textwidth}
        \begin{tikzpicture}
\begin{axis}[
legend style={nodes={scale=0.5, transform shape}},width=9.0cm,height=6.0cm,ymode=log,ytick pos=left,xtick pos=bottom,xlabel={Hypothesis size},ylabel={Total symbols and resets},xmin=1, xmax=30,ymin=100, ymax=100000000,xtick={0,10,20,30},ytick={100, 10000, 1000000, 100000000},legend pos=south east,scale=0.6,vasymptote=5,]

% m159 full symb 1
\addplot[color=cyan,mark=o,mark size=0.5pt,]
coordinates{
(1,240)(2,747)(3,23992)(5,68169)(7,90667)(10,264266)(13,569648)(17,1415781)(23,1462931)(27,1787698)(27,200000000)};
\addlegendentry{baseline}

% m159 all symb 1
\addplot[color=green!60!black]
coordinates{
(1,240)(2,747)(3,23992)(5,68169)(7,77118)(11,653953)(14,704809)(17,794057)(23,836171)(27,908909)(30,23096833)};
\addlegendentry{\moe{*}}

% m159 all symb 1, expert 0
\addplot[only marks,color=green!60!black,mark=o,mark size=0.5pt,mark options={fill=green!60!black},]
coordinates{
(1,240)(2,747)(3,23992)(5,68169)(7,77118)(11,653953)};
\addlegendentry{E$_{\mathsf{T}}$}

% m159 all symb 1, expert 1
\addplot[only marks,color=green!60!black,mark=triangle,mark size=0.5pt,mark options={fill=green!60!black},]
coordinates{
(30,23096833)};
\addlegendentry{E$_{\mathsf{AI}}$}

% m159 all symb 1, expert 2
\addplot[only marks,color=green!60!black,mark=square,mark size=0.5pt,mark options={fill=green!60!black},]
coordinates{
(14,704809)(17,794057)(23,836171)(27,908909)};
\addlegendentry{E$_{\mathsf{F}}$}

% m159 all symb 1, expert 3
\addplot[only marks,color=green!60!black,mark=diamond,mark size=0.5pt,mark options={fill=green!60!black},]
coordinates{
};
% \addlegendentry{E$_{\mathsf{C}}$}

% m159 full symb 2
\addplot[color=cyan,mark=o,mark size=0.5pt,]
coordinates{
(1,240)(2,763)(3,40403)(4,54495)(5,225179)(7,263735)(11,437603)(12,515753)(17,575493)(27,655543)(30,675590)};

% m159 full symb 3
\addplot[color=cyan,mark=o,mark size=0.5pt,]
coordinates{
(1,240)(2,757)(3,18017)(4,42412)(5,52192)(7,61953)(10,85708)(13,103916)(14,669317)(23,871320)(23,200000000)};

% m159 all symb 2
\addplot[color=green!60!black]
coordinates{
(1,240)(2,763)(3,40403)(4,54495)(5,225179)(7,245072)(11,311211)(13,340736)(17,402710)(20,425326)(26,465926)(30,561186)};

% m159 all symb 3
\addplot[color=green!60!black]
coordinates{
(1,240)(2,757)(3,18017)(4,42412)(5,52192)(7,71757)(11,398077)(13,407414)(17,535224)(20,551784)(26,594484)(30,594484)};

% m159 all symb 2, expert 0
\addplot[only marks,color=green!60!black,mark=o,mark size=0.5pt,mark options={fill=green!60!black},]
coordinates{
(1,240)(2,763)(3,40403)(4,54495)(5,225179)(11,311211)};

% m159 all symb 2, expert 1
\addplot[only marks,color=green!60!black,mark=triangle,mark size=0.5pt,mark options={fill=green!60!black},]
coordinates{
(7,245072)(20,425326)};

% m159 all symb 2, expert 2
\addplot[only marks,color=green!60!black,mark=square,mark size=0.5pt,mark options={fill=green!60!black},]
coordinates{
(13,340736)(17,402710)(26,465926)(30,561186)};

% m159 all symb 2, expert 3
\addplot[only marks,color=green!60!black,mark=diamond,mark size=0.5pt,mark options={fill=green!60!black},]
coordinates{
};

% m159 all symb 3, expert 0
\addplot[only marks,color=green!60!black,mark=o,mark size=0.5pt,mark options={fill=green!60!black},]
coordinates{
(1,240)(2,757)(3,18017)(4,42412)(5,52192)(7,71757)(11,398077)};

% m159 all symb 3, expert 1
\addplot[only marks,color=green!60!black,mark=triangle,mark size=0.5pt,mark options={fill=green!60!black},]
coordinates{
(13,407414)};

% m159 all symb 3, expert 2
\addplot[only marks,color=green!60!black,mark=square,mark size=0.5pt,mark options={fill=green!60!black},]
coordinates{
(17,535224)(20,551784)(26,594484)(30,594484)};

% m159 all symb 3, expert 3
\addplot[only marks,color=green!60!black,mark=diamond,mark size=0.5pt,mark options={fill=green!60!black},]
coordinates{
};

\end{axis}
\end{tikzpicture}
    \end{subfigure}
    \begin{subfigure}[b]{0.48\textwidth}
        \begin{tikzpicture}
\begin{axis}[
legend style={nodes={scale=0.5, transform shape}},width=9.0cm,height=6.0cm,ymode=log,xtick pos=bottom,ytick pos=left,xlabel={Hypothesis size},ylabel={Total symbols and resets},xmin=1, xmax=189,ymin=100, ymax=100000000,xtick={0,50,100,150,189},ytick={100, 10000, 1000000, 100000000},legend pos=south east,scale=0.6,vasymptote=5,]

% m189 full symb 1
\addplot[color=cyan,mark=o,mark size=0.5pt,]
coordinates{
(1,276)(2,885)(3,4506)(4,28763)(5,34080)(6,81434)(8,263772)(9,533943)(11,580323)(17,689048)(19,710403)(20,722548)(21,801673)(26,899407)(29,1079103)(30,1783792)(36,2466944)(37,3551794)(39,3957477)(40,4011477)(42,4051487)(44,4784311)(45,4898956)(46,5207361)(49,5977900)(53,8209627)(54,9511971)(55,11983477)(56,13448863)(57,13565205)(63,13649494)(64,13734931)(69,13829121)(71,18060106)(74,20309063)(77,20501794)(81,21823982)(82,29834930)(92,32340271)(93,32785385)(95,33588415)(102,34159396)(105,35587531)(106,35726127)(107,36004808)(109,43041219)(110,43093642)(118,47011955)(119,53937707)(120,83308813)(122,91972117)(123,93490068)(123,200000000)};
\addlegendentry{baseline}

% m189 all symb 1
\addplot[color=green!60!black]
coordinates{
(1,276)(2,885)(3,4506)(4,28763)(5,34080)(6,153249)(8,163269)(9,173044)(10,518825)(12,845758)(19,912714)(22,950074)(23,962170)(27,1010920)(31,1048813)(32,1059283)(36,1131455)(37,1151718)(38,1165666)(39,1180808)(40,1194703)(41,1227463)(42,1353225)(43,1508551)(46,1624242)(47,2044694)(48,2060079)(59,2405329)(60,2428852)(63,2484314)(65,2550597)(66,2575445)(69,2716143)(73,2834379)(74,3071452)(75,3116172)(77,5244944)(82,5298758)(85,6922009)(90,9818977)(99,10086496)(101,10183182)(102,10244411)(103,10303585)(104,10477315)(107,11537211)(114,11797885)(115,11927993)(117,12467973)(126,12835310)(129,12898823)(130,13893145)(131,14981157)(133,15023899)(135,40207941)(146,40610227)(148,40713481)(149,41137962)(151,44622629)(153,45145343)(157,49911189)(158,49969338)(160,50300847)(161,87624228)(165,87712602)(166,87734703)(166,200000000)};
\addlegendentry{\moe{*}}

% m189 all symb 1, expert 0
\addplot[only marks,color=green!60!black,mark=o,mark size=0.5pt,mark options={fill=green!60!black},]
coordinates{
(1,276)(2,885)(3,4506)(4,28763)(5,34080)(6,153249)(8,163269)(12,845758)(19,912714)(82,5298758)(99,10086496)(146,40610227)};
\addlegendentry{E$_{\mathsf{T}}$}

% m189 all symb 1, expert 1
\addplot[only marks,color=green!60!black,mark=triangle,mark size=0.5pt,mark options={fill=green!60!black},]
coordinates{
(38,1165666)(48,2060079)(104,10477315)(135,40207941)};
\addlegendentry{E$_{\mathsf{AI}}$}

% m189 all symb 1, expert 2
\addplot[only marks,color=green!60!black,mark=square,mark size=0.5pt,mark options={fill=green!60!black},]
coordinates{
(9,173044)(10,518825)(22,950074)(23,962170)(27,1010920)(31,1048813)(32,1059283)(36,1131455)(37,1151718)(39,1180808)(40,1194703)(41,1227463)(42,1353225)(43,1508551)(46,1624242)(47,2044694)(59,2405329)(60,2428852)(63,2484314)(65,2550597)(66,2575445)(73,2834379)(74,3071452)(75,3116172)(77,5244944)(85,6922009)(90,9818977)(101,10183182)(102,10244411)(103,10303585)(107,11537211)(114,11797885)(115,11927993)(117,12467973)(126,12835310)(130,13893145)(131,14981157)(133,15023899)(148,40713481)(149,41137962)(151,44622629)(153,45145343)(157,49911189)(158,49969338)(160,50300847)(161,87624228)(165,87712602)(166,87734703)(166,200000000)};
\addlegendentry{E$_{\mathsf{F}}$}

% m189 all symb 1, expert 3
\addplot[only marks,color=green!60!black,mark=diamond,mark size=0.5pt,mark options={fill=green!60!black},]
coordinates{
(69,2716143)(129,12898823)};
\addlegendentry{E$_{\mathsf{C}}$}

% m189 full symb 2
\addplot[color=cyan,mark=o,mark size=0.5pt,]
coordinates{
(1,276)(2,904)(3,2989)(4,14234)(5,25643)(6,31314)(7,87483)(9,189890)(10,294086)(12,415543)(15,513205)(20,558914)(22,585627)(24,1319682)(28,1356199)(34,1488508)(35,1510244)(36,1522550)(40,1673305)(41,1788517)(46,2390420)(50,4410024)(55,4916638)(56,5227680)(57,5489239)(61,7428422)(62,7814032)(63,8416439)(64,8912066)(65,12980015)(68,15822171)(78,16072520)(79,16213900)(82,16396418)(84,16477676)(85,16640311)(86,18156581)(93,18498137)(96,18570363)(97,19764394)(99,20990415)(100,21189139)(101,27932786)(102,31942805)(104,34992193)(111,35660799)(113,38637092)(114,40803176)(115,45485276)(119,47475440)(120,47551280)(122,48611203)(123,53902290)(124,56203312)(124,200000000)};

% m189 full symb 3
\addplot[color=cyan,mark=o,mark size=0.5pt,]
coordinates{
(1,276)(2,897)(3,9367)(4,131682)(5,139942)(6,263729)(8,381543)(14,459657)(15,471314)(17,497379)(18,509536)(19,602216)(20,696825)(23,1142199)(26,1251372)(30,1309700)(32,1988043)(34,2411252)(36,2616874)(37,2685887)(38,3617208)(41,4805461)(47,5329396)(48,8323923)(52,8601434)(53,8950411)(55,9113860)(56,9585740)(59,23417363)(61,23445634)(62,24587889)(63,29148949)(76,39079307)(77,39760259)(78,39796450)(79,39832424)(83,40538284)(84,40555260)(85,40574359)(92,40852161)(95,40937085)(98,42115128)(99,43584690)(100,45831314)(107,46183107)(108,47069845)(110,47515629)(112,47575067)(114,47890127)(116,53084563)(117,84753745)(118,84875226)(119,99214347)(119,200000000)};

% m189 all symb 2
\addplot[color=green!60!black]
coordinates{
(1,276)(2,904)(3,2989)(4,14234)(5,25643)(6,121724)(7,364303)(9,378440)(10,390555)(12,648414)(16,683775)(17,692897)(19,1908995)(23,1945532)(30,2040844)(34,2109045)(35,2116615)(36,2132990)(38,2166258)(39,2188411)(43,2259669)(47,2313806)(49,2345056)(52,2393556)(53,2436253)(56,2493922)(63,2664503)(68,2788339)(69,2830095)(70,2889795)(71,2920935)(74,3162466)(75,3208696)(80,3415644)(85,3554557)(86,3784364)(90,3891636)(91,3927208)(92,3973807)(96,4074872)(97,4094478)(98,4186639)(101,4221284)(102,4321153)(103,4708076)(106,5236479)(109,5408028)(113,5560678)(115,5674765)(116,5762792)(121,7253506)(122,7832065)(124,11685823)(134,29999675)(135,30028846)(147,30588369)(149,30983300)(151,31060919)(153,31179830)(155,31232333)(158,31343788)(160,32986905)(164,33164506)(166,54580324)(166,200000000)};

% m189 all symb 3
\addplot[color=green!60!black]
coordinates{
(1,276)(2,897)(3,9367)(4,131682)(5,139942)(6,517338)(8,527921)(9,538733)(10,557942)(11,566377)(21,816950)(25,866942)(29,944822)(32,1034306)(33,1062069)(36,1122374)(40,1170421)(42,1195000)(43,1252672)(49,1403251)(61,1686767)(62,1703867)(66,1788857)(67,1809427)(68,1831239)(69,1848363)(73,1931851)(74,1963638)(76,2015514)(77,2047212)(78,2429431)(79,2731026)(80,5651340)(84,7742181)(85,8037001)(95,11567848)(96,11587551)(105,11799416)(106,13185979)(110,13268258)(111,13297699)(113,13414749)(114,13639349)(115,16150973)(121,21685931)(122,21756614)(125,21856828)(130,25094100)(138,25369332)(139,25403302)(140,25566989)(142,25640826)(144,26183368)(154,32059559)(159,32148148)(160,32170887)(162,32593108)(166,32771799)(168,36258755)(168,200000000)};

% m189 all symb 2, expert 0
\addplot[only marks,color=green!60!black,mark=o,mark size=0.5pt,mark options={fill=green!60!black},]
coordinates{
(1,276)(2,904)(3,2989)(4,14234)(5,25643)(6,121724)(7,364303)(9,378440)(16,683775)(23,1945532)(36,2132990)(85,3554557)};

% m189 all symb 2, expert 1
\addplot[only marks,color=green!60!black,mark=triangle,mark size=0.5pt,mark options={fill=green!60!black},]
coordinates{
(19,1908995)(43,2259669)(68,2788339)(115,5674765)};

% m189 all symb 2, expert 2
\addplot[only marks,color=green!60!black,mark=square,mark size=0.5pt,mark options={fill=green!60!black},]
coordinates{
(10,390555)(12,648414)(17,692897)(30,2040844)(35,2116615)(38,2166258)(39,2188411)(49,2345056)(52,2393556)(53,2436253)(56,2493922)(63,2664503)(69,2830095)(70,2889795)(71,2920935)(74,3162466)(75,3208696)(86,3784364)(90,3891636)(91,3927208)(92,3973807)(96,4074872)(97,4094478)(98,4186639)(101,4221284)(102,4321153)(103,4708076)(106,5236479)(109,5408028)(113,5560678)(116,5762792)(121,7253506)(122,7832065)(124,11685823)(134,29999675)(135,30028846)(147,30588369)(149,30983300)(153,31179830)(155,31232333)(158,31343788)(164,33164506)(166,54580324)(166,200000000)};

% m189 all symb 2, expert 3
\addplot[only marks,color=green!60!black,mark=diamond,mark size=0.5pt,mark options={fill=green!60!black},]
coordinates{
(34,2109045)(47,2313806)(80,3415644)(151,31060919)(160,32986905)};

% m189 all symb 3, expert 0
\addplot[only marks,color=green!60!black,mark=o,mark size=0.5pt,mark options={fill=green!60!black},]
coordinates{
(1,276)(2,897)(3,9367)(4,131682)(5,139942)(6,517338)(11,566377)(25,866942)(121,21685931)(159,32148148)};

% m189 all symb 3, expert 1
\addplot[only marks,color=green!60!black,mark=triangle,mark size=0.5pt,mark options={fill=green!60!black},]
coordinates{
(8,527921)(9,538733)(29,944822)(49,1403251)(66,1788857)(78,2429431)(96,11587551)(154,32059559)};

% m189 all symb 3, expert 2
\addplot[only marks,color=green!60!black,mark=square,mark size=0.5pt,mark options={fill=green!60!black},]
coordinates{
(10,557942)(21,816950)(32,1034306)(33,1062069)(40,1170421)(42,1195000)(61,1686767)(62,1703867)(67,1809427)(68,1831239)(74,1963638)(76,2015514)(77,2047212)(79,2731026)(80,5651340)(84,7742181)(85,8037001)(95,11567848)(105,11799416)(106,13185979)(110,13268258)(111,13297699)(113,13414749)(114,13639349)(115,16150973)(122,21756614)(125,21856828)(130,25094100)(138,25369332)(139,25403302)(140,25566989)(142,25640826)(144,26183368)(160,32170887)(162,32593108)(166,32771799)(168,36258755)(168,200000000)};

% m189 all symb 3, expert 3
\addplot[only marks,color=green!60!black,mark=diamond,mark size=0.5pt,mark options={fill=green!60!black},]
coordinates{
(36,1122374)(43,1252672)(69,1848363)(73,1931851)};

\end{axis}
\end{tikzpicture}
    \end{subfigure}
    \vspace{-0.5cm}
    \caption{Results Experiment 4 for \emph{m159} (left) and \emph{m189} (right).}
    \label{fig:res4}
\end{figure}
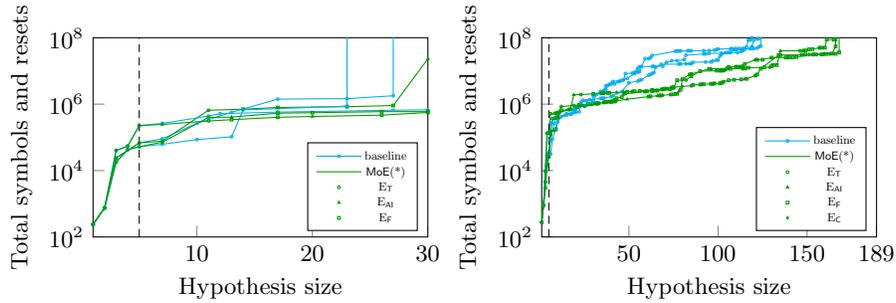

\section{Related Work} \label{sec:related}

\myparagraph{Test suites.}
The use of conformance testing~\cite{2004test} is standard in automata learning~\cite{VaandragerLearning2017} and goes back to~\cite{PeledVY02}.
There are several recent evaluations comparing sample-based conformance testing techniques~\cite{AichernigBenchmarking2020,AslamThesis2021,GarhewalExperimental2023}; these comparisons are orthogonal to the current paper.
Another idea is to use mutation testing~\cite{AichernigMutationTesting2019}. Mutation testing performs well on small models ($<100$\,states) but \cite{AichernigMutationTesting2019} notices that this technique is computationally too expensive for large models. 

\myparagraph{Increasing the alphabet size.} 
Instead of reducing the alphabet size for a more guided counterexample finding, a common theme is to use abstraction refinement~\cite{VaandragerLearning2017, HowarS2018} during learning to iteratively refine the alphabet. Bobaru et al.~\cite{BobaruPG08CAV} learn models using abstractions of components to later show that a property holds or is violated. Additionally, Vaandrager and Wi{\ss}man~\cite{VaandragerW23ActionCodes} formally describe the relation between high-level state machines and low-level models using abstraction refinement.

\myparagraph{Using the automata structure.}
A recent trend is gray-box automata learning, which assumes partial information on the SUL and aims to exploit this information. In particular, learning algorithms addressing various types of composition (sequential, parallel, product) have been investiged~\cite{AbelSerialComposition2016,MoermanProduct2018,LabbafCompositional2023, NeeleCompositional2023}.
However, all these techniques adapt the learning algorithm, not the testing algorithm, as in the current paper. Furthermore, while the results in Sec.~\ref{sec:subalphabet} are similar to a gray-box setting, the idea in Sec.~\ref{sec:subalphabetselection} is that this work leads to better performance in the strict black-box setting, as highlighted by the experiments.  

\myparagraph{Algorithm selection.} 
Machine learning for algorithm selection is an active area of research, see e.g.,~\cite{DBLP:journals/ml/AbdulrahmanBRV18,DBLP:journals/ec/KerschkeHNT19} and has been applied successfully, e.g., in the context of SAT checking~\cite{DBLP:series/faia/HoosHL21}. In formal methods, multi-armed bandits framework has been used, e.g., to prioritize SMT solver over others~\cite{PimpalkhareMedleysolver2021} or to guide falsification processes for hybrid systems~\cite{DBLP:conf/cav/ZhangHA19}. In  automata learning, bandits have recently been applied to select between different oracles for answering output queries~\cite{DBLP:journals/corr/abs-2307-10434}.

\section{Conclusion} \label{sec:conclusion}
In this paper, we introduced smaller test suites for conformance testing that preserve the typical completeness guarantees under natural assumptions on the learned system. The paper demonstrates that a combination of these test suites and a multi-armed bandit formulation significantly accelerates modern active automata learning, even when the assumptions do not hold. Natural extensions include adding additional small test suites, designing variations of the presented experts to, for example, handle parallel components\cite{LabbafCompositional2023,NeeleCompositional2023}, and using a multi-armed bandit to select the essential parameter $k$. Furthermore, our approach paves the way for using similar assumptions to those made for the completeness of the expert test suites in other aspects of active automata learning. 

\newpage
\bibliographystyle{plainurl}
\bibliography{references.bib}

\ifthenelse{\boolean{arxivversion}}{
  \newpage 
  \appendix

\section{Proofs for Sec.~\ref{sec:completeness}} \label{sec:appA}

\begin{proof}[Lemma~\ref{lem:reach-new}]
Let $n = |\H|$ be the number of states in $\H$. First, we prove that the state cover $P$ for $\H$ reaches $n$
distinct states in $\S$. To this end, let $w,v \in P$, $p = \delta^{\H}(w)$ and $q = \delta^{\H}(v)$ such that $p \nsim_W q$. Let $p' = \delta^{\S}(w)$ and $q' = \delta^{\S}(v)$. Then $p' \nsim_W q'$ because the machines agree on tests of the form $P \cdot W$. This indicates that $\Delta^\S(P)$ contains at least $n$ elements.

Next, we prove that $P \cdot I^{\leq k}$ is a state cover for $\S$. First, notice that $I^{\leq k + n}$ is trivially a state cover for $\S$, and since $\varepsilon \in P$, so is $P \cdot I^{\leq k+n}$.
Now suppose $q \in Q^{\S}$ is reachable; take a word $w \in P \cdot I^{n+k}$ so that the corresponding path in $\S$ from $q_0$ is cycle-free.
Decompose $w=vu$ where $v$ is the longest prefix such that  $\delta^\S(v) \in \Delta^\S(P)$; at least one such prefix exists since $w \in P \cdot I^*$. The path reading $u$ starting from $\delta^\S(v)$ is cycle-free (by the assumption on $w$); and since $v$ is chosen to be of maximal length, we must have $|u| \leq k$ (for $u$ to be longer, it either needs a cycle in the underlying path or encounter a state in $\Delta^\S(P)$, since $|\Delta^S(P)| \geq n$ and $|\S| \leq n+k$).

Finally, since $\delta^\S(v) \in \Delta^\S(P)$ there exists $v' \in P$ such that $\delta^\S(v) = \delta^\S(v')$, and we derive:
$$\delta^\S(v'u) = \delta^{\S}(vu) = \delta^{\S}(w) = q \,.$$ 
Since $v'u \in P \cdot I^{\leq k}$, we are done. 
\end{proof}

The proof of Lemma~\ref{lem:bisim-new} follows~\cite[Ch.\ 2]{MoermanThesis2019} closely, which in turn is based on~\cite{ChowTesting1978}. We reproduce the proof from~\cite{MoermanThesis2019} here for convenience. It makes use of bisimulations.

\begin{definition}
    A relation $R \subseteq Q \times Q$ on states is a \emph{bisimulation up-to} $\sim$ if for every $p ~R~ q$ we have
    \begin{itemize}[beginpenalty=99,topsep=1pt]
        \item equal outputs: $\lambda(p,i)=\lambda(q,i)$ for all $i \in I$,
        \item related successor states: $\delta(p,i) ~R~ \delta(q,i)$ or 
        there exists states $s, t \in Q$ such that $\delta(p,i) \sim s ~R~ t \sim \delta(q,i)$ for all $i \in I$.
    \end{itemize}
\end{definition}
\begin{lemma} \label{lem:bisimeq}
    Let $R$ be a bisimulation up-to $\sim$. If $(q_0^{\S},q_0^{\H}) \in R$, then $\H \sim \S$.
\end{lemma}

\begin{proof}
    This follows by a standard argument: show that $\sim R \sim$ is a bisimulation.
\end{proof}

\begin{proof}[Lemma~\ref{lem:bisim-new}]
We show that all reached states are bisimilar up-to $\sim$ using the following relation: $$R = \{ (\delta^{\H}(w), \delta^{\S}(w)) \mid w \in L \}\,.$$

First, we prove the outputs are equal. Since $T=L \cdot I^{\leq 1} \cdot W$ and $W$ is non-empty, for all $w\in L$ and $i \in I$ there exists
$v \in W$ such that $wiv \in T$. Since $\H \sim_T \S$ we get $\lambda^{\H}(q_0^{\H}, wiv) = \lambda^{\S}(q_0^{\S},wiv)$ which implies 
$\lambda^{\H}(q_0^{\H}, wi) = \lambda^{\S}(q_0^{\S},wi)$. As a consequence we have 
$\lambda^{\H}(\delta^{\H}(w),i) = \lambda^{\S}(\delta^{\S}(w),i)$ as needed.

Next, we prove that the successors are related, that is, $\delta^\H(wi) \sim R \sim \delta^\S(wi)$; the challenge here
is that we do not know whether $wi \in L$. Since $L$ is a state cover for $\S$, there is a word $v \in L$ such that
$\delta^\S(wi) = \delta^\S(v)$, and since $\H \sim_T \S$, we have $\delta^\H(v) \sim_W \delta^\S(v)$. 
We derive $\delta^\H(wi) \sim_W \delta^S(wi) = \delta^\S(v) \sim_W \delta^\H(v)$ so that $\delta^\H(wi) \sim_W \delta^\H(v)$.
Since $W$ is a characterization set for $\H$, this implies $\delta^\H(wi) \sim \delta^\H(v)$. We conclude
$$
\delta^\H(wi) \sim \delta^\H(v) \mathrel{R} \delta^\S(v) = \delta^\S(wi)
$$
as desired.
\end{proof}

\section{Proofs for Sec.~\ref{sec:subalphabet}} \label{sec:proofs-experts}

In this section we prove the specific completeness results for the various experts, which are stated in Theorems~\ref{thm:activeinputs},~\ref{thm:future} and Theorem~\ref{thm:components}. The common proof strategy is, in each of the cases, to observe that the assumptions yield variants of Lemma~\ref{lem:reach-new}. The completeness results then follow from Lemma~\ref{lem:bisim-new}.

In order to apply this lemma, it is useful to reformulate the test suite generated by an expert a bit. We have:
\begin{align*}
\ETS_{E,k}(\H) &= \bigcup_{v \in P} ( v \cdot (\bigcup_{J \in E(\H,v)} J^{\leq k-1}) \cdot I^{\leq 2} \cdot W) \\
&= \left(\bigcup_{v \in P} v \cdot (\bigcup_{J \in E(\H,v)} J^{\leq k-1}) \cdot I^{\leq 1}\right)\cdot I^{\leq 1} \cdot W
\end{align*}
Let $L$ be the left part:
$$
 L = \bigcup_{v \in P} v \cdot (\bigcup_{J \in E(\H,v)} J^{\leq k-1}) \cdot I^{\leq 1} \,.
$$
Now, it suffices to show that in each of the theorems, $L$ is a state cover for $\S$. 

\begin{itemize}
    \item \textbf{Theorem~\ref{thm:activeinputs}.} Here $L$ reduces to 
    \[ P \cdot (I^{\activeMM{(\H)}})^{\leq k-1} \cdot I^{\leq 1} \,.\]

    We prove that $P \cdot (I^{\activeMM(\H)})^{\leq k}$ is a state cover for $\S$. From assumption 
    \[Q^{\S} \setminus Q_{sink}^{\S} \subseteq \Delta^\S(P \cdot (I^{\activeMM(\H)})^{\leq k}) \,,\]
    it follows that $L$ reaches all states in $Q^{\S} \setminus Q_{sink}^{\S}$. 
    Next, we prove that all sink states are also reached by $L$. Let $q \in Q_{sink}$ and $p$ such that for some $i \in I$, $\delta^{\S}(\delta^{\S}(p),i)=q$. Due to the definition of sink states, $p \in Q^{\S} \setminus Q_{sink}^{\S}$. Because $\S$ has at most $k$ more states than $\H$, in the worst case $q$ is reachable in $k$ steps from some state in $\H$. Therefore, $p$ is reachable in at most $k-1$ steps. By rewriting the assumption to
    \[Q^{\S} \setminus Q_{sink}^{\S} \subseteq \Delta^\S(P \cdot (I^{\activeMM(\H)})^{\leq k-1} \cdot I) \,, \]
    it follows that all sink states are also reached. Because we reach all states of $\S$ using $L$, $L$ is a state cover for $\S$.
    \item \textbf{Theorem~\ref{thm:future}.} Here $L$ reduces to 
    \[ \bigcup_{v \in P} v \cdot (I_{v,l})^{\leq k-1} \cdot I^{\leq 1} \,. \]
    From assumption
    \[ Q^\S \setminus Q^\S_{sink} \subseteq \bigcup_{v\in P} \Delta^\S(v \cdot I_{v, l}^{\leq k})\,,\]
    it follows that $L$ reaches all states in $Q^\S \setminus Q^\S_{sink}$. Using the same rewriting trick as above, we find that all sink states are also reached. Because we reach all states of $\S$ using $L$, $L$ is a state cover for $\S$.
    \item \textbf{Theorem~\ref{thm:components}.} Here $L$ reduces to
    \[ P \cdot (\bigcup_{X \in g(\H)} I_X^{\leq k-1}) \cdot I^{\leq 1} \,, \]
    From assumption
    \[ Q^{\S} \setminus Q_{sink}^{\S} \subseteq \bigcup_{X \in g(\H)} \Delta^\S(P \cdot I_X^{\leq k})\,,\]
    it follows that $L$ reaches all states in $Q^\S \setminus Q^\S_{sink}$. Using the same rewriting trick as above, we find that all sink states are also reached. Because we reach all states of $\S$ using $L$, $L$ is a state cover for $\S$.
    
\end{itemize}

  \section{Complete benchmarking results} \label{sec:appC}
In Tables~\ref{tab:exp1} to~\ref{tab:exp3_random}, we list the number of times the model was fully learned, the average number of learned states and the average number of symbols and resets for each model. For Experiment 2, the total number of symbols and resets is sometimes higher than the budget because we always complete an asked OQ.

\begin{table}[t]
    \centering
    \begin{adjustbox}{width=\textwidth}
{\footnotesize
\begin{tabular}{@{}@{}l@{\hspace{1mm}}l@{}c@{}l@{}c@{}r@{\hspace{1mm}}r@{}c@{}r@{\hspace{1mm}}r@{}c@{}r@{\hspace{1mm}}r@{}c@{}r@{}@{}}
\toprule%
\multicolumn{2}{@{}l@{}}{Model}
&\hspace*{2mm}&
\multicolumn{1}{@{}l@{}}{Algo-}
&\hspace*{2mm}&
\multicolumn{2}{@{}l@{}}{Learned}
&\hspace*{2mm}&
\multicolumn{2}{@{}l@{}}{Learn-}
&\hspace*{2mm}&
\multicolumn{2}{@{}l@{}}{Test-}
&\hspace*{2mm}&
\multicolumn{1}{@{}l@{}}{}
\\
\cline{1-2}
% no cline for column groups of only one column
\cline{6-7}
\cline{9-10}
\cline{12-13}
% no cline for column groups of only one column
\\[-4pt]
$n$
&
$k$
&
&
rithm
&
&
Seeds
&
States
&
&
Inputs
&
Resets
&
&
Inputs
&
Resets
&
&
Total
\\
\midrule
\multicolumn{2}{@{}l@{}}{$\mathsf{ASML}$\_3\_10}&&\moe{$*$}&&\bestvalue{30}&\bestvalue{12.0}&&2061.87&475.97&&\bestvalue{4366.77}&\bestvalue{2012.43}&&\bestvalue{8917.03}\\
&&&\moe{$E_{\mathsf{T}}, E_{\mathsf{AI}}$}&&\bestvalue{30}&\bestvalue{12.0}&&2062.87&475.77&&4705.87&2120.57&&9365.07\\
$n$ = 12&$k$ = 15&&baseline&&\bestvalue{30}&\bestvalue{12.0}&&\bestvalue{1730.83}&\bestvalue{415.53}&&61100.80&17481.83&&80729.00\\
\arrayrulecolor{black!20}\hline\arrayrulecolor{black}
\multicolumn{2}{@{}l@{}}{$\mathsf{ASML}$\_3\_15}&&\moe{$*$}&&\bestvalue{30}&\bestvalue{12.0}&&2816.60&655.87&&\bestvalue{6992.27}&\bestvalue{3290.60}&&\bestvalue{13755.33}\\
&&&\moe{$E_{\mathsf{T}}, E_{\mathsf{AI}}$}&&\bestvalue{30}&\bestvalue{12.0}&&2823.97&657.00&&7150.50&3359.73&&13991.20\\
\bottomrule
\end{tabular}
}
\end{adjustbox}
\caption{Experiment 1.}
\label{tab:exp1}
\end{table}
\begin{table}[t]
    \centering
    \begin{adjustbox}{width=\textwidth}
{\footnotesize
\begin{tabular}{@{}@{}l@{\hspace{1mm}}l@{}c@{}l@{}c@{}r@{\hspace{1mm}}r@{}c@{}r@{\hspace{1mm}}r@{}c@{}r@{\hspace{1mm}}r@{}c@{}r@{}@{}}
\toprule%
\multicolumn{2}{@{}l@{}}{Model}
&\hspace*{2mm}&
\multicolumn{1}{@{}l@{}}{Algo-}
&\hspace*{2mm}&
\multicolumn{2}{@{}l@{}}{Learned}
&\hspace*{2mm}&
\multicolumn{2}{@{}l@{}}{Learn-}
&\hspace*{2mm}&
\multicolumn{2}{@{}l@{}}{Test-}
&\hspace*{2mm}&
\multicolumn{1}{@{}l@{}}{}
\\
\cline{1-2}
% no cline for column groups of only one column
\cline{6-7}
\cline{9-10}
\cline{12-13}
% no cline for column groups of only one column
\\[-4pt]
$n$
&
$k$
&
&
rithm
&
&
Seeds
&
States
&
&
Inputs
&
Resets
&
&
Inputs
&
Resets
&
&
Total
\\
\midrule
\multicolumn{2}{@{}l@{}}{m106}&&\moe{$*$}&&\bestvalue{49}&\bestvalue{78.2}&&966575.80&76418.06&&\bestvalue{18462579.92}&\bestvalue{2464088.66}&&\bestvalue{21969662.44}\\
$n$ = 79&$k$ = 83&&baseline&&2&61.6&&\bestvalue{636337.08}&\bestvalue{54918.00}&&88598787.58&9389769.20&&98679811.86\\
\arrayrulecolor{black!20}\hline\arrayrulecolor{black}
\multicolumn{2}{@{}l@{}}{m135}&&\moe{$*$}&&\bestvalue{49}&\bestvalue{56.8}&&1183035.94&93093.44&&\bestvalue{10382351.86}&\bestvalue{1707836.92}&&\bestvalue{13366318.16}\\
$n$ = 57&$k$ = 177&&baseline&&5&48.4&&\bestvalue{903687.70}&\bestvalue{74401.40}&&84800842.84&8739453.70&&94518385.64\\
\arrayrulecolor{black!20}\hline\arrayrulecolor{black}
\multicolumn{2}{@{}l@{}}{m158}&&\moe{$*$}&&\bestvalue{43}&\bestvalue{27.2}&&144418.38&23827.92&&\bestvalue{15572983.04}&\bestvalue{2870102.12}&&\bestvalue{18611331.46}\\
\bottomrule
\end{tabular}
}
\end{adjustbox}
\caption{Experiment 2, Budget $10^8$.}
\label{tab:exp2_full}
\end{table}
\begin{table}[t]
    \centering
    \begin{adjustbox}{width=\textwidth}
{\footnotesize
\begin{tabular}{@{}@{}l@{\hspace{1mm}}l@{}c@{}l@{}c@{}r@{\hspace{1mm}}r@{}c@{}r@{\hspace{1mm}}r@{}c@{}r@{\hspace{1mm}}r@{}c@{}r@{}@{}}
\toprule%
\multicolumn{2}{@{}l@{}}{Model}
&\hspace*{2mm}&
\multicolumn{1}{@{}l@{}}{Algo-}
&\hspace*{2mm}&
\multicolumn{2}{@{}l@{}}{Learned}
&\hspace*{2mm}&
\multicolumn{2}{@{}l@{}}{Learn-}
&\hspace*{2mm}&
\multicolumn{2}{@{}l@{}}{Test-}
&\hspace*{2mm}&
\multicolumn{1}{@{}l@{}}{}
\\
\cline{1-2}
% no cline for column groups of only one column
\cline{6-7}
\cline{9-10}
\cline{12-13}
% no cline for column groups of only one column
\\[-4pt]
$n$
&
$k$
&
&
rithm
&
&
Seeds
&
States
&
&
Inputs
&
Resets
&
&
Inputs
&
Resets
&
&
Total
\\
\midrule
\multicolumn{2}{@{}l@{}}{m106}&&\moe{$*$}&&\bestvalue{47}&\bestvalue{76.7}&&940846.06&74749.16&&\bestvalue{16482047.54}&\bestvalue{2217370.32}&&\bestvalue{19715013.08}\\
$n$ = 79&$k$ = 83&&baseline&&0&57.9&&\bestvalue{564424.54}&\bestvalue{51072.42}&&44431835.82&4952672.88&&50000005.66\\
\arrayrulecolor{black!20}\hline\arrayrulecolor{black}
\multicolumn{2}{@{}l@{}}{m135}&&\moe{$*$}&&\bestvalue{49}&\bestvalue{56.8}&&1183035.94&93093.44&&\bestvalue{9461189.48}&\bestvalue{1628999.40}&&\bestvalue{12366318.26}\\
$n$ = 57&$k$ = 177&&baseline&&3&43.5&&\bestvalue{796915.26}&\bestvalue{67280.54}&&42856609.20&4940568.88&&48661373.88\\
\arrayrulecolor{black!20}\hline\arrayrulecolor{black}
\multicolumn{2}{@{}l@{}}{m158}&&\moe{$*$}&&\bestvalue{42}&\bestvalue{27.1}&&144172.54&23788.68&&\bestvalue{9261593.18}&\bestvalue{1727139.34}&&\bestvalue{11156693.74}\\
\bottomrule
\end{tabular}
}
\end{adjustbox}
\caption{Experiment 2, Budget $5\cdot 10^7$.}
\label{tab:exp2_half}
\end{table}
\begin{table}[t]
    \centering
    \begin{adjustbox}{width=\textwidth}
{\footnotesize
\begin{tabular}{@{}@{}l@{\hspace{1mm}}l@{}c@{}l@{}c@{}r@{\hspace{1mm}}r@{}c@{}r@{\hspace{1mm}}r@{}c@{}r@{\hspace{1mm}}r@{}c@{}r@{}@{}}
\toprule%
\multicolumn{2}{@{}l@{}}{Model}
&\hspace*{2mm}&
\multicolumn{1}{@{}l@{}}{Algo-}
&\hspace*{2mm}&
\multicolumn{2}{@{}l@{}}{Learned}
&\hspace*{2mm}&
\multicolumn{2}{@{}l@{}}{Learn-}
&\hspace*{2mm}&
\multicolumn{2}{@{}l@{}}{Test-}
&\hspace*{2mm}&
\multicolumn{1}{@{}l@{}}{}
\\
\cline{1-2}
% no cline for column groups of only one column
\cline{6-7}
\cline{9-10}
\cline{12-13}
% no cline for column groups of only one column
\\[-4pt]
$n$
&
$k$
&
&
rithm
&
&
Seeds
&
States
&
&
Inputs
&
Resets
&
&
Inputs
&
Resets
&
&
Total
\\
\midrule
\multicolumn{2}{@{}l@{}}{BitVise}&&\moe{$*$}&&\bestvalue{30}&\bestvalue{66.0}&&33631.13&3036.03&&\bestvalue{1765871.93}&\bestvalue{161971.90}&&\bestvalue{1964511.00}\\
$n$ = 66&$k$ = 13&&baseline&&\bestvalue{30}&\bestvalue{66.0}&&\bestvalue{33299.40}&\bestvalue{3035.93}&&2872148.33&251058.90&&3159542.57\\
\arrayrulecolor{black!20}\hline\arrayrulecolor{black}
\multicolumn{2}{@{}l@{}}{DropBear}&&\moe{$*$}&&\bestvalue{30}&\bestvalue{17.0}&&\bestvalue{3653.20}&\bestvalue{595.37}&&\bestvalue{428684.97}&\bestvalue{65717.03}&&\bestvalue{498650.57}\\
$n$ = 17&$k$ = 13&&baseline&&\bestvalue{30}&\bestvalue{17.0}&&3678.90&600.83&&485835.80&73076.00&&563191.53\\
\arrayrulecolor{black!20}\hline\arrayrulecolor{black}
\multicolumn{2}{@{}l@{}}{GnuTLS\_3.3.8\_client\_full}&&\moe{$*$}&&\bestvalue{30}&\bestvalue{15.0}&&956.97&\bestvalue{194.23}&&\bestvalue{34058.13}&\bestvalue{6316.13}&&\bestvalue{41525.47}\\
\bottomrule
\end{tabular}
}
\end{adjustbox}
\caption{Experiment 3, $\ETS$ ($k = 2$), Part 1.}
\label{tab:exp3_fixed_1}
\end{table}

\begin{table}[t]
    \centering
    \begin{adjustbox}{width=\textwidth}
{\footnotesize
\begin{tabular}{@{}@{}l@{\hspace{1mm}}l@{}c@{}l@{}c@{}r@{\hspace{1mm}}r@{}c@{}r@{\hspace{1mm}}r@{}c@{}r@{\hspace{1mm}}r@{}c@{}r@{}@{}}
\toprule%
\multicolumn{2}{@{}l@{}}{Model}
&\hspace*{2mm}&
\multicolumn{1}{@{}l@{}}{Algo-}
&\hspace*{2mm}&
\multicolumn{2}{@{}l@{}}{Learned}
&\hspace*{2mm}&
\multicolumn{2}{@{}l@{}}{Learn-}
&\hspace*{2mm}&
\multicolumn{2}{@{}l@{}}{Test-}
&\hspace*{2mm}&
\multicolumn{1}{@{}l@{}}{}
\\
\cline{1-2}
% no cline for column groups of only one column
\cline{6-7}
\cline{9-10}
\cline{12-13}
% no cline for column groups of only one column
\\[-4pt]
$n$
&
$k$
&
&
rithm
&
&
Seeds
&
States
&
&
Inputs
&
Resets
&
&
Inputs
&
Resets
&
&
Total
\\
\midrule
\multicolumn{2}{@{}l@{}}{pion\_ecdhe\_cert\_none}&&\moe{$*$}&&\bestvalue{30}&\bestvalue{66.0}&&23527.83&2876.40&&\bestvalue{820419.30}&\bestvalue{80167.83}&&\bestvalue{926991.37}\\
$n$ = 66&$k$ = 10&&baseline&&\bestvalue{30}&\bestvalue{66.0}&&\bestvalue{23463.43}&\bestvalue{2865.30}&&1083303.10&108173.27&&1217805.10\\
\arrayrulecolor{black!20}\hline\arrayrulecolor{black}
\multicolumn{2}{@{}l@{}}{pion\_ecdhe\_cert\_nreq}&&\moe{$*$}&&\bestvalue{30}&\bestvalue{66.0}&&23527.83&2876.40&&\bestvalue{820419.30}&\bestvalue{80167.83}&&\bestvalue{926991.37}\\
$n$ = 66&$k$ = 10&&baseline&&\bestvalue{30}&\bestvalue{66.0}&&\bestvalue{23463.43}&\bestvalue{2865.30}&&1083303.10&108173.27&&1217805.10\\
\arrayrulecolor{black!20}\hline\arrayrulecolor{black}
\multicolumn{2}{@{}l@{}}{pion\_ecdhe\_cert\_req}&&\moe{$*$}&&\bestvalue{30}&\bestvalue{66.0}&&\bestvalue{23525.47}&\bestvalue{2888.17}&&\bestvalue{891781.70}&\bestvalue{90189.37}&&\bestvalue{1008384.70}\\
\bottomrule
\end{tabular}
}
\end{adjustbox}
\caption{Experiment 3, $\ETS$ ($k = 2$), Part 2.}
\label{tab:exp3_fixed_2}
\end{table}
\begin{table}[t]
    \centering
    \begin{adjustbox}{width=\textwidth}
{\footnotesize

\begin{tabular}{@{}@{}l@{\hspace{1mm}}l@{}c@{}l@{}c@{}r@{\hspace{1mm}}r@{}c@{}r@{\hspace{1mm}}r@{}c@{}r@{\hspace{1mm}}r@{}c@{}r@{}@{}}
\toprule%
\multicolumn{2}{@{}l@{}}{Model}
&\hspace*{2mm}&
\multicolumn{1}{@{}l@{}}{Algo-}
&\hspace*{2mm}&
\multicolumn{2}{@{}l@{}}{Learned}
&\hspace*{2mm}&
\multicolumn{2}{@{}l@{}}{Learn-}
&\hspace*{2mm}&
\multicolumn{2}{@{}l@{}}{Test-}
&\hspace*{2mm}&
\multicolumn{1}{@{}l@{}}{}
\\
\cline{1-2}
% no cline for column groups of only one column
\cline{6-7}
\cline{9-10}
\cline{12-13}
% no cline for column groups of only one column
\\[-4pt]
$n$
&
$k$
&
&
rithm
&
&
Seeds
&
States
&
&
Inputs
&
Resets
&
&
Inputs
&
Resets
&
&
Total
\\
\midrule
\multicolumn{2}{@{}l@{}}{BitVise}&&\moe{$*$}&&\bestvalue{30}&\bestvalue{66.0}&&\bestvalue{30709.83}&\bestvalue{2966.50}&&174703.40&13884.57&&222264.30\\
$n$ = 66&$k$ = 13&&baseline&&\bestvalue{30}&\bestvalue{66.0}&&31358.73&2985.23&&\bestvalue{167043.70}&\bestvalue{13377.13}&&\bestvalue{214764.80}\\
\arrayrulecolor{black!20}\hline\arrayrulecolor{black}
\multicolumn{2}{@{}l@{}}{DropBear}&&\moe{$*$}&&\bestvalue{30}&\bestvalue{17.0}&&\bestvalue{3914.07}&636.33&&4114.27&477.20&&9141.87\\
$n$ = 17&$k$ = 13&&baseline&&\bestvalue{30}&\bestvalue{17.0}&&3927.20&\bestvalue{635.87}&&\bestvalue{3739.63}&\bestvalue{436.10}&&\bestvalue{8738.80}\\
\arrayrulecolor{black!20}\hline\arrayrulecolor{black}
\multicolumn{2}{@{}l@{}}{GnuTLS\_3.3.8\_client\_full}&&\moe{$*$}&&\bestvalue{30}&\bestvalue{15.0}&&1068.43&225.83&&\bestvalue{40706.87}&\bestvalue{5336.80}&&\bestvalue{47337.93}\\
\bottomrule
\end{tabular}
}
\end{adjustbox}
\caption{Experiment 3, Randomised $\ETS$, Part 1.}
\label{tab:exp3_inf_1}
\end{table}

\begin{table}[t]
    \centering
    \begin{adjustbox}{width=\textwidth}
{\footnotesize
\begin{tabular}
{@{}@{}l@{\hspace{1mm}}l@{}c@{}l@{}c@{}r@{\hspace{1mm}}r@{}c@{}r@{\hspace{1mm}}r@{}c@{}r@{\hspace{1mm}}r@{}c@{}r@{}@{}}
\toprule%
\multicolumn{2}{@{}l@{}}{Model}
&\hspace*{2mm}&
\multicolumn{1}{@{}l@{}}{Algo-}
&\hspace*{2mm}&
\multicolumn{2}{@{}l@{}}{Learned}
&\hspace*{2mm}&
\multicolumn{2}{@{}l@{}}{Learn-}
&\hspace*{2mm}&
\multicolumn{2}{@{}l@{}}{Test-}
&\hspace*{2mm}&
\multicolumn{1}{@{}l@{}}{}
\\
\cline{1-2}
% no cline for column groups of only one column
\cline{6-7}
\cline{9-10}
\cline{12-13}
% no cline for column groups of only one column
\\[-4pt]
$n$
&
$k$
&
&
rithm
&
&
Seeds
&
States
&
&
Inputs
&
Resets
&
&
Inputs
&
Resets
&
&
Total
\\
\midrule
\multicolumn{2}{@{}l@{}}{etinydtls\_psk}&&\moe{$*$}&&\bestvalue{30}&\bestvalue{22.0}&&3328.43&\bestvalue{429.93}&&\bestvalue{16386.40}&\bestvalue{1385.40}&&\bestvalue{21530.17}\\
$n$ = 22&$k$ = 7&&baseline&&\bestvalue{30}&\bestvalue{22.0}&&\bestvalue{3290.07}&430.83&&16971.90&1436.33&&22129.13\\
\arrayrulecolor{black!20}\hline\arrayrulecolor{black}
\multicolumn{2}{@{}l@{}}{gnutls-3.5.19\_psk\_rsa\_cert\_nreq}&&\moe{$*$}&&\bestvalue{30}&\bestvalue{29.0}&&10601.37&1452.83&&\bestvalue{35581.30}&\bestvalue{2934.63}&&\bestvalue{50570.13}\\
$n$ = 29&$k$ = 12&&baseline&&\bestvalue{30}&\bestvalue{29.0}&&\bestvalue{10321.77}&\bestvalue{1432.37}&&46110.97&3752.97&&61618.07\\
\arrayrulecolor{black!20}\hline\arrayrulecolor{black}
\multicolumn{2}{@{}l@{}}{gnutls-3.6.7\_all\_cert\_nreq}&&\moe{$*$}&&\bestvalue{30}&\bestvalue{19.0}&&5342.03&990.43&&257663.27&26649.80&&290645.53\\
\bottomrule
\end{tabular}
}
\end{adjustbox}
\caption{Experiment 3, Randomised $\ETS$, Part 2.}
\label{tab:exp3_inf_2}
\end{table}

\begin{table}[t]
    \centering
    \begin{adjustbox}{width=0.8\textwidth}
{\footnotesize
\begin{tabular}{@{}@{}l@{\hspace{1mm}}l@{}c@{}l@{}c@{}r@{\hspace{1mm}}r@{}c@{}r@{\hspace{1mm}}r@{}c@{}r@{\hspace{1mm}}r@{}c@{}r@{}@{}}
\toprule%
\multicolumn{2}{@{}l@{}}{Model}
&\hspace*{2mm}&
\multicolumn{1}{@{}l@{}}{Algo-}
&\hspace*{2mm}&
\multicolumn{2}{@{}l@{}}{Learned}
&\hspace*{2mm}&
\multicolumn{2}{@{}l@{}}{Learn-}
&\hspace*{2mm}&
\multicolumn{2}{@{}l@{}}{Test-}
&\hspace*{2mm}&
\multicolumn{1}{@{}l@{}}{}
\\
\cline{1-2}
% no cline for column groups of only one column
\cline{6-7}
\cline{9-10}
\cline{12-13}
% no cline for column groups of only one column
\\[-4pt]
$n$
&
$k$
&
&
rithm
&
&
Seeds
&
States
&
&
Inputs
&
Resets
&
&
Inputs
&
Resets
&
&
Total
\\
\midrule
\multicolumn{2}{@{}l@{}}{gen\_MM\_20\_10\_3}&&\moe{$*$}&&\bestvalue{30}&\bestvalue{20.0}&&\bestvalue{3061.63}&592.97&&\bestvalue{69.37}&\bestvalue{9.00}&&\bestvalue{3732.97}\\
$n$ = 20&$k$ = 11&&baseline&&\bestvalue{30}&\bestvalue{20.0}&&3096.10&\bestvalue{592.40}&&87.90&11.27&&3787.67\\
\arrayrulecolor{black!20}\hline\arrayrulecolor{black}
\multicolumn{2}{@{}l@{}}{gen\_MM\_20\_10\_5}&&\moe{$*$}&&\bestvalue{30}&\bestvalue{20.0}&&\bestvalue{2972.30}&\bestvalue{517.93}&&\bestvalue{74.80}&\bestvalue{10.10}&&\bestvalue{3575.13}\\
$n$ = 20&$k$ = 11&&baseline&&\bestvalue{30}&\bestvalue{20.0}&&3019.07&519.13&&95.97&12.57&&3646.73\\
\arrayrulecolor{black!20}\hline\arrayrulecolor{black}
\multicolumn{2}{@{}l@{}}{gen\_MM\_20\_10\_7}&&\moe{$*$}&&\bestvalue{30}&\bestvalue{20.0}&&2545.30&\bestvalue{450.33}&&\bestvalue{55.93}&\bestvalue{7.63}&&\bestvalue{3059.20}\\
\bottomrule
\end{tabular}
}
\end{adjustbox}
\caption{Experiment 3, Randomly Generated Mealy Machines.}
\label{tab:exp3_random}
\end{table}
\clearpage
  \section{Additional Figures} \label{sec:appD}
In Fig.~\ref{fig:res_exp2half}, we show the results for Experiment 2 with budget $5 \cdot 10^7$.

\begin{figure}[h]
    \centering
    \begin{tikzpicture}
\begin{axis}[
    width=12.0cm,height=4.0cm,
    xticklabels={m54, m41, m199, m196, m95, m158, m159, m76, m45, m164, m106, m34, m22, m135, m49, m172, m201, m185, m167, m55, m24, m27, m189},
    ytick={0,10,20,30,40,50},
    xtick pos=bottom,
    ytick pos=left,
    enlargelimits=0.05,
	xlabel=$\mathsf{ASML}$ models,
    ylabel=\scriptsize{Number of learned models},
    x tick label style={rotate=70,anchor=east},
	ybar interval=0.5,
]
\addplot[color=green!60!black,fill=green!60!black] 
	coordinates {(0,50)(1,50)(2,18)(3,21)(4,50)(5,42)(6,45)(7,0)(8,0)(9,49)(10,47)(11,38)(12,32)(13,49)(14,18)(15,41)(16,0)(17,0)(18,49)(19,0)(20,0)(21,0)(22,0)(23,0)};
\addplot[color=cyan,fill=cyan]
	coordinates {(0,50)(1,50)(2,16)(3,11)(4,31)(5,31)(6,26)(7,0)(8,0)(9,49)(10,0)(11,11)(12,2)(13,3)(14,2)(15,0)(16,0)(17,0)(18,0)(19,0)(20,0)(21,0)(22,0)(23,0)};
\legend{\moe{*}, baseline}
\end{axis}
\end{tikzpicture}
    \caption{Results Experiment 2 with budget $5 \cdot 10^7$}
    \label{fig:res_exp2half}
\end{figure}
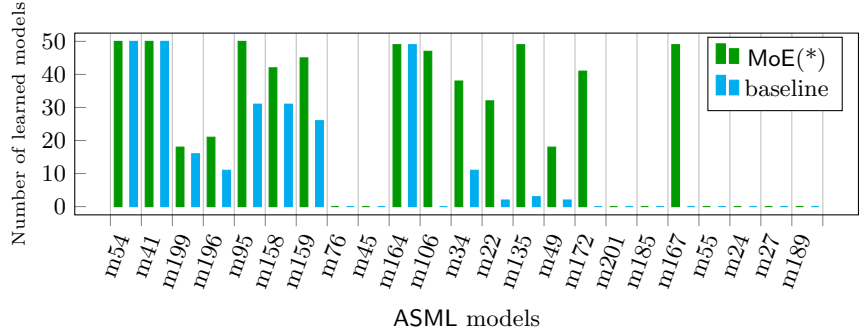

In Figs.~\ref{fig:res_exp1_sym}-\ref{fig:res4_resets}, we show the results for Experiments 1, 3 and 4 again but with separate plots for the total number of symbols and total number of resets.

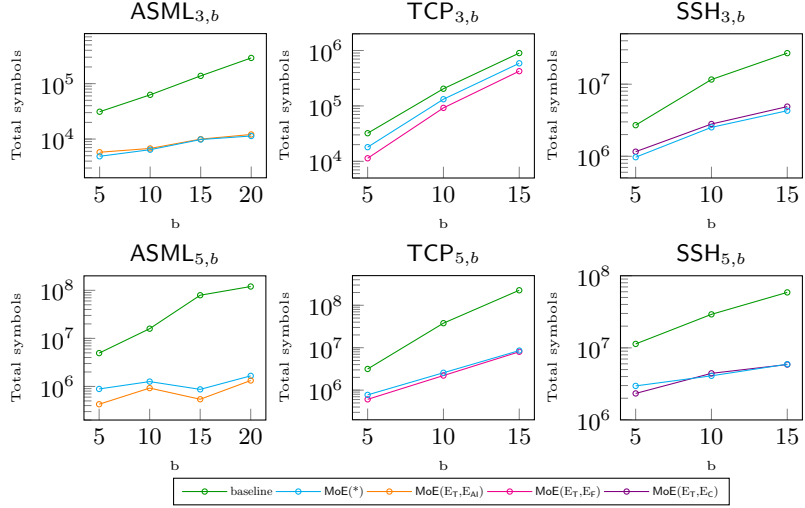
\begin{figure}[h]
    \centering
    \begin{tikzpicture}

\begin{axis}[
    legend style={nodes={scale=0.5, transform shape}},
    title style={yshift=-1.5ex,},
    title={$\mathsf{ASML}_{3,b}$},
    width=4cm,height=3.5cm,
    ymode=log,
    xlabel=\tiny{b},
    ylabel=\tiny{Total symbols},
    xtick={5, 10, 15, 20},
    ytick={1000, 10000, 100000, 1000000},
    xtick pos=bottom,
    ytick pos=left,
    ymin=2000, ymax=800000,
    name=plot1,
]

\addplot[
    color=green!60!black,
    mark=o,
    mark size=1pt
    ]
    coordinates {
        (5,31120.466666666667)(10,62831.63333333333)(15,138188.2)(20,291072.1666666667)
    };

\addplot[
    color=orange,
    mark=o,
    mark size=1pt
    ]
    coordinates {
        (5,5782.566666666667)(10,6768.733333333334)(15,9974.466666666667)(20,12094.766666666666)
    };

\addplot[
    color=cyan,
    mark=o,
    mark size=1pt
    ]
    coordinates {
        (5,4883.766666666666)(10,6428.633333333333)(15,9808.866666666667)(20,11360.933333333332)
    };

\end{axis}

\begin{axis}[
    at=(plot1.right of south east), anchor=left of south west,
    legend style={nodes={scale=0.5, transform shape}},
    title style={yshift=-1.5ex,},
    title={$\mathsf{TCP}_{3,b}$},
    width=4cm,height=3.5cm,
    ymode=log,
    xlabel=\tiny{b},
    ylabel=\tiny{Total symbols},
    xtick={5, 10, 15},
    ytick={1000, 10000, 100000, 1000000, 10000000},
    xtick pos=bottom,
    ytick pos=left,
    ymin=5000, ymax=2000000,
    legend pos=north west,
    name=plot2,
    ]

    \addplot[
    color=green!60!black,
    mark=o,
    mark size=1pt
    ]
    coordinates {
        (5,32156.0)(10,204155.9)(15,898425.9666666667)
    };

    \addplot[
        color=magenta,
        mark=o,
        mark size=1pt
        ]
        coordinates {
            (5,11351.466666666667)(10,92072.5)(15,423420.4)
        };

    \addplot[
        color=cyan,
        mark=o,
        mark size=1pt
        ]
        coordinates {
            (5,18018.633333333335)(10,131955.76666666666)(15,586999.6333333333)
        };

\end{axis}

\begin{axis}[
    at=(plot2.right of south east), anchor=left of south west,
    legend style={nodes={scale=0.5, transform shape}},
    title style={yshift=-1.5ex,},
    title={$\mathsf{SSH}_{3,b}$},
    width=4cm,height=3.5cm,
    ymode=log,
    xlabel=\tiny{b},
    ylabel=\tiny{Total symbols},
    xtick={5, 10, 15},
    ytick={100000, 1000000, 10000000, 100000000},
    xtick pos=bottom,
    ytick pos=left,
    ymin=500000, ymax=50000000,
    legend pos=north west,
    name=plot3,
    ]

    \addplot[
    color=green!60!black,
    mark=o,
    mark size=1pt
    ]
    coordinates {
        (5,2699592.3)(10,11606604.9)(15,26950059.6)
    };

    \addplot[
        color=violet,
        mark=o,
        mark size=1pt
        ]
        coordinates {
            (5,1158045.6)(10,2800175.966666667)(15,4883949.833333333)
        };

    \addplot[
        color=cyan,
        mark=o,
        mark size=1pt
        ]
        coordinates {
            (5,970604.7333333333)(10,2518161.066666667)(15,4284411.4)
        };

\end{axis}

\begin{axis}[
    at=(plot3.below south west), anchor=above north west,
    legend style={nodes={scale=0.5, transform shape}},
    title style={yshift=-1.5ex,},
    title={$\mathsf{SSH}_{5,b}$},
    width=4cm,height=3.5cm,
    ymode=log,
    xlabel=\tiny{b},
    ylabel=\tiny{Total symbols},
    xtick={5, 10, 15},
    ytick={100000, 1000000, 10000000, 100000000},
    xtick pos=bottom,
    ytick pos=left,
    ymin=1000000, ymax=100000000,
    legend pos=north west,
    name=plot4,
    ]

    \addplot[
    color=green!60!black,
    mark=o,
    mark size=1pt
    ]
    coordinates {
        (5,11334694.533333333)(10,29158639.666666668)(15,58781063.8)
    };

    \addplot[
        color=violet,
        mark=o,
        mark size=1pt
        ]
        coordinates {
            (5,2332324.8666666667)(10,4431812.266666667)(15,5868706.633333334)
        };

    \addplot[
        color=cyan,
        mark=o,
        mark size=1pt
        ]
        coordinates {
            (5,2965857.7)(10,4093180.466666667)(15,5931307.166666667)
        };

\end{axis}

\begin{axis}[
    at=(plot4.left of south west), anchor=right of south east,
    legend style={nodes={scale=0.5, transform shape}},
    title style={yshift=-1.5ex,},
    title={$\mathsf{TCP}_{5,b}$},
    width=4cm,height=3.5cm,
    ymode=log,
    xlabel=\tiny{b},
    ylabel=\tiny{Total symbols},
    xtick={5, 10, 15},
    ytick={100000, 1000000, 10000000, 100000000, 1000000000},
    xtick pos=bottom,
    ytick pos=left,
    ymin=200000, ymax=500000000,
    legend pos=north west,
    name=plot5,
    ]

    \addplot[
    color=green!60!black,
    mark=o,
    mark size=1pt
    ]
    coordinates {
        (5,3154885.8333333335)(10,37843659.46666667)(15,225843328.73333332)
    };

    \addplot[
        color=magenta,
        mark=o,
        mark size=1pt
        ]
        coordinates {
            (5,607241.5333333333)(10,2206862.466666667)(15,8015914.9)
        };

    \addplot[
        color=cyan,
        mark=o,
        mark size=1pt
        ]
        coordinates {
            (5,773839.5)(10,2569388.0)(15,8548868.533333333)
        };

\end{axis}

\begin{axis}[
    at=(plot5.left of south west),
    shift={(-71pt,0)},
    legend style={legend columns=-1, nodes={scale=0.5, transform shape}, at={(3.55, -0.4)}},
    title style={yshift=-1.5ex,},
    title={$\mathsf{ASML}_{5,b}$},
    width=4cm,height=3.5cm,
    ymode=log,
    xlabel=\tiny{b},
    ylabel=\tiny{Total symbols},
    xtick={5, 10, 15, 20},
    ytick={100000, 1000000, 10000000, 100000000},
    xtick pos=bottom,
    ytick pos=left,
    ymin=200000, ymax=200000000,
    name=plot6,
    ]

    \addplot[
    color=green!60!black,
    mark=o,
    mark size=1pt
    ]
    coordinates {
        (5,4935534.0)(10,15912023.033333333)(15,79032178.9)(20,119982587.23333333)
    };
    \addlegendentry{baseline}

    \addplot[
    color=cyan,
    mark=o,
    mark size=1pt
    ]
    coordinates {
        (5,886282.4)(10,1261706.8333333333)(15,869237.3666666667)(20,1660758.0666666667)
    };
    \addlegendentry{\moe{*}}

    \addplot[
        color=orange,
        mark=o,
        mark size=1pt
        ]
        coordinates {
            (5,426461.9666666667)(10,924480.3666666667)(15,541363.9666666667)(20,1328530.9333333333)
        };
    \addlegendentry{\moe{E$_{\mathsf{T}}$,E$_{\mathsf{AI}}$}}

    \addplot[
        color=magenta,
        mark=o,
        mark size=1pt
        ]
        coordinates {
            (2,2)(3,3)
        };
    \addlegendentry{\moe{E$_{\mathsf{T}}$,E$_{\mathsf{F}}$}}

    \addplot[
        color=violet,
        mark=o,
        mark size=1pt
        ]
        coordinates {
            (2,2)(3,3)
        };
    \addlegendentry{\moe{E$_{\mathsf{T}}$,E$_{\mathsf{C}}$}}

\end{axis}

\end{tikzpicture}
    \caption{Results Experiment 1 (Total number of symbols).}
    \label{fig:res_exp1_sym}
\end{figure}

\begin{figure}[h]
    \centering
    \begin{tikzpicture}

\begin{axis}[
    legend style={nodes={scale=0.5, transform shape}},
    title style={yshift=-1.5ex,},
    title={$\mathsf{ASML}_{3,b}$},
    width=4cm,height=3.5cm,
    ymode=log,
    xlabel=\tiny{b},
    ylabel=\tiny{Total resets},
    xtick={5, 10, 15, 20},
    ytick={100, 1000, 10000, 100000, 1000000},
    xtick pos=bottom,
    ytick pos=left,
    ymin=500, ymax=200000,
    name=plot1,
]

\addplot[
    color=green!60!black,
    mark=o,
    mark size=1pt
    ]
    coordinates {
        (5,8608.166666666666)(10,17897.366666666665)(15,38736.13333333333)(20,80809.56666666667)
    };

\addplot[
    color=orange,
    mark=o,
    mark size=1pt
    ]
    coordinates {
        (5,1825.0666666666666)(10,2596.3333333333335)(15,4016.733333333333)(20,5379.766666666666)
    };

\addplot[
    color=cyan,
    mark=o,
    mark size=1pt
    ]
    coordinates {
        (5,1583.0333333333333)(10,2488.4)(15,3946.4666666666667)(20,5175.466666666666)
    };

\end{axis}

\begin{axis}[
    at=(plot1.right of south east), anchor=left of south west,
    legend style={nodes={scale=0.5, transform shape}},
    title style={yshift=-1.5ex,},
    title={$\mathsf{TCP}_{3,b}$},
    width=4cm,height=3.5cm,
    ymode=log,
    xlabel=\tiny{b},
    ylabel=\tiny{Total resets},
    xtick={5, 10, 15},
    ytick={1000, 10000, 100000, 1000000, 10000000},
    xtick pos=bottom,
    ytick pos=left,
    ymin=1000, ymax=200000,
    legend pos=north west,
    name=plot2,
    ]

    \addplot[
    color=green!60!black,
    mark=o,
    mark size=1pt
    ]
    coordinates {
        (5,6793.066666666667)(10,31494.366666666665)(15,104910.7)
    };

    \addplot[
        color=magenta,
        mark=o,
        mark size=1pt
        ]
        coordinates {
            (5,2866.8)(10,19095.866666666665)(15,68952.2)
        };

    \addplot[
        color=cyan,
        mark=o,
        mark size=1pt
        ]
        coordinates {
            (5,4816.366666666667)(10,28257.4)(15,98838.4)
        };

\end{axis}

\begin{axis}[
    at=(plot2.right of south east), anchor=left of south west,
    legend style={nodes={scale=0.5, transform shape}},
    title style={yshift=-1.5ex,},
    title={$\mathsf{SSH}_{3,b}$},
    width=4cm,height=3.5cm,
    ymode=log,
    xlabel=\tiny{b},
    ylabel=\tiny{Total resets},
    xtick={5, 10, 15},
    ytick={1000, 10000, 100000, 1000000, 10000000, 100000000},
    xtick pos=bottom,
    ytick pos=left,
    ymin=50000, ymax=5000000,
    legend pos=north west,
    name=plot3,
    ]

    \addplot[
    color=green!60!black,
    mark=o,
    mark size=1pt
    ]
    coordinates {
        (5,285707.6)(10,1183709.6666666667)(15,2681880.3)
    };

    \addplot[
        color=violet,
        mark=o,
        mark size=1pt
        ]
        coordinates {
            (5,125051.33333333333)(10,300112.4)(15,503260.5333333333)
        };

    \addplot[
        color=cyan,
        mark=o,
        mark size=1pt
        ]
        coordinates {
            (5,106514.66666666667)(10,266140.43333333335)(15,445976.9)
        };

\end{axis}

\begin{axis}[
    at=(plot3.below south west), anchor=above north west,
    legend style={nodes={scale=0.5, transform shape}},
    title style={yshift=-1.5ex,},
    title={$\mathsf{SSH}_{5,b}$},
    width=4cm,height=3.5cm,
    ymode=log,
    xlabel=\tiny{b},
    ylabel=\tiny{Total resets},
    xtick={5, 10, 15},
    ytick={10000, 100000, 1000000, 10000000, 100000000},
    xtick pos=bottom,
    ytick pos=left,
    ymin=100000, ymax=10000000,
    legend pos=north west,
    name=plot4,
    ]

    \addplot[
    color=green!60!black,
    mark=o,
    mark size=1pt
    ]
    coordinates {
        (5,1182583.3)(10,2973960.1666666665)(15,5967076.6)
    };

    \addplot[
        color=violet,
        mark=o,
        mark size=1pt
        ]
        coordinates {
            (5,257726.93333333332)(10,478372.93333333335)(15,621614.6)
        };

    \addplot[
        color=cyan,
        mark=o,
        mark size=1pt
        ]
        coordinates {
            (5,324661.5)(10,446455.4666666667)(15,631585.9)
        };

\end{axis}

\begin{axis}[
    at=(plot4.left of south west), anchor=right of south east,
    legend style={nodes={scale=0.5, transform shape}},
    title style={yshift=-1.5ex,},
    title={$\mathsf{TCP}_{5,b}$},
    width=4cm,height=3.5cm,
    ymode=log,
    xlabel=\tiny{b},
    ylabel=\tiny{Total resets},
    xtick={5, 10, 15},
    ytick={10000, 100000, 1000000, 10000000, 100000000, 1000000000},
    xtick pos=bottom,
    ytick pos=left,
    ymin=20000, ymax=50000000,
    legend pos=north west,
    name=plot5,
    ]

    \addplot[
    color=green!60!black,
    mark=o,
    mark size=1pt
    ]
    coordinates {
        (5,523554.6666666667)(10,4145259.466666667)(15,19293531.433333334)
    };

    \addplot[
        color=magenta,
        mark=o,
        mark size=1pt
        ]
        coordinates {
            (5,96590.96666666666)(10,258515.2)(15,739790.6666666666)
        };

    \addplot[
        color=cyan,
        mark=o,
        mark size=1pt
        ]
        coordinates {
            (5,121388.9)(10,301997.1666666667)(15,784771.3666666667)
        };

\end{axis}

\begin{axis}[
    at=(plot5.left of south west),
    shift={(-71pt,0)},
    legend style={legend columns=-1, nodes={scale=0.5, transform shape}, at={(3.55, -0.4)}},
    title style={yshift=-1.5ex,},
    title={$\mathsf{ASML}_{5,b}$},
    width=4cm,height=3.5cm,
    ymode=log,
    xlabel=\tiny{b},
    ylabel=\tiny{Total resets},
    xtick={5, 10, 15, 20},
    ytick={10000, 100000, 1000000, 10000000, 100000000},
    xtick pos=bottom,
    ytick pos=left,
    ymin=20000, ymax=40000000,
    name=plot6,
    ]

    \addplot[
    color=green!60!black,
    mark=o,
    mark size=1pt
    ]
    coordinates {
        (5,695389.6)(10,2248927.9)(15,11167318.333333334)(20,16971138.266666666)
    };
    \addlegendentry{baseline}

    \addplot[
    color=cyan,
    mark=o,
    mark size=1pt
    ]
    coordinates {
        (5,124820.26666666666)(10,178694.23333333334)(15,126298.33333333333)(20,236894.5)
    };
    \addlegendentry{\moe{*}}

    \addplot[
        color=orange,
        mark=o,
        mark size=1pt
        ]
        coordinates {
            (5,60842.166666666664)(10,132177.7)(15,81490.9)(20,195448.63333333333)
        };
    \addlegendentry{\moe{E$_{\mathsf{T}}$,E$_{\mathsf{AI}}$}}

    \addplot[
        color=magenta,
        mark=o,
        mark size=1pt
        ]
        coordinates {
            (2,2)(3,3)
        };
    \addlegendentry{\moe{E$_{\mathsf{T}}$,E$_{\mathsf{F}}$}}

    \addplot[
        color=violet,
        mark=o,
        mark size=1pt
        ]
        coordinates {
            (2,2)(3,3)
        };
    \addlegendentry{\moe{E$_{\mathsf{T}}$,E$_{\mathsf{C}}$}}

\end{axis}

\end{tikzpicture}
    \caption{Results Experiment 1 (Total number of resets).}
    \label{fig:res_exp1_resets}
\end{figure}

\begin{figure}[h]
    \centering
    \begin{subfigure}[b]{0.3\textwidth}
        \begin{tikzpicture}
    \begin{axis}[
        legend style={nodes={scale=0.5, transform shape}},
        width=4cm,height=4cm,
        xmode=log,
        ymode=log,
        xlabel=\tiny{Total symbols baseline},
        ylabel=\tiny{Total symbols \moe{*}},
        xtick={10000, 100000, 1000000, 10000000},
        ytick={10000, 100000, 1000000, 10000000},
        xtick pos=bottom,
        ytick pos=left,
        xmin=5000, xmax=50000000,
        ymin=5000, ymax=50000000,
        legend pos=north west,
    ]

    \addplot[
        only marks,
        color=magenta,
        mark=o,
        mark size=1pt
        ]
        coordinates {
            (2905447.7333333334,1799503.0666666667)(489514.7,432338.1666666667)(5069243.966666667,4485182.1)
        };
        \addlegendentry{SSH}

    \addplot[
        only marks,
        color=green!60!black,
        mark=o,
        mark size=1pt
        ]
        coordinates {
            (4268748.733333333,3261608.1)(52100.53333333333,31421.266666666666)(2992687.7333333334,2767105.7666666666)
        };
        \addlegendentry{TCP}

    \addplot[
        only marks,
        color=violet,
        mark=o,
        mark size=1pt
        ]
        coordinates {
            (53524.7,35015.1)(33011.7,11219.733333333334)
        };
    \addlegendentry{TLS}

    \addplot[
        only marks,
        color=cyan,
        mark=o,
        mark size=1pt
        ]
        coordinates {
            (486747.3333333333,384260.13333333336)(1061297.7333333334,888813.0333333333)(132592.16666666666,91403.8)(229461.3,245982.1)(526978.3666666667,448320.4)(62928.23333333333,54468.4)(1084531.0,812485.8333333334)(1380266.3666666667,1105133.9)(703937.4333333333,438601.5333333333)(3719475.3333333335,3155057.1)(4287866.166666667,3731477.8333333335)(5302385.833333333,4797545.833333333)(586697.3666666667,467017.3333333333)(398495.93333333335,321694.5)(1845612.1333333333,1656538.9)(794251.7666666667,668540.5)(1106766.5333333334,843947.1333333333)(1106766.5333333334,843947.1333333333)(1232244.6333333333,915307.1666666666)(658191.1,455862.9666666667)(1104944.7666666666,759375.2)(123487.53333333334,89098.0)(232036.06666666668,139664.56666666668)(562558.9333333333,496423.8)
        };
    \addlegendentry{DTLS}

    \addplot[
        only marks,
        color=orange,
        mark=o,
        mark size=1pt
        ]
        coordinates {
            (643360.2666666667,488135.4)(5420750.766666667,3743541.9)
        };
    \addlegendentry{Philips}

    \addplot[color=gray]
        coordinates {(5000,5000)(50000000,50000000) };

    \addplot[color=gray,style=dashed]
        coordinates {(5000,5000/2)(50000000,50000000/2) };

    \addplot[color=gray,style=dashed]
        coordinates {(5000,5000*2)(50000000,50000000*2) };

    \end{axis}
    \end{tikzpicture}
        \caption{$\mathsf{ETS}, k=2$}
        \label{fig:res3a_sym}
    \end{subfigure}
    \begin{subfigure}[b]{0.3\textwidth}
        \begin{tikzpicture}
    \begin{axis}[
        legend style={nodes={scale=0.5, transform shape}},
        width=4cm,height=4cm,
        xmode=log,
        ymode=log,
        xlabel=\tiny{Total symbols baseline},
        ylabel=\tiny{Total symbols \moe{*}},
        xtick pos=bottom,
        ytick pos=left,
        xmin=1000, xmax=3000000,
        ymin=1000, ymax=3000000,
        legend pos=north west,
    ]

    \addplot[
        only marks,
        color=magenta,
        mark=o,
        mark size=1pt
        ]
        coordinates {
            (198402.43333333332,205413.23333333334)(7666.833333333333,8028.333333333333)(1106376.6,593576.4666666667)
        };
        \addlegendentry{SSH}

    \addplot[
        only marks,
        color=green!60!black,
        mark=o,
        mark size=1pt
        ]
        coordinates {
            (137651.06666666668,126764.06666666667)(23130.733333333334,16163.866666666667)(130165.33333333333,110151.23333333334)
        };
        \addlegendentry{TCP}

    \addplot[
        only marks,
        color=violet,
        mark=o,
        mark size=1pt
        ]
        coordinates {
            (136682.93333333332,41775.3)(5720.033333333334,6558.166666666667)
        };
    \addlegendentry{TLS}

    \addplot[
        only marks,
        color=cyan,
        mark=o,
        mark size=1pt
        ]
        coordinates {
            (21210.8,21330.8)(41398.666666666664,44416.36666666667)(14359.933333333332,13854.466666666667)(60757.9,57711.333333333336)(54510.03333333333,70173.03333333334)(20261.966666666667,19714.833333333332)(56432.73333333333,46182.666666666664)(180682.63333333333,263005.3)(256328.23333333334,242053.1)(423644.36666666664,353622.9666666667)(737008.9333333333,685067.8333333334)(961007.3666666667,835343.1333333333)(32118.133333333335,39084.76666666667)(88792.9,65315.6)(128170.83333333333,118470.36666666667)(172864.5,158276.36666666667)(329224.06666666665,122988.2)(329222.76666666666,122987.2)(237137.0,121953.46666666666)(96305.86666666667,87170.33333333333)(104356.43333333333,80638.6)(16937.466666666667,10818.666666666666)(12927.2,14752.533333333333)(976405.4333333333,1744136.8333333333)
        };
    \addlegendentry{DTLS}

    \addplot[
        only marks,
        color=orange,
        mark=o,
        mark size=1pt
        ]
        coordinates {
            (51191.23333333333,48580.53333333333)(72314.96666666666,53315.13333333333)
        };
    \addlegendentry{Philips}

    \addplot[color=gray]
        coordinates {(1000,1000)(3000000,3000000) };

    \addplot[color=gray,style=dashed]
        coordinates {(1000,1000/2)(3000000,3000000/2) };

    \addplot[color=gray,style=dashed]
        coordinates {(1000,1000*2)(3000000,3000000*2) };

    \end{axis}
    \end{tikzpicture}
        \caption{Randomised $\mathsf{ETS}$}
        \label{fig:res3b_sym}
    \end{subfigure}
    \begin{subfigure}[b]{0.3\textwidth}
        \begin{tikzpicture}
    \begin{axis}[
        legend style={nodes={scale=0.5, transform shape}},
        width=4cm,height=4cm,
        xmode=log,
        ymode=log,
        xlabel=\tiny{Total symbols baseline},
        ylabel=\tiny{Total symbols \moe{*}},
        xtick pos=bottom,
        ytick pos=left,
        xmin=1000, xmax=100000,
        ymin=1000, ymax=100000,
        legend pos=north west,
    ]

        \addplot[
        only marks,
        color=cyan,
        mark=o,
        mark size=1pt
        ]
        coordinates {
            (3184.0,3131.0)(3115.0333333333333,3047.1)(2625.5333333333333,2601.233333333333)(5952.0,5911.333333333333)(5187.733333333334,5202.066666666667)(5315.7,5355.133333333333)(9536.5,9661.833333333334)(8393.1,8458.366666666667)(7820.166666666667,7690.966666666666)(9761.566666666668,9613.533333333333)(8371.666666666666,8410.633333333333)(7810.633333333333,7886.866666666667)(17362.566666666666,17392.5)(15970.933333333332,16153.2)(16143.933333333332,16029.7)(26449.6,26059.8)(24740.133333333335,25525.5)(25323.233333333334,25351.7)(17720.666666666668,17696.7)(14180.133333333333,14274.7)(14761.6,14703.766666666666)(31682.233333333334,31275.9)(32070.566666666666,32062.466666666667)(26546.3,26578.1)(46941.833333333336,47158.833333333336)(40691.4,41542.3)(39364.23333333333,38978.833333333336)
        };

    \addplot[color=gray]
        coordinates {(1000,1000)(100000,100000) };

    \addplot[color=gray,style=dashed]
        coordinates {(1000,1000/2)(100000,100000/2) };

    \addplot[color=gray,style=dashed]
        coordinates {(1000,1000*2)(100000,100000*2) };

    \end{axis}
    \end{tikzpicture}
        \caption{Randomly generated}
        \label{fig:res3c_sym}
    \end{subfigure}
    \caption{Results Experiment 3 (Total number of symbols).}
    \label{fig:res3_sym}
\end{figure}

\begin{figure}[h]
    \centering
    \begin{subfigure}[b]{0.3\textwidth}
        \begin{tikzpicture}
    \begin{axis}[
        legend style={nodes={scale=0.5, transform shape}},
        width=4cm,height=4cm,
        xmode=log,
        ymode=log,
        xlabel=\tiny{Total resets baseline},
        ylabel=\tiny{Total resets \moe{*}},
        xtick={1000, 10000, 100000, 1000000},
        ytick={1000, 10000, 100000, 1000000},
        xtick pos=bottom,
        ytick pos=left,
        xmin=500, xmax=5000000,
        ymin=500, ymax=5000000,
        legend pos=north west,
    ]

    \addplot[
        only marks,
        color=magenta,
        mark=o,
        mark size=1pt
        ]
        coordinates {
            (254094.83333333334,165007.93333333332)(73676.83333333333,66312.4)(504056.56666666665,447443.1)
        };
        \addlegendentry{SSH}

    \addplot[
        only marks,
        color=green!60!black,
        mark=o,
        mark size=1pt
        ]
        coordinates {
            (328587.3333333333,255299.63333333333)(6337.633333333333,3942.6)(232023.9,215816.76666666666)
        };
        \addlegendentry{TCP}

    \addplot[
        only marks,
        color=violet,
        mark=o,
        mark size=1pt
        ]
        coordinates {
            (9473.666666666666,6510.366666666667)(4636.833333333333,1778.5)
        };
    \addlegendentry{TLS}

    \addplot[
        only marks,
        color=cyan,
        mark=o,
        mark size=1pt
        ]
        coordinates {
            (47519.63333333333,37959.933333333334)(92853.7,78029.4)(13061.0,9433.0)(23415.966666666667,24968.9)(45988.53333333333,38185.73333333333)(7234.166666666667,6243.6)(102960.26666666666,78748.1)(174330.6,143548.4)(91701.33333333333,62583.86666666667)(340635.6,289914.8)(406713.06666666665,353557.73333333334)(507184.6666666667,456457.4)(73080.96666666666,59513.53333333333)(50826.2,40064.26666666667)(234724.33333333334,215407.93333333332)(107074.16666666667,91603.46666666666)(111038.56666666667,83044.23333333334)(111038.56666666667,83044.23333333334)(126242.46666666666,93077.53333333334)(67533.26666666666,47647.0)(109309.26666666666,77606.8)(13921.366666666667,10656.2)(28200.566666666666,19757.433333333334)(84314.0,74777.16666666667)
        };
    \addlegendentry{DTLS}

    \addplot[
        only marks,
        color=orange,
        mark=o,
        mark size=1pt
        ]
        coordinates {
            (78492.06666666667,60539.1)(477742.7,325060.8333333333)
        };
    \addlegendentry{Philips}

    \addplot[color=gray]
        coordinates {(500,500)(5000000,5000000) };

    \addplot[color=gray,style=dashed]
        coordinates {(500,500/2)(5000000,5000000/2) };

    \addplot[color=gray,style=dashed]
        coordinates {(500,500*2)(5000000,5000000*2) };

    \end{axis}
    \end{tikzpicture}
        \caption{$\mathsf{ETS}, k=2$}
        \label{fig:res3a_resets}
    \end{subfigure}
    \begin{subfigure}[b]{0.3\textwidth}
        \begin{tikzpicture}
    \begin{axis}[
        legend style={nodes={scale=0.5, transform shape}},
        width=4cm,height=4cm,
        xmode=log,
        ymode=log,
        xlabel=\tiny{Total resets baseline},
        ylabel=\tiny{Total resets \moe{*}},
        xtick pos=bottom,
        ytick pos=left,
        xmin=100, xmax=300000,
        ymin=100, ymax=300000,
        legend pos=north west,
    ]

    \addplot[
        only marks,
        color=magenta,
        mark=o,
        mark size=1pt
        ]
        coordinates {
            (16362.366666666667,16851.066666666666)(1071.9666666666667,1113.5333333333333)(105585.93333333333,57915.96666666667)
        };
        \addlegendentry{SSH}

    \addplot[
        only marks,
        color=green!60!black,
        mark=o,
        mark size=1pt
        ]
        coordinates {
            (10437.233333333334,9792.6)(2513.5666666666666,1853.5)(9876.7,8554.433333333332)
        };
        \addlegendentry{TCP}

    \addplot[
        only marks,
        color=violet,
        mark=o,
        mark size=1pt
        ]
        coordinates {
            (17674.433333333334,5562.633333333333)(830.6333333333333,929.8666666666667)
        };
    \addlegendentry{TLS}

    \addplot[
        only marks,
        color=cyan,
        mark=o,
        mark size=1pt
        ]
        coordinates {
            (2092.9666666666667,2104.733333333333)(3433.4,3652.3)(1423.8333333333333,1378.0333333333333)(5143.133333333333,4884.933333333333)(4204.833333333333,5244.566666666667)(1867.1666666666667,1815.3333333333333)(5185.333333333333,4387.466666666666)(19151.733333333334,27640.233333333334)(26367.266666666666,25332.066666666666)(41671.666666666664,34032.666666666664)(71454.26666666666,64594.2)(93763.96666666666,79878.33333333333)(3552.766666666667,4165.033333333334)(9114.9,6854.466666666666)(13774.966666666667,12724.1)(18359.866666666665,16994.266666666666)(30850.7,11787.633333333333)(30850.633333333335,11787.6)(22734.466666666667,11735.333333333334)(9115.333333333334,8364.166666666666)(9595.533333333333,7401.8)(1816.4666666666667,1206.4666666666667)(1433.5666666666666,1561.0666666666666)(111037.93333333333,196073.66666666666)
        };
    \addlegendentry{DTLS}

    \addplot[
        only marks,
        color=orange,
        mark=o,
        mark size=1pt
        ]
        coordinates {
            (5174.5,4892.733333333334)(7052.066666666667,5544.2)
        };
    \addlegendentry{Philips}

    \addplot[color=gray]
        coordinates {(100,100)(300000,300000) };

    \addplot[color=gray,style=dashed]
        coordinates {(100,100/2)(300000,300000/2) };

    \addplot[color=gray,style=dashed]
        coordinates {(100,100*2)(300000,300000*2) };

    \end{axis}
    \end{tikzpicture}
        \caption{Randomised $\mathsf{ETS}$}
        \label{fig:res3b_resets}
    \end{subfigure}
    \begin{subfigure}[b]{0.3\textwidth}
        \begin{tikzpicture}
    \begin{axis}[
        legend style={nodes={scale=0.5, transform shape}},
        width=4cm,height=4cm,
        xmode=log,
        ymode=log,
        xlabel=\tiny{Total resets baseline},
        ylabel=\tiny{Total resets \moe{*}},
        xtick pos=bottom,
        ytick pos=left,
        xmin=100, xmax=20000,
        ymin=100, ymax=20000,
        legend pos=north west,
    ]

    \addplot[
        only marks,
        color=cyan,
        mark=o,
        mark size=1pt
        ]
        coordinates {
            (603.6666666666666,601.9666666666667)(531.7,528.0333333333333)(462.0,457.96666666666664)(1153.1333333333334,1153.2333333333333)(953.8333333333334,952.1333333333333)(904.6,910.3333333333334)(1869.8666666666666,1867.3)(1487.2,1482.3666666666666)(1412.9,1414.1666666666667)(1528.5,1518.8333333333333)(1167.1333333333334,1163.4666666666667)(1104.5666666666666,1106.7)(2834.366666666667,2830.1)(2382.3,2384.0666666666666)(1959.2666666666667,1956.9)(3951.1,3948.9333333333334)(3632.8333333333335,3635.8)(3092.2,3095.2)(2524.133333333333,2532.733333333333)(1967.9666666666667,1969.0666666666666)(1784.5666666666666,1787.1333333333334)(4280.033333333334,4279.133333333333)(3619.7,3631.8)(3232.3,3234.8)(6527.5,6524.533333333334)(5465.166666666667,5463.8)(4962.933333333333,4961.233333333334)
        };

    \addplot[color=gray]
        coordinates {(100,100)(20000,20000) };

    \addplot[color=gray,style=dashed]
        coordinates {(100,100/2)(20000,20000/2) };

    \addplot[color=gray,style=dashed]
        coordinates {(100,100*2)(20000,20000*2) };

    \end{axis}
    \end{tikzpicture}
        \caption{Randomly generated}
        \label{fig:res3c_resets}
    \end{subfigure}
    \caption{Results Experiment 3 (Total number of resets).}
    \label{fig:res3_resets}
\end{figure}

\begin{figure}[h]
    \centering
    \begin{subfigure}[b]{0.45\textwidth}
        \begin{tikzpicture}
\begin{axis}[
legend style={nodes={scale=0.5, transform shape}},width=9.0cm,height=6.0cm,ymode=log,ytick pos=left,xtick pos=bottom,xlabel={Hypothesis size},ylabel={Total symbols},xmin=1, xmax=30,ymin=100, ymax=100000000,xtick={0,10,20,30},ytick={100, 10000, 1000000, 100000000},legend pos=south east,scale=0.6,vasymptote=5,]

% m159 full symb 1
\addplot[color=cyan,mark=o,mark size=0.5pt,]
coordinates{
(1,120)(2,374)(3,12142)(5,35280)(7,48506)(10,175093)(13,413212)(17,1078804)(23,1113418)(27,1383567)(27,86348081)};
\addlegendentry{baseline}

% m159 all symb 1
\addplot[color=green!60!black]
coordinates{
(1,120)(2,374)(3,21142)(5,38443)(7,45678)(11,484665)(14,527874)(17,602844)(23,640028)(27,704076)(30,19922812)};
\addlegendentry{\moe{*}}

% m159 all symb 1, expert 0
\addplot[only marks,color=green!60!black,mark=o,mark size=0.5pt,mark options={fill=green!60!black},]
coordinates{
(1,120)(2,374)(3,21142)(5,38443)(7,45678)(11,484665)};
\addlegendentry{E$_{\mathsf{T}}$}

% m159 all symb 1, expert 1
\addplot[only marks,color=green!60!black,mark=triangle,mark size=0.5pt,mark options={fill=green!60!black},]
coordinates{
(30,19922812)};
\addlegendentry{E$_{\mathsf{AI}}$}

% m159 all symb 1, expert 2
\addplot[only marks,color=green!60!black,mark=square,mark size=0.5pt,mark options={fill=green!60!black},]
coordinates{
(14,527874)(17,602844)(23,640028)(27,704076)};
\addlegendentry{E$_{\mathsf{F}}$}

% m159 all symb 1, expert 3
\addplot[only marks,color=green!60!black,mark=diamond,mark size=0.5pt,mark options={fill=green!60!black},]
coordinates{
};
% \addlegendentry{E$_{\mathsf{C}}$}

% m159 full symb 2
\addplot[color=cyan,mark=o,mark size=0.5pt,]
coordinates{
(1,120)(2,382)(3,20367)(4,28777)(5,138008)(7,164714)(11,294346)(12,356763)(17,410421)(27,480943)(30,498130)};

% m159 full symb 3
\addplot[color=cyan,mark=o,mark size=0.5pt,]
coordinates{
(1,120)(2,379)(3,9149)(4,23487)(5,29986)(7,37523)(10,56477)(13,71833)(14,526639)(23,696085)(23,84956440)};

% m159 all symb 2
\addplot[color=green!60!black]
coordinates{
(1,120)(2,382)(3,20367)(4,28777)(5,138008)(7,152513)(11,203732)(13,227776)(17,280737)(20,300028)(26,335893)(30,419281)};

% m159 all symb 3
\addplot[color=green!60!black]
coordinates{
(1,120)(2,379)(3,9149)(4,23487)(5,29986)(7,44269)(11,288437)(13,296335)(17,402562)(20,416771)(26,454449)(30,524178)};

% m159 all symb 2, expert 0
\addplot[only marks,color=green!60!black,mark=o,mark size=0.5pt,mark options={fill=green!60!black},]
coordinates{
(1,120)(2,382)(3,20367)(4,28777)(5,138008)(11,203732)};

% m159 all symb 2, expert 1
\addplot[only marks,color=green!60!black,mark=triangle,mark size=0.5pt,mark options={fill=green!60!black},]
coordinates{
(7,152513)(20,300028)};

% m159 all symb 2, expert 2
\addplot[only marks,color=green!60!black,mark=square,mark size=0.5pt,mark options={fill=green!60!black},]
coordinates{
(13,227776)(17,280737)(26,335893)(30,419281)};

% m159 all symb 2, expert 3
\addplot[only marks,color=green!60!black,mark=diamond,mark size=0.5pt,mark options={fill=green!60!black},]
coordinates{
};

% m159 all symb 3, expert 0
\addplot[only marks,color=green!60!black,mark=o,mark size=0.5pt,mark options={fill=green!60!black},]
coordinates{
(1,120)(2,379)(3,9149)(4,23487)(5,29986)(7,44269)(11,288437)};

% m159 all symb 3, expert 1
\addplot[only marks,color=green!60!black,mark=triangle,mark size=0.5pt,mark options={fill=green!60!black},]
coordinates{
(13,296335)};

% m159 all symb 3, expert 2
\addplot[only marks,color=green!60!black,mark=square,mark size=0.5pt,mark options={fill=green!60!black},]
coordinates{
(17,402562)(20,416771)(26,454449)(30,524178)};

% m159 all symb 3, expert 3
\addplot[only marks,color=green!60!black,mark=diamond,mark size=0.5pt,mark options={fill=green!60!black},]
coordinates{
};

\end{axis}
\end{tikzpicture}
    \end{subfigure}
    \begin{subfigure}[b]{0.45\textwidth}
        \begin{tikzpicture}
\begin{axis}[
legend style={nodes={scale=0.5, transform shape}},width=9.0cm,height=6.0cm,ymode=log,ytick pos=left,xtick pos=bottom,xlabel={Hypothesis size},ylabel={Total symbols},xmin=1, xmax=189,ymin=100, ymax=100000000,xtick={0,50,100,150,189},ytick={100, 10000, 1000000, 100000000},legend pos=south east,scale=0.6,vasymptote=5,]

% m189 full symb 1
\addplot[color=cyan,mark=o,mark size=0.5pt,]
coordinates{
(1,138)(2,443)(3,2417)(4,16784)(5,20617)(6,53875)(8,186283)(9,392413)(11,428967)(17,518642)(19,537419)(20,548120)(21,615169)(26,700454)(29,855967)(30,1463398)(36,2060935)(37,3017104)(39,3376122)(40,3424246)(42,3459952)(44,4108339)(45,4209767)(46,4482444)(49,5165907)(53,7147485)(54,8301242)(55,10493512)(56,11820891)(57,11924447)(63,11998258)(64,12073998)(69,12159492)(71,15913061)(74,17908530)(77,18082715)(81,19262371)(82,26537400)(92,28791209)(93,29194616)(95,29921151)(102,30449726)(105,31752049)(106,31878539)(107,32132865)(109,38646796)(110,38696202)(118,42294019)(119,48657404)(120,75529282)(122,83598478)(123,84990290)(123,90956064)};
\addlegendentry{baseline}

% m189 all symb 1
\addplot[color=green!60!black]
coordinates{
(1,138)(2,443)(3,2417)(4,16784)(5,20617)(6,106532)(8,114801)(9,122809)(10,396826)(12,655764)(19,715168)(22,748872)(23,759854)(27,803866)(31,837312)(32,846443)(36,910577)(37,928680)(38,941300)(39,955096)(40,967697)(41,996915)(42,1108155)(43,1244923)(46,1348710)(47,1720297)(48,1734382)(59,2052122)(60,2074405)(63,2126174)(65,2187733)(66,2211218)(69,2339647)(73,2449164)(74,2664974)(75,2706187)(77,4641901)(82,4689013)(85,6162954)(90,8797578)(99,9049647)(101,9136880)(102,9194049)(103,9249050)(104,9408905)(107,10383513)(114,10630848)(115,10752070)(117,11251408)(126,11590813)(129,11649041)(130,12563938)(131,13565477)(133,13604507)(135,36860292)(146,37240083)(148,37337193)(149,37729193)(151,40949950)(153,41435097)(157,45846473)(158,45900850)(160,46209837)(161,80885619)(165,80965837)(166,80985931)(166,92342337)};
\addlegendentry{\moe{*}}

% m189 all symb 1, expert 0
\addplot[only marks,color=green!60!black,mark=o,mark size=0.5pt,mark options={fill=green!60!black},]
coordinates{
(1,138)(2,443)(3,2417)(4,16784)(5,20617)(6,106532)(8,114801)(12,655764)(19,715168)(82,4689013)(99,9049647)(146,37240083)};
\addlegendentry{E$_{\mathsf{T}}$}

% m189 all symb 1, expert 1
\addplot[only marks,color=green!60!black,mark=triangle,mark size=0.5pt,mark options={fill=green!60!black},]
coordinates{
(38,941300)(48,1734382)(104,9408905)(135,36860292)};
\addlegendentry{E$_{\mathsf{AI}}$}

% m189 all symb 1, expert 2
\addplot[only marks,color=green!60!black,mark=square,mark size=0.5pt,mark options={fill=green!60!black},]
coordinates{
(9,122809)(10,396826)(22,748872)(23,759854)(27,803866)(31,837312)(32,846443)(36,910577)(37,928680)(39,955096)(40,967697)(41,996915)(42,1108155)(43,1244923)(46,1348710)(47,1720297)(59,2052122)(60,2074405)(63,2126174)(65,2187733)(66,2211218)(73,2449164)(74,2664974)(75,2706187)(77,4641901)(85,6162954)(90,8797578)(101,9136880)(102,9194049)(103,9249050)(107,10383513)(114,10630848)(115,10752070)(117,11251408)(126,11590813)(130,12563938)(131,13565477)(133,13604507)(148,37337193)(149,37729193)(151,40949950)(153,41435097)(157,45846473)(158,45900850)(160,46209837)(161,80885619)(165,80965837)(166,80985931)(166,92342337)};
\addlegendentry{E$_{\mathsf{F}}$}

% m189 all symb 1, expert 3
\addplot[only marks,color=green!60!black,mark=diamond,mark size=0.5pt,mark options={fill=green!60!black},]
coordinates{
(69,2339647)(129,11649041)};
\addlegendentry{E$_{\mathsf{C}}$}

% m189 full symb 2
\addplot[color=cyan,mark=o,mark size=0.5pt,]
coordinates{
(1,138)(2,453)(3,1659)(4,8488)(5,16232)(6,20504)(7,61161)(9,137499)(10,218001)(12,313635)(15,392373)(20,432162)(22,455293)(24,1071217)(28,1103036)(34,1217154)(35,1236381)(36,1247442)(40,1378677)(41,1478802)(46,1999649)(50,3749140)(55,4196064)(56,4469308)(57,4699052)(61,6402741)(62,6742527)(63,7273305)(64,7710437)(65,11298417)(68,13809998)(78,14045130)(79,14173367)(82,14343257)(84,14418853)(85,14565883)(86,15933590)(93,16250626)(96,16318035)(97,17404603)(99,18524284)(100,18706788)(101,24835907)(102,28537987)(104,31313615)(111,31935221)(113,34661099)(114,36655241)(115,40965721)(119,42791696)(120,42861667)(122,43839090)(123,48711798)(124,50830408)(124,90957406)};

% m189 full symb 3
\addplot[color=cyan,mark=o,mark size=0.5pt,]
coordinates{
(1,138)(2,449)(3,4852)(4,76463)(5,82207)(6,168663)(8,254441)(14,318947)(15,329211)(17,351999)(18,362911)(19,442297)(20,523869)(23,906560)(26,1002702)(30,1055074)(32,1644977)(34,2014750)(36,2194607)(37,2255579)(38,3069171)(41,4111768)(47,4578399)(48,7235905)(52,7487732)(53,7799796)(55,7948190)(56,8371372)(59,20765294)(61,20789746)(62,21814467)(63,25901982)(76,34822167)(77,35441142)(78,35475410)(79,35509487)(83,36148697)(84,36163629)(85,36180646)(92,36441855)(95,36521120)(98,37602882)(99,38947914)(100,41000846)(107,41333677)(108,42148373)(110,42561302)(112,42616361)(114,42909919)(116,47705771)(117,76956523)(118,77069201)(119,90288679)(119,91011313)};

% m189 all symb 2
\addplot[color=green!60!black]
coordinates{
(1,138)(2,453)(3,1659)(4,8488)(5,16232)(6,85478)(7,262674)(9,274235)(10,284158)(12,488930)(16,519545)(17,527481)(19,1537881)(23,1569704)(30,1655846)(34,1718923)(35,1725237)(36,1740349)(38,1771164)(39,1791447)(43,1856382)(47,1905489)(49,1934204)(52,1978825)(53,2017032)(56,2070742)(63,2231124)(68,2348591)(69,2387459)(70,2442200)(71,2471569)(74,2691110)(75,2733764)(80,2924375)(85,3056351)(86,3267392)(90,3368579)(91,3402142)(92,3445895)(96,3540745)(97,3558840)(98,3643740)(101,3673769)(102,3765206)(103,4120296)(106,4609542)(109,4772021)(113,4918161)(115,5026355)(116,5108894)(121,6483119)(122,7018504)(124,10563297)(134,27412604)(135,27439824)(147,27967216)(149,28331639)(151,28405444)(153,28518150)(155,28566769)(158,28672653)(160,30194002)(164,30363983)(166,50181309)(166,92192939)};

% m189 all symb 3
\addplot[color=green!60!black]
coordinates{
(1,138)(2,449)(3,4852)(4,76463)(5,82207)(6,353174)(8,362239)(9,371382)(10,386839)(11,393761)(21,602550)(25,647737)(29,717397)(32,798692)(33,823520)(36,877196)(40,919638)(42,941488)(43,992585)(49,1129215)(61,1395276)(62,1410946)(66,1489954)(67,1508785)(68,1528864)(69,1544555)(73,1622313)(74,1652336)(76,1701090)(77,1730927)(78,2079922)(79,2354246)(80,5010128)(84,6911592)(85,7180144)(95,10393237)(96,10411298)(105,10608998)(106,11879270)(110,11955449)(111,11983194)(113,12091348)(114,12298327)(115,14602046)(121,19685574)(122,19752197)(125,19847153)(130,22834460)(138,23094615)(139,23126564)(140,23278840)(142,23348313)(144,23851506)(154,29301598)(159,29380125)(160,29400804)(162,29792878)(166,29962477)(168,33191698)(168,92285185)};

% m189 all symb 2, expert 0
\addplot[only marks,color=green!60!black,mark=o,mark size=0.5pt,mark options={fill=green!60!black},]
coordinates{
(1,138)(2,453)(3,1659)(4,8488)(5,16232)(6,85478)(7,262674)(9,274235)(16,519545)(23,1569704)(36,1740349)(85,3056351)};

% m189 all symb 2, expert 1
\addplot[only marks,color=green!60!black,mark=triangle,mark size=0.5pt,mark options={fill=green!60!black},]
coordinates{
(19,1537881)(43,1856382)(68,2348591)(115,5026355)};

% m189 all symb 2, expert 2
\addplot[only marks,color=green!60!black,mark=square,mark size=0.5pt,mark options={fill=green!60!black},]
coordinates{
(10,284158)(12,488930)(17,527481)(30,1655846)(35,1725237)(38,1771164)(39,1791447)(49,1934204)(52,1978825)(53,2017032)(56,2070742)(63,2231124)(69,2387459)(70,2442200)(71,2471569)(74,2691110)(75,2733764)(86,3267392)(90,3368579)(91,3402142)(92,3445895)(96,3540745)(97,3558840)(98,3643740)(101,3673769)(102,3765206)(103,4120296)(106,4609542)(109,4772021)(113,4918161)(116,5108894)(121,6483119)(122,7018504)(124,10563297)(134,27412604)(135,27439824)(147,27967216)(149,28331639)(153,28518150)(155,28566769)(158,28672653)(164,30363983)(166,50181309)(166,92192939)};

% m189 all symb 2, expert 3
\addplot[only marks,color=green!60!black,mark=diamond,mark size=0.5pt,mark options={fill=green!60!black},]
coordinates{
(34,1718923)(47,1905489)(80,2924375)(151,28405444)(160,30194002)};

% m189 all symb 3, expert 0
\addplot[only marks,color=green!60!black,mark=o,mark size=0.5pt,mark options={fill=green!60!black},]
coordinates{
(1,138)(2,449)(3,4852)(4,76463)(5,82207)(6,353174)(11,393761)(25,647737)(121,19685574)(159,29380125)};

% m189 all symb 3, expert 1
\addplot[only marks,color=green!60!black,mark=triangle,mark size=0.5pt,mark options={fill=green!60!black},]
coordinates{
(8,362239)(9,371382)(29,717397)(49,1129215)(66,1489954)(78,2079922)(96,10411298)(154,29301598)};

% m189 all symb 3, expert 2
\addplot[only marks,color=green!60!black,mark=square,mark size=0.5pt,mark options={fill=green!60!black},]
coordinates{
(10,386839)(21,602550)(32,798692)(33,823520)(40,919638)(42,941488)(61,1395276)(62,1410946)(67,1508785)(68,1528864)(74,1652336)(76,1701090)(77,1730927)(79,2354246)(80,5010128)(84,6911592)(85,7180144)(95,10393237)(105,10608998)(106,11879270)(110,11955449)(111,11983194)(113,12091348)(114,12298327)(115,14602046)(122,19752197)(125,19847153)(130,22834460)(138,23094615)(139,23126564)(140,23278840)(142,23348313)(144,23851506)(160,29400804)(162,29792878)(166,29962477)(168,33191698)(168,92285185)};

% m189 all symb 3, expert 3
\addplot[only marks,color=green!60!black,mark=diamond,mark size=0.5pt,mark options={fill=green!60!black},]
coordinates{
(36,877196)(43,992585)(69,1544555)(73,1622313)};

\end{axis}
\end{tikzpicture}
    \end{subfigure}
    \caption{Results Experiment 4 for \emph{m159} (left) and \emph{m189} (right) with total number of symbols.}
    \label{fig:res4_sym}
    \bigbreak
    \begin{subfigure}[b]{0.45\textwidth}
        \begin{tikzpicture}
\begin{axis}[
legend style={nodes={scale=0.5, transform shape}},width=9.0cm,height=6.0cm,ymode=log,ytick pos=left,xtick pos=bottom,xlabel={Hypothesis size},ylabel={Total resets},xmin=1, xmax=30,ymin=100, ymax=100000000,xtick={0,10,20,30},ytick={100, 10000, 1000000, 100000000},legend pos=south east,scale=0.6,vasymptote=5,]

% m159 full symb 1
\addplot[color=cyan,mark=o,mark size=0.5pt,]
coordinates{
(1,120)(2,373)(3,11850)(5,29729)(7,35686)(10,77947)(13,144298)(17,307396)(23,313166)(27,359146)(27,13652568)};
\addlegendentry{baseline}

% m159 all symb 1
\addplot[color=green!60!black]
coordinates{
(1,120)(2,373)(3,11850)(5,29726)(7,31440)(11,169288)(14,176935)(17,191213)(23,196143)(27,204833)(30,3171021)};
\addlegendentry{\moe{*}}

% m159 all symb 1, expert 0
\addplot[only marks,color=green!60!black,mark=o,mark size=0.5pt,mark options={fill=green!60!black},]
coordinates{
(1,120)(2,373)(3,11850)(5,29726)(7,31440)(11,169288)};
\addlegendentry{E$_{\mathsf{T}}$}

% m159 all symb 1, expert 1
\addplot[only marks,color=green!60!black,mark=triangle,mark size=0.5pt,mark options={fill=green!60!black},]
coordinates{
(30,3171021)};
\addlegendentry{E$_{\mathsf{AI}}$}

% m159 all symb 1, expert 2
\addplot[only marks,color=green!60!black,mark=square,mark size=0.5pt,mark options={fill=green!60!black},]
coordinates{
(14,176935)(17,191213)(23,196143)(27,204833)};
\addlegendentry{E$_{\mathsf{F}}$}

% m159 all symb 1, expert 3
% \addplot[only marks,color=green!60!black,mark=diamond,mark size=0.5pt,mark options={fill=green!60!black},]
% coordinates{
% };
% \addlegendentry{E$_{\mathsf{C}}$}

% m159 full symb 2
\addplot[color=cyan,mark=o,mark size=0.5pt,]
coordinates{
(1,120)(2,381)(3,20036)(4,25718)(5,87171)(7,99021)(11,143257)(12,158990)(17,165072)(27,174600)(30,177460)};

% m159 full symb 3
\addplot[color=cyan,mark=o,mark size=0.5pt,]
coordinates{
(1,120)(2,378)(3,8868)(4,18925)(5,22206)(7,24430)(10,29231)(13,32083)(14,142678)(23,175235)(23,15043566)};

% m159 all symb 2
\addplot[color=green!60!black]
coordinates{
(1,120)(2,381)(3,20036)(4,25718)(5,87171)(7,92559)(11,107479)(13,112960)(17,121973)(20,125298)(26,130033)(30,141905)};

% m159 all symb 3
\addplot[color=green!60!black]
coordinates{
(1,120)(2,378)(3,8868)(4,18925)(5,22206)(7,27488)(11,109640)(13,111079)(17,132662)(20,135013)(26,140035)(30,149710)};

% m159 all symb 2, expert 0
\addplot[only marks,color=green!60!black,mark=o,mark size=0.5pt,mark options={fill=green!60!black},]
coordinates{
(1,120)(2,381)(3,20036)(4,25718)(5,87171)(11,107479)};

% m159 all symb 2, expert 1
\addplot[only marks,color=green!60!black,mark=triangle,mark size=0.5pt,mark options={fill=green!60!black},]
coordinates{
(7,92559)(20,125298)};

% m159 all symb 2, expert 2
\addplot[only marks,color=green!60!black,mark=square,mark size=0.5pt,mark options={fill=green!60!black},]
coordinates{
(13,112960)(17,121973)(26,130033)(30,141905)};

% m159 all symb 2, expert 3
\addplot[only marks,color=green!60!black,mark=diamond,mark size=0.5pt,mark options={fill=green!60!black},]
coordinates{
};

% m159 all symb 3, expert 0
\addplot[only marks,color=green!60!black,mark=o,mark size=0.5pt,mark options={fill=green!60!black},]
coordinates{
(1,120)(2,378)(3,8868)(4,18925)(5,22206)(7,27488)(11,109640)};

% m159 all symb 3, expert 1
\addplot[only marks,color=green!60!black,mark=triangle,mark size=0.5pt,mark options={fill=green!60!black},]
coordinates{
(13,111079)};

% m159 all symb 3, expert 2
\addplot[only marks,color=green!60!black,mark=square,mark size=0.5pt,mark options={fill=green!60!black},]
coordinates{
(17,132662)(20,135013)(26,140035)(30,149710)};

% m159 all symb 3, expert 3
\addplot[only marks,color=green!60!black,mark=diamond,mark size=0.5pt,mark options={fill=green!60!black},]
coordinates{
};

\end{axis}
\end{tikzpicture}
    \end{subfigure}
    \begin{subfigure}[b]{0.45\textwidth}
        \begin{tikzpicture}
\begin{axis}[
legend style={nodes={scale=0.5, transform shape}},width=9.0cm,height=6.0cm,ymode=log,ytick pos=left,xtick pos=bottom,xlabel={Hypothesis size},ylabel={Total resets},xmin=1, xmax=189,ymin=100, ymax=100000000,xtick={0,50,100,150,189},ytick={100, 10000, 1000000, 100000000},legend pos=south east,scale=0.6,vasymptote=5,]

% m189 full symb 1
\addplot[color=cyan,mark=o,mark size=0.5pt,]
coordinates{
(1,138)(2,442)(3,2089)(4,11979)(5,13463)(6,27559)(8,77489)(9,141530)(11,151356)(17,170406)(19,172984)(20,174428)(21,186504)(26,198953)(29,223136)(30,320394)(36,406009)(37,534690)(39,581355)(40,587231)(42,591535)(44,675972)(45,689189)(46,724917)(49,811993)(53,1062142)(54,1210729)(55,1489965)(56,1627972)(57,1640758)(63,1651236)(64,1660933)(69,1669629)(71,2147045)(74,2400533)(77,2419079)(81,2561611)(82,3297530)(92,3549062)(93,3590769)(95,3667264)(102,3709670)(105,3835482)(106,3847588)(107,3871943)(109,4394423)(110,4397440)(118,4717936)(119,5280303)(120,7779531)(122,8373639)(123,8499778)(123,9043944)};
\addlegendentry{baseline}

% m189 all symb 1
\addplot[color=green!60!black]
coordinates{
(1,138)(2,442)(3,2089)(4,11979)(5,13463)(6,46717)(8,48468)(9,50235)(10,121999)(12,189994)(19,197546)(22,201202)(23,202316)(27,207054)(31,211501)(32,212840)(36,220878)(37,223038)(38,224366)(39,225712)(40,227006)(41,230548)(42,245070)(43,263628)(46,275532)(47,324397)(48,325697)(59,353207)(60,354447)(63,358140)(65,362864)(66,364227)(69,376496)(73,385215)(74,406478)(75,409985)(77,603043)(82,609745)(85,759055)(90,1021399)(99,1036849)(101,1046302)(102,1050362)(103,1054535)(104,1068410)(107,1153698)(114,1167037)(115,1175923)(117,1216565)(126,1244497)(129,1249782)(130,1329207)(131,1415680)(133,1419392)(135,3347649)(146,3370144)(148,3376288)(149,3408769)(151,3672679)(153,3710246)(157,4064716)(158,4068488)(160,4091010)(161,6738609)(165,6746765)(166,6748772)(166,7657663)};
\addlegendentry{\moe{*}}

% m189 all symb 1, expert 0
\addplot[only marks,color=green!60!black,mark=o,mark size=0.5pt,mark options={fill=green!60!black},]
coordinates{
(1,138)(2,442)(3,2089)(4,11979)(5,13463)(6,46717)(8,48468)(12,189994)(19,197546)(82,609745)(99,1036849)(146,3370144)};
\addlegendentry{E$_{\mathsf{T}}$}

% m189 all symb 1, expert 1
\addplot[only marks,color=green!60!black,mark=triangle,mark size=0.5pt,mark options={fill=green!60!black},]
coordinates{
(38,224366)(48,325697)(104,1068410)(135,3347649)};
\addlegendentry{E$_{\mathsf{AI}}$}

% m189 all symb 1, expert 2
\addplot[only marks,color=green!60!black,mark=square,mark size=0.5pt,mark options={fill=green!60!black},]
coordinates{
(9,50235)(10,121999)(22,201202)(23,202316)(27,207054)(31,211501)(32,212840)(36,220878)(37,223038)(39,225712)(40,227006)(41,230548)(42,245070)(43,263628)(46,275532)(47,324397)(59,353207)(60,354447)(63,358140)(65,362864)(66,364227)(73,385215)(74,406478)(75,409985)(77,603043)(85,759055)(90,1021399)(101,1046302)(102,1050362)(103,1054535)(107,1153698)(114,1167037)(115,1175923)(117,1216565)(126,1244497)(130,1329207)(131,1415680)(133,1419392)(148,3376288)(149,3408769)(151,3672679)(153,3710246)(157,4064716)(158,4068488)(160,4091010)(161,6738609)(165,6746765)(166,6748772)(166,7657663)};
\addlegendentry{E$_{\mathsf{F}}$}

% m189 all symb 1, expert 3
\addplot[only marks,color=green!60!black,mark=diamond,mark size=0.5pt,mark options={fill=green!60!black},]
coordinates{
(69,376496)(129,1249782)};
\addlegendentry{E$_{\mathsf{C}}$}

% m189 full symb 2
\addplot[color=cyan,mark=o,mark size=0.5pt,]
coordinates{
(1,138)(2,451)(3,1330)(4,5746)(5,9411)(6,10810)(7,26322)(9,52391)(10,76085)(12,101908)(15,120832)(20,126752)(22,130334)(24,248465)(28,253163)(34,271354)(35,273863)(36,275108)(40,294628)(41,309715)(46,390771)(50,660884)(55,720574)(56,758372)(57,790187)(61,1025681)(62,1071505)(63,1143134)(64,1201629)(65,1681598)(68,2012173)(78,2027390)(79,2040533)(82,2053161)(84,2058823)(85,2074428)(86,2222991)(93,2247511)(96,2252328)(97,2359791)(99,2466131)(100,2482351)(101,3096879)(102,3404818)(104,3678578)(111,3725578)(113,3975993)(114,4147935)(115,4519555)(119,4683744)(120,4689613)(122,4772113)(123,5190492)(124,5372904)(124,9042609)};

% m189 full symb 3
\addplot[color=cyan,mark=o,mark size=0.5pt,]
coordinates{
(1,138)(2,448)(3,4515)(4,55219)(5,57735)(6,95066)(8,127102)(14,140710)(15,142103)(17,145380)(18,146625)(19,159919)(20,172956)(23,235639)(26,248670)(30,254626)(32,343066)(34,396502)(36,422267)(37,430308)(38,548037)(41,693693)(47,750997)(48,1088018)(52,1113702)(53,1150615)(55,1165670)(56,1214368)(59,2652069)(61,2655888)(62,2773422)(63,3246967)(76,4257140)(77,4319117)(78,4321040)(79,4322937)(83,4389587)(84,4391631)(85,4393713)(92,4410306)(95,4415965)(98,4512246)(99,4636776)(100,4830468)(107,4849430)(108,4921472)(110,4954327)(112,4958706)(114,4980208)(116,5378792)(117,7797222)(118,7806025)(119,8925668)(119,8988696)};

% m189 all symb 2
\addplot[color=green!60!black]
coordinates{
(1,138)(2,451)(3,1330)(4,5746)(5,9411)(6,36246)(7,101629)(9,104205)(10,106397)(12,159484)(16,164230)(17,165416)(19,371114)(23,375828)(30,384998)(34,390122)(35,391378)(36,392641)(38,395094)(39,396964)(43,403287)(47,408317)(49,410852)(52,414731)(53,419221)(56,423180)(63,433379)(68,439748)(69,442636)(70,447595)(71,449366)(74,471356)(75,474932)(80,491269)(85,498206)(86,516972)(90,523057)(91,525066)(92,527912)(96,534127)(97,535638)(98,542899)(101,547515)(102,555947)(103,587780)(106,626937)(109,636007)(113,642517)(115,648410)(116,653898)(121,770387)(122,813561)(124,1122526)(134,2587071)(135,2589022)(147,2621153)(149,2651661)(151,2655475)(153,2661680)(155,2665564)(158,2671135)(160,2792903)(164,2800523)(166,4399015)(166,7807082)};

% m189 all symb 3
\addplot[color=green!60!black]
coordinates{
(1,138)(2,448)(3,4515)(4,55219)(5,57735)(6,164164)(8,165682)(9,167351)(10,171103)(11,172616)(21,214400)(25,219205)(29,227425)(32,235614)(33,238549)(36,245178)(40,250783)(42,253512)(43,260087)(49,274036)(61,291491)(62,292921)(66,298903)(67,300642)(68,302375)(69,303808)(73,309538)(74,311302)(76,314424)(77,316285)(78,349509)(79,376780)(80,641212)(84,830589)(85,856857)(95,1174611)(96,1176253)(105,1190418)(106,1306709)(110,1312809)(111,1314505)(113,1323401)(114,1341022)(115,1548927)(121,2000357)(122,2004417)(125,2009675)(130,2259640)(138,2274717)(139,2276738)(140,2288149)(142,2292513)(144,2331862)(154,2757961)(159,2768023)(160,2770083)(162,2800230)(166,2809322)(168,3067057)(168,7714815)};

% m189 all symb 2, expert 0
\addplot[only marks,color=green!60!black,mark=o,mark size=0.5pt,mark options={fill=green!60!black},]
coordinates{
(1,138)(2,451)(3,1330)(4,5746)(5,9411)(6,36246)(7,101629)(9,104205)(16,164230)(23,375828)(36,392641)(85,498206)};

% m189 all symb 2, expert 1
\addplot[only marks,color=green!60!black,mark=triangle,mark size=0.5pt,mark options={fill=green!60!black},]
coordinates{
(19,371114)(43,403287)(68,439748)(115,648410)};

% m189 all symb 2, expert 2
\addplot[only marks,color=green!60!black,mark=square,mark size=0.5pt,mark options={fill=green!60!black},]
coordinates{
(10,106397)(12,159484)(17,165416)(30,384998)(35,391378)(38,395094)(39,396964)(49,410852)(52,414731)(53,419221)(56,423180)(63,433379)(69,442636)(70,447595)(71,449366)(74,471356)(75,474932)(86,516972)(90,523057)(91,525066)(92,527912)(96,534127)(97,535638)(98,542899)(101,547515)(102,555947)(103,587780)(106,626937)(109,636007)(113,642517)(116,653898)(121,770387)(122,813561)(124,1122526)(134,2587071)(135,2589022)(147,2621153)(149,2651661)(153,2661680)(155,2665564)(158,2671135)(164,2800523)(166,4399015)(166,7807082)};

% m189 all symb 2, expert 3
\addplot[only marks,color=green!60!black,mark=diamond,mark size=0.5pt,mark options={fill=green!60!black},]
coordinates{
(34,390122)(47,408317)(80,491269)(151,2655475)(160,2792903)};

% m189 all symb 3, expert 0
\addplot[only marks,color=green!60!black,mark=o,mark size=0.5pt,mark options={fill=green!60!black},]
coordinates{
(1,138)(2,448)(3,4515)(4,55219)(5,57735)(6,164164)(11,172616)(25,219205)(121,2000357)(159,2768023)};

% m189 all symb 3, expert 1
\addplot[only marks,color=green!60!black,mark=triangle,mark size=0.5pt,mark options={fill=green!60!black},]
coordinates{
(8,165682)(9,167351)(29,227425)(49,274036)(66,298903)(78,349509)(96,1176253)(154,2757961)};

% m189 all symb 3, expert 2
\addplot[only marks,color=green!60!black,mark=square,mark size=0.5pt,mark options={fill=green!60!black},]
coordinates{
(10,171103)(21,214400)(32,235614)(33,238549)(40,250783)(42,253512)(61,291491)(62,292921)(67,300642)(68,302375)(74,311302)(76,314424)(77,316285)(79,376780)(80,641212)(84,830589)(85,856857)(95,1174611)(105,1190418)(106,1306709)(110,1312809)(111,1314505)(113,1323401)(114,1341022)(115,1548927)(122,2004417)(125,2009675)(130,2259640)(138,2274717)(139,2276738)(140,2288149)(142,2292513)(144,2331862)(160,2770083)(162,2800230)(166,2809322)(168,3067057)(168,7714815)};

% m189 all symb 3, expert 3
\addplot[only marks,color=green!60!black,mark=diamond,mark size=0.5pt,mark options={fill=green!60!black},]
coordinates{
(36,245178)(43,260087)(69,303808)(73,309538)};

\end{axis}
\end{tikzpicture}
    \end{subfigure}
    \caption{Results Experiment 4 for \emph{m159} (left) and \emph{m189} (right) with total number of resets.}
    \label{fig:res4_resets}
\end{figure}

}
    
\end{document}